%% file: pzk.tex
\documentclass[11pt]{amsart}

\input{ams_header.tex}

\usepackage[top=1.5in,bottom=1.5in,left=1.5in,right=1.5in]{geometry}
\usepackage{braket}
\usepackage{mathrsfs}

\newcommand{\ang}[1]{\langle #1 \rangle}
\DeclareMathOperator{\MIP}{MIP}
\DeclareMathOperator{\SynMIP}{SynMIP}
\DeclareMathOperator{\AM}{AM}
\DeclareMathOperator{\NP}{NP}
\DeclareMathOperator{\poly}{poly}
\DeclareMathOperator{\polylog}{polylog}
\DeclareMathOperator{\RE}{RE}
\DeclareMathOperator{\coRE}{coRE}
\DeclareMathOperator{\true}{true}
\DeclareMathOperator{\false}{false}
\DeclareMathOperator{\SynAlg}{SynAlg}
\DeclareMathOperator{\supp}{supp}
\DeclareMathOperator{\df}{def}

\DeclareMathOperator{\IP}{IP}
\DeclareMathOperator{\PSPACE}{PSPACE}
\DeclareMathOperator{\NEXP}{NEXP}
\DeclareMathOperator{\PZK}{PZK}
\DeclareMathOperator{\QMIP}{QMIP}
\DeclareMathOperator{\HALT}{HALT}
\DeclareMathOperator{\BCS}{BCS}
\DeclareMathOperator{\Tab}{Tab}
\DeclareMathOperator{\Obl}{Obl}
\newcommand{\AND}{\wedge}
\newcommand{\OR}{\vee}
\newcommand{\ul}{\underline}
\usepackage{bbm}

\title[Two prover perfect zero knowledge]{Two prover perfect zero knowledge for MIP*}

\author[K. Mastel]{Kieran Mastel$^{1,2}$}
\author[W. Slofstra]{William Slofstra$^{1,2}$}

\address[1]{Institute for Quantum Computing, University of Waterloo, Canada}
\address[2]{Department of Pure Mathematics, University of Waterloo, Canada}

\email{william.slofstra@uwaterloo.ca}
\email{kmastel@uwaterloo.ca}

\begin{document}

\begin{abstract}
	The recent $\MIP^*=\RE$ theorem of Ji, Natarajan, Vidick, Wright, and Yuen shows
	that the complexity class $\MIP^*$ of multiprover proof systems with entangled
	provers contains all recursively enumerable languages.  Prior work of Grilo,
	Slofstra, and Yuen [FOCS '19] further shows (via a technique called simulatable
	codes) that every language in $\MIP^*$ has a perfect zero knowledge ($\PZK$) $\MIP^*$
	protocol.  The $\MIP^*=\RE$ theorem uses two-prover one-round proof systems, and
	hence such systems are complete for $\MIP^*$. However, the construction in Grilo,
	Slofstra, and Yuen uses six provers, and there is no obvious way to get perfect
	zero knowledge with two provers via simulatable codes. This leads to a natural
	question: are there two-prover $\PZK$-$\MIP^*$ protocols for all of $\MIP^*$? 
	
	In this paper, we show that every language in $\MIP^*$ has a two-prover one-round
	$\PZK$-$\MIP^*$ protocol, answering the question in the affirmative. For the proof, we
	use a new method based on a key consequence of the $\MIP^*=\RE$ theorem, which is
	that every $\MIP^*$ protocol can be turned into a family of boolean constraint
	system (BCS) nonlocal games. This makes it possible to work with $\MIP^*$ protocols
	as boolean constraint systems, and in particular allows us to use a variant of
	a construction due to Dwork, Feige, Kilian, Naor, and Safra [Crypto '92] which
	gives a classical $\MIP$ protocol for 3SAT with perfect zero knowledge. To show
	quantum soundness of this classical construction, we develop a toolkit for
	analyzing quantum soundness of reductions between BCS games, which we expect to
	be useful more broadly. This toolkit also applies to commuting operator
	strategies, and our argument shows that every language with a commuting
	operator BCS protocol has a two prover $\PZK$ commuting operator protocol.
\end{abstract}

\maketitle

\section{Introduction}

In an interactive proof protocol, a prover tries to convince a verifier that a
string $x$ belongs to $\mcL$. Interactive proof systems can be more powerful
than non-interactive systems; famously, the class $\IP$ of interactive proofs
with a polynomial time verifier and a single prover is equal to $\PSPACE$
\cite{Shamir92pspace}, and the class $\MIP$ with a polynomial time verifier and
multiple provers is equal to $\NEXP$ \cite{BFL90}.  In this latter class, the
provers can communicate with the verifier, but are assumed not to be able to
communicate with each other. The proof systems used in \cite{BFL90} are very
efficient, and require only two provers and one-round of communication.
Interactive proof systems also allow zero knowledge protocols, in which the
prover demonstrates that $x \in \mcL$ without revealing any other information
to the verifier. As a result, interactive proof systems are important to both
complexity theory and cryptography. The first zero knowledge proof systems go
back to the invention of interactive proof systems by Goldwasser, Micali, and
Rackoff \cite{Goldwasser85}, and every language in MIP admits a two-prover
one-round perfect zero knowledge proof system by a result of Ben-Or,
Goldwasser, Kilian, and Wigderson \cite{Ben-Or88}. Perfect means that
absolutely no information is revealed to the verifier, in contrast to
statistical zero knowledge (in which the amount of knowledge gained by the
verifier is small but bounded), or computational zero knowledge (in which zero
knowledge relies on some computational intractability assumption). 

Since the provers in a MIP protocol are not allowed to communicate, it is
natural to ask what happens if they are allowed to share entanglement. This
leads to the complexity class $\MIP^*$, first introduced by Cleve, Hoyer,
Toner, and Watrous \cite{cleve2010consequences}. Entanglement allows the
provers to break some classical proof systems by coordinating their answers,
but the improved ability of the provers also allows the verifier to set harder
tasks. As a result, figuring out the power of $\MIP^*$ has been difficult, and
there have been successive lower bounds in \cite{KKM+11, IKM09, IV12, Vid16,
Vid20eratum, Ji16, NV18b, Ji17, NV18a, FJVY19}. Most recently (and
spectacularly), Ji, Natarajan, Vidick, Wright, and Yuen showed that
$\MIP^*=\RE$, the class of languages equivalent to the halting problem
\cite{ji2022mipre}. 
 Reichardt, Unger, and Vazirani also showed that $\MIP^*$
is equal to the class $\QMIP^*$, in which the verifier is quantum, and can
communicate with the provers via quantum channels \cite{ReichardtLeash13}. On
the perfect zero knowledge front, Chiesa, Forbes, Gur, and Spooner showed that
every language in $\NEXP$ (and hence in classical $\MIP$) has a perfect zero
knowledge $\MIP^*$ proof system, or in other words belongs to $\PZK$-$\MIP^*$
\cite{chiesa2018spatial}. Grilo, Slofstra, and Yuen show that all of $\MIP^*$
belongs to $\PZK$-$\MIP^*$ \cite{grilo2019perfect}. 

Combining $\PZK$-$\MIP^*=\MIP^*$ with $\MIP^*=\RE$ shows that there are one-round
perfect zero-knowledge $\MIP^*$ proof systems for all languages that can be
reduced to the halting problem, a very large class. However, the construction
in \cite{grilo2019perfect} is involved. The idea behind the proof is to encode a circuit
for an arbitrary $\MIP$ verifier in a ``simulatable'' quantum error correcting code, and then
hide information from the verifier by splitting the physical qubits of this
code between different provers. The resulting proof systems in \cite{grilo2019perfect}
require $6$ provers, and because the core concept of the proof is to split
information between provers, bringing this down to $2$ provers (as can be done
with perfect zero-knowledge for $\MIP$) seems to require new ideas. 

The purpose of this paper is to show that all languages in $\MIP^*$ do indeed have two-prover one-round
perfect zero knowledge proof systems. Specifically, we show that: 
\begin{theorem}\label{thm:main}
    Every language in $\MIP^*$ (and hence in $\RE$) admits a two-prover one-round
    perfect zero knowledge $\MIP^*$ protocol with completeness probability $c=1$ 
    and soundness probability $s=1/2$, in which the verifier chooses questions 
    uniformly at random. 
\end{theorem}
The idea behind the proof is to use the output of the $\MIP^*=\RE$ theorem, rather
than encoding arbitrary $\MIP^*$-protocols. The proof that $\MIP^*=\RE$ in
\cite{ji2022mipre} is very difficult, but requires only two-prover one-round proof
systems.  Natarajan and Zhang have sharpened the proof to show that these proof
systems require only a constant number of questions, and $\polylog$ length
answers from the provers \cite{natarajan2023quantum}. This shows that $\MIP^* = \AM^*(2)$, the
complexity class of languages with two-prover $\MIP^*$-protocols in which the verifier
chooses their messages to the prover uniformly at random. A one-round $\MIP$ or
$\MIP^*$ proof system is equivalent to a family of nonlocal games, in which the
provers (now also called players) are given questions and return answers to a
verifier (now also called a referee), who decides whether to accept (in which
case the players are said to win) or reject (the players lose). In both
\cite{ji2022mipre} and \cite{natarajan2023quantum}, the games are synchronous,
meaning that if the players receive the same question then they must reply with
the same answer, and admit what are called oracularizable strategies. As we
observe in this paper, one-round $\MIP^*$ proof systems in which the games are
synchronous and oracularizable are equivalent to the class of $\BCS$-$\MIP^*$
proof systems, which are one-round two-prover proof systems in which the
nonlocal games are boolean constraint system (BCS) games. In a boolean
constraint system, two provers try to convince the verifier that a given BCS is
satisfiable. BCS games were introduced by Cleve and Mittal
\cite{cleve2013characterization}, and include famous examples of nonlocal games
such as the Mermin-Peres magic square \cite{mermin90simple,PERES1990107}.
Boolean constraint systems are much easier to work with than general $\MIP^*$
protocols, so rather than showing that every $\MIP^*$ protocol can be
transformed to a perfect zero knowledge protocol, we prove \Cref{thm:main} by
showing that every $\BCS$-$\MIP^*$ protocol can be transformed to a perfect
zero knowledge protocol. As we explain at the end of \Cref{sec:nonlocal}, when
combined with the $\MIP^*=\RE$ theorem this gives an effective way to transform
any $\MIP^*$-protocol (including protocols with many provers and rounds) into a
perfect zero knowledge $\BCS$-$\MIP^*$ protocol.

One way to transform a $\BCS$-$\MIP^*$ protocol to a perfect zero-knowledge
protocol is to use graph colouring games, which are famous examples of perfect
zero knowledge games. Classically, every BCS instance can be transformed to
a graph such that the graph is $3$-colourable if and only if the BCS is
satisfiable. Ji has shown that every BCS can be transformed to a graph such
that the original BCS game has a perfect quantum strategy if and only if the
$3$-colouring game for the graph has a perfect quantum strategy
\cite{ji2013binary} (see also \cite{harris2023universality}). Using the
techniques in this paper, it is also possible to show that this transformation
preserves soundness of $\BCS$-$\MIP^*$ protocols, and hence that every
$\BCS$-$\MIP^*$ protocol can be transformed to a $\MIP^*$ protocol based on
graph colouring games. Unfortunately graph colouring games are only perfect zero
knowledge against honest verifiers, so this construction does not give a
perfect zero knowledge protocol for dishonest verifiers. Instead, we use another
classical transformation due to Dwork, Feige, Kilian, Naor, and Safra
\cite{Dwork1992LowC2}, which takes every 3SAT instance to a perfect
zero-knowledge $\MIP$ protocol. We show that a modest variant of this
construction remains perfect zero knowledge in the quantum setting, and
preserves soundness of $\BCS$-$\MIP^*$ protocols. In both the original argument
and our argument, it is necessary for soundness to work with $\BCS$-$\MIP$
protocols with small (meaning $\log$ or  $\polylog$) question length. In the
classical setting, $\BCS$-$\MIP$ with $\log$ question length is equal to $\NP$,
so the construction in \cite{Dwork1992LowC2} only shows that $\NP$ is contained in
$\PZK$-$\MIP$, rather than all of $\NEXP$. In the quantum setting,
$\BCS$-$\MIP^*$ with $\polylog$ question length is equal to $\MIP^*$ and 
this construction suffices to prove perfect zero knowledge for any
$\MIP^*$ protocol --- an interesting difference in what techniques can be used
between the classical and quantum setting.

In general, it's a difficult question to figure out if a classical
transformation of constraint systems (of which there are many) remains
sound (meaning that it preserves soundness of protocols) in the quantum
setting. For instance, one of the key parts of the $\MIP^*=\RE$ theorem is the
construction of PCP of proximity which is quantum sound. On the other hand,
there are some transformations which lift fairly easily to the quantum setting.
We identify two such classes of transformations, ``classical transformations''
which are applied constraint by constraint, and ``context subdivision
transformations'', in which each constraint is split into a number of
subclauses. Both types of transformations are used implicitly throughout the
literature on nonlocal games, including in \cite{ji2013binary}, which was the
first paper to consider reductions between quantum strategies in BCS games.
In this paper, we systematically investigate the quantum soundness of these
transformations.  It's relatively easy to show that classical transformations
preserve soundness, and this is shown in \Cref{sec:weightedhom}.  In
subdivision, each subclause becomes a different question in the associated BCS
game, and thus a strategy for the subdivided game has many more observables
than the original game. Since these new observables don't need to commute with
each other, subdivision is more difficult to work with. Nonetheless, we show
that if the subclauses have a bounded number of variables, then subdivision
preserves soundness with a polynomial dropoff.  This is shown in
\Cref{sec:stability}. The construction in \cite{Dwork1992LowC2} can be
described as a composition of classical transformations and context subdivision
transformations, so quantum soundness (with polynomial dropoff) of this
construction follows from combining the soundness of these two transformations.
We recover a constant soundness gap by using parallel repetition, which preserves
the class of BCS games. 

While reductions between nonlocal games have been important in previous work,
they are difficult to reason about, since it's necessary to keep track of how
strategies for one game map to strategies for the other game. One advantage of
working with constraint systems in the classical setting is that it's more
convenient to work with assignments (and think about the fraction of
constraints in the system that can be satisfied) than it is to work with
strategies and winning probabilities. In the quantum setting, it isn't
possible to work with assignments, because strategies involve observables
that don't necessarily commute with each other. However, we can achieve 
a similar conceptual simplification
by replacing assignments with representations of the BCS algebra of the
constraint system. This algebra is the same as the synchronous algebra of the
BCS game introduced in \cite{helton2017algebras,Kim_2018}; we refer to
\cite{PS23} for more background.  With this approach, reductions between BCS
games can be expressed as homomorphisms between BCS algebras, and these are
much easier to describe and work with than mappings between strategies. For
soundness arguments, we need to work with near-perfect strategies, and these
correspond to approximate representations of the BCS algebra \cite{Pad22}.
Previous work using this idea (see e.g. \cite{Pad22,
harris2023universality}) has focused on reductions between single games, and
the definitions are not suitable for working with protocols, as they do not
incorporate question distributions. To solve this problem, we introduce a
notion of weighted algebras and weighted homomorphisms, which allows us to
keep track of soundness of reductions between games using completely algebraic
arguments involving sums of squares. 

Another advantage of the weighted algebras framework is that arguments can be made
simultaneously for both quantum and commuting operator strategies. Our
proof methods extend to commuting operator strategies as a result.  However,
our results here are not as conclusive, as the exact characterization of the
corresponding complexity class $\MIP^{co}$ is not known. There is a conjecture
that $\MIP^{co} = \coRE$, and with that conjecture and a parallel repetition
theorem for commuting operator strategies, we expect that it would be possible
to extend \Cref{thm:main} to show that all languages in $\MIP^{co}$ have a
perfect zero knowledge commuting operator protocol. Without these ingredients,
we are limited to showing that $\BCS$-$\MIP^{co} = \PZK$-$\BCS$-$\MIP^{co}$. 
Previous work on perfect zero knowledge for commuting operator protocols
does not preserve soundness gaps \cite{coudron2019complexity}.  

Our results also have applications for the membership problem for quantum
correlations. For exact membership, the cohalting problem is many-one reducible
to membership in the set of quantum-approximable correlations $C_{qa}$, and to membership in the
set of commuting operator correlations $C_{qc}$ 
\cite{Sl19,coudron2019complexity,FMS2021}. It follows from $\MIP^*=\RE$ that
the halting problem is Turing reducible to approximate membership in $C_q$, the
set of quantum correlations, but this is not a many-one reduction. The proof of
Theorem \ref{thm:main} immediately implies that there is a many-one reduction
from the halting problem to approximate membership in $C_q$.

Because we use parallel repetition to reduce an inverse-polynomial soundness
gap to a constant soundness gap, the protocols in \Cref{thm:main} use
polynomial length questions and answers. If an inverse-polynomial soundness gap is
allowed, we get perfect zero-knowledge protocols with $\polylog$ question
length and constant answer length. Whether it is possible to get
perfect zero-knowledge protocols with $\polylog$ question
length, constant answer length, and constant soundness gap is an interesting
open question.  This would be possible with an improved analysis or
construction for subdivision such as appears in the low degree test
\cite{JNVWY21} used in the $\MIP^* = \RE$ theorem.

\subsection*{Acknowledgements}
We thank Connor Paddock and Henry Yuen for helpful conversations.
KM is supported by NSERC.  WS is supported by NSERC DG 2018-03968 and an Alfred
P. Sloan Research Fellowship.

\section{Background on $*$-algebras}\label{sec:algebras}

We recall some of the key concepts in the theory of $*$-algebras. See \cite{ozawa2013connes, Schmdgen2020} for a more complete background. A complex $*$-algebra $\mcA$ is a unital algebra over $\mbC$ with an antilinear involution $a \mapsto a^*$, such that and $(ab)^*=b^*a^*$. We let $\mbC^*\langle X\rangle$ denote the free complex $*$-algebra generated by the set $X$. If $R \subseteq \mbC^*\langle X\rangle$, we let $\mbC^*\langle X: R\rangle$  denote the quotient of $\mbC^*\langle X\rangle$ by the two-sided ideal generated by $R$. If $X$ and $R$ are finite then we call $\mbC^*\langle X:R\rangle$ a \textbf{finitely presented} $*$-algebra.

A $*$\textbf{-homomorphism} $\phi:\mcA \to \mcB$ between $*$-algebras is an algebra homomorphism such that $\phi(x^*) = \phi(x)^*$ for all $x \in A$. A $*$\textbf{-representation} of $\mcA$ is a $*$-homomorphism $\rho: \mcA \to \mcB(\mcH)$ from $\mcA$ to the $*$-algebra of bounded operators on the Hilbert space $\mcH$. If $\mcA$ and $\mcB$ are $*$-algebras, and $\C^*\langle X:R\rangle$ is a presentation of $A$, then $*$-homomorphisms $\mcA\to \mcB$ correspond to homomorphisms $\phi: \mcC\langle X\rangle \to \mcB$ such that $\phi(r) = 0$ for all $r\in R$. Thus, a $*$-representation is an assignment of operators to the elements of $X$ that satisfies the defining relations $R$.

If $\mcA$ is a $*$-algebra, then $a \geq b$ if
$a-b$ is a sum of hermitian squares, i.e. there is $k \geq 0$ and
$c_1,\ldots,c_k \in \mcA$ such that $a-b = \sum_{i=1}^k c_i^* c_i$. A finitely presented $*$-algebra $\mcA$ is called \textbf{archimedean} if for all $a\in \mcA$ there exists a $\lambda>0$ such that $a^*a\leq \lambda1$. The algebras we consider in this work are all archimedean. If $f:\mcA\to \mbC$ is a linear functional then $f$ is \textbf{positive} if $f(a)\geq0$ whenever $a\geq0$. A \textbf{state} on $\mcA$ is a positive linear functional $\tau:\mcA \to \mbC$ with $\tau(a^*a)\geq 0$ for all $a \in \mcA$, $\tau(1) = 1$ and $\tau(a^*) = \overline{\tau(a)}$ for all $a \in \mcA$. A state is \textbf{tracial} if $\tau(ab) = \tau(ba)$ for all $a,b \in \mcA$, and \textbf{faithful} if $\tau(a^*a)>0$ for all $a\neq 0$. A tracial state $\tau$ induces the \textbf{trace norm} $\|a\|_{\tau} := \sqrt{\tau(a^* a)}$, also called the $\tau$-norm. Trace norms are unitarily invariant, meaning that $\|uav\|_{\tau} = \|a\|_{\tau}$ for all $a \in \mcA$, and all unitaries $u$ and $v$. An element $u \in \mcA$ is called \textbf{unitary} if $u^*u = 1 = uu^*$.

If $\rho: \mcA \to \mcB(\mcH)$ is a $*$-algebra representation, then a vector $|v\rangle \in \mcH$ is \textbf{cyclic} for $\rho$ if the closure of $\rho(\mcA)|v\rangle$ with respect to the Hilbert space norm is equal to $\mcH$. A \textbf{cyclic representation} of $\mcA$ is a tuple $(\rho,\mcH,|v\rangle)$, where $\rho$ is a representation of $\mcA$ on $\mcH$ and $|v\rangle$ is a cyclic vector for $\rho$. If $\tau:\mcA \to \mbC$ is a positive linear functional on $\mcA$, then there is a cyclic representation $\rho_{\tau}$ of $\mcA$, called the \textbf{GNS representation} of $\tau$, such that $\tau(a) = \bra{\xi_{\tau}}\rho_{\tau}\ket{\xi_{\tau}}$ for all $a \in \mcA$. Two representations $\rho: \mcA \to \mcB(\mcH)$ and $\pi: \mcA \to \mcB(\mcK)$ of $\mcA$ are \textbf{unitarily equivalent} if there is a unitary operator $U: \mcH\to \mcK$ such that $U\rho(a)U^* = \pi(a)$ for all $a \in \mcA$. If $\tau$ is the state defined by $\tau(a) = \bra{\xi}\rho(a)\ket{\xi}$ for all $a \in \mcA$ and some cyclic representation $(\rho,\mcH,\ket{\xi})$, then $(\rho,\mcH,\ket{\xi})$ is unitarily equivalent to the GNS representation. A state $\tau$ is \textbf{finite-dimensional} if the Hilbert space $\mcH_{\tau}$ in the GNS representation $(\rho_{\tau},\mcH_{\tau},|\xi_{\tau}\rangle)$ is finite-dimensional. A state $\tau$ on $\mcA$ is called \textbf{Connes-embeddable} if there is a trace-preserving embedding of $\mcA$ into the ultrapower of the hyperfinite $\text{II}_1$ factor.

If $\mcA$ is a $*$-algebra then two
elements $a,b \in \mcA$ are said to be \textbf{cyclically equivalent} if there
is $k \geq 0$ and $f_1,\ldots,f_k,g_1,\ldots,g_k \in \mcA$ such that $a - b =
\sum_{i=1}^k [f_i, g_i]$, where $[f,g] = fg - gf$. We say that $a \gtrsim b$ if
$a-b$ is cyclically equivalent to a sum of squares.
If $\tau$ is a tracial state on $\mcA$ then $\tau(c_i^*c_i) \geq 0$ and $\tau([f_j,g_j]) = 0$. 
Thus if $a \gtrsim b$ then $\tau(a)\geq \tau(b)$, and if $a$ and $b$ are cyclically equivalent then $\tau(a-b) = 0$.

The $*$-algebras we use in this work are built out of the group algebras of the finitely presented groups
\begin{equation*}
	\Z_q^{*V} = \langle V: x^q = 1\rangle \text{ and } \Z_q^V = \langle V: x^q = 1, xy=yx \text{ for all } x,y \in V\rangle.
\end{equation*}
The group algebra $\C\Z^{*V}_q$ is the $*$-algebra generated by variables $x \in V$ with the defining relations from $\Z_q^{*V}$, along with the relations $x^*x = xx^* = 1$ for all $x \in V$. Similarly $\C \Z_q^{V}$ is the $*$-algebra generated by variables $x \in V$ with the defining relations of $\Z_q^V$, along with the relations $x^q = x^* x = x x^* = 1$ for all $x \in V$. Notice that $\C
\Z_q^{V}$ is the quotient of $\C \Z_q^{*V}$ by the relations $xy = yx$ for all
$x,y \in V$. If $\mcA$ and $\mcB$ are complex $*$-algebras, then we let $\mcA \ast \mcB$ denote their free product, and $\mcA \otimes \mcB$ denote their tensor product. Both are again complex $*$-algebras. 

When working with $\C\Z_q^V$, a \textbf{monomial} in $V$ is an element of the form $\prod_{x\in V}x^{a_x}$, where $0 \leq a_x < q$. 
We say that the monomial contains a variable $y\in V$ if $a_y>0$. The degree of a monomial is
$\sum_x a_x$. If $\mcA_1$ and $\mcA_2$ are $*$-algebras for which we have a defined notion of monomial, then a monomial in $\mcA_1\otimes\mcA_2$ is an element of the form $v_1v_2$, where $v_i$ is a monomial in $\mcA_i$. The degree of $v_1v_2$ is the sum of the degrees of $v_1$ and $v_2$, and a variable $y$ is contained in $v_1v_2$ if $y$ is contained in $v_1$ or $v_2$. For instance, a monomial in $\C\Z_{q_1}^{V_1}\otimes\C\Z_{q_2}^{V_2}$ is an element of the form $\prod_{x\in V_1}x^{a_x}\cdot\prod_{y\in
V_2}y^{b_y}$, where $0\leq a_x < q_1$ and $0\leq b_y<q_2$. Similarly, a monomial in $\mcA_1\ast\mcA_2$ is an element of the form $v_1\cdots v_k$, where $v_j$ is a monomial in $\mcA_{i_j}$ for all $1\leq j\leq k$, and $i_j\neq i_{j+1}$ for all $1\leq j < k$. In this case, the degree of $v_1\cdots v_k$ is the sum of the degrees of $v_1,\ldots,v_k$, and a variable $y$ is contained in $v_1\cdots v_k$ if $y$ is contained in one of the monomials $v_1,\ldots,v_k$. In any $*$-algebra where we have a defined notion of monomial, a polynomial is a linear combination of monomials.

A $C^*$\textbf{-algebra} $\mcA$ is a complex $*$-algebra with a submultiplicative Banach norm that satisfies the $C^*$ identity $\|aa^*\| = \|a\|^2$ for all $a\in \mcA$. Every $C^*$-algebra can be realized as a norm-closed $*$-subalgebra of the algebra of bounded operators $\mcB(\mcH)$ on some Hilbert space $\mcH$. A $C^*$-algebra is a von Neumann algebra if it can be realized as a $*$-subalgebra of $\mcB(\mcH)$ which is closed in the weak operator topology. More background on $C^*$-algebras and von Neumann algebras can be found in \cite{blackadar06}.

\section{Nonlocal games and MIP*}\label{sec:nonlocal}

A two-player \textbf{nonlocal} (or \textbf{Bell}) \textbf{scenario} consists of
a finite set of questions $I$, and a collection of finite answer sets $(O_i)_{i
\in I}$. Often in this definition there are separate question and answer sets
for each player, but it's convenient for us to assume that both players have
the same question and answer sets, and we don't lose any generality by assuming
this. We often think of the question and answer sets as being subsets of
$\{0,1\}^n$ and $\{0,1\}^{m_i}$, $i \in I$ respectively, in which case we say
that the questions have length $n$ and the answers have length $\max_{i \in I}
m_i$. A \textbf{nonlocal game} consists of a nonlocal scenario $(I,(O_i)_{i
\in I})$, along with a probability distribution $\pi$ on $I \times I$ and a
family of functions $V(\cdot,\cdot|i,j) : O_i \times O_j \to \{0,1\}$ for
$(i,j) \in I \times I$.  In the game, the players (commonly called Alice and
Bob) receive questions $i$ and $j$ from $I$ with probability $\pi(i,j)$, and
reply with answers $a \in O_i$ and $b \in O_j$ respectively. They win if
$V(a,b|i,j) = 1$, and lose otherwise.

A \textbf{correlation} for scenario $(I,\{O_i\}_{i \in I})$ is a family $p$ of
probability distributions $p(\cdot,\cdot|i,j)$ on $O_i \times O_j$ for all
$(i,j) \in I \times I$. Correlations are used to describe the players'
behaviour in a nonlocal scenario. The probability $p(a,b|i,j)$ is interpreted
as the probability that the players answer $(a,b)$ on questions $(i,j)$.
A correlation $p$ is \textbf{quantum} if there are 
\begin{enumerate}[(a)]
    \item finite-dimensional Hilbert spaces $H_A$ and $H_B$,
    \item a projective measurement $\{M^i_a\}_{a \in O_i}$ on $H_A$ for every $i \in I$,
    \item a projective measurement $\{N^i_a\}_{a \in O_i}$ on $H_B$ for every $i \in I$, and
    \item a state $\ket{v} \in H_A \otimes H_B$
\end{enumerate}
such that $p(a,b|i,j) = \braket{v|M^i_a \otimes N^j_b|v}$ for all $i,j \in I$,
$a \in O_i$, $b \in O_j$. A collection $(H_A,H_B,\{M^i_a\}, \{N^i_a\},
\ket{v})$ as in (a)-(d) is called a \textbf{quantum strategy}. 
A correlation $p$ is \textbf{commuting operator} if there is
\begin{enumerate}[(i)]
    \item a Hilbert space $H$,
    \item projective measurements $\{M^i_a\}_{a \in O_i}$ and $\{N^i_a\}_{a \in O_i}$ on $H$
        for every $i \in I$, and
    \item a state $\ket{v} \in H$
\end{enumerate}
such that $M^i_a N^j_b = N^j_b M^i_a$ and $p(a,b|i,j) = \braket{v|M^i_a
N^j_b|v}$ for all $i,j \in I$ and $a \in O_i$, $b \in O_j$. A collection
$(H,\{M^i_a\},\{N^i_a\},\ket{v})$ as in (i)-(iii) is called a \textbf{commuting
operator strategy}.  The set of quantum correlations for a scenario
$(I,\{O_i\})$ is denoted by $C_q(I,\{O_i\})$, and the set of commuting operator
correlations is denoted by $C_{qc}(I,\{O_i\})$. If the scenario is clear from
context, then we denote these sets by $C_q$ and $C_{qc}$. Any quantum
correlation is also a commuting operator correlation, so $C_q \subseteq
C_{qc}$. If a commuting operator correlation has a commuting operator strategy
on a finite-dimensional Hilbert space $H$, then it is also a quantum correlation,
but in general $C_{qc}$ is strictly larger than $C_q$. 

The \textbf{winning probability} of a correlation $p$ in a nonlocal game $\mcG = (I,\{O_i\},\pi,V)$
is 
\begin{equation*}
    \omega(\mcG;p) := \sum_{i,j \in I} \sum_{a \in O_i, b \in O_j} \pi(i,j) V(a,b|i,j) p(a,b|i,j).
\end{equation*}
The \textbf{quantum value} of $\mcG$ is
\begin{equation*}
    \omega_q(\mcG) := \sup_{p\in C_q} \omega(\mcG;p)
\end{equation*}
and the \textbf{commuting operator value} is
\begin{equation*}
    \omega_{qc}(\mcG) := \sup_{p \in C_{qc}} \omega(\mcG;p).
\end{equation*}
A correlation $p$ is \textbf{perfect} for $\mcG$ if $\omega(\mcG;p) = 1$, and
$\eps$-\textbf{perfect} if $\omega(\mcG;p) \geq 1-\eps$. A strategy
is $\eps$-perfect if its corresponding correlation is $\eps$-perfect. 
The set $C_{qc}$ is closed and compact, so $\mcG$ has a perfect commuting
operator correlation if and only if $\omega_{qc}(\mcG) = 1$. However, $C_{q}$ is
not necessarily closed, and there are games $\mcG$ with $\omega_q(\mcG)=1$
which do not have a perfect quantum correlation. A correlation $p$ is
\textbf{quantum approximable} if it belongs to the closure $C_{qa} :=
\overline{C_q}$, and a game $\mcG$ has a perfect quantum approximable
correlation if and only if $\omega_q(\mcG)=1$.

A nonlocal game $\mcG = (I,\{O_i\},\pi,V)$ is \textbf{synchronous} if
$V(a,b|i,i) = 0$ for all $i \in I$ and $a \neq b \in O_i$.
A correlation $p$ is \textbf{synchronous} if $p(a,b|i,i) =
0$ for all $i \in I$ and $a \neq b \in O_i$. The set of synchronous quantum
(resp. commuting operator) correlations is denoted by $C_{q}^s$ (resp.
$C_{qc}^s$). A correlation $p$ belongs to
$C_{qc}^s$ (resp. $C_{q}^s$) if and only if there is
\begin{enumerate}[(A)]
    \item a Hilbert space $H$ (resp. finite-dimensional Hilbert space $H$),
    \item a projective measurement $\{M^i_a\}_{a \in O_i}$ on $H$ for all $i \in I$, and
    \item a state $\ket{v} \in H$
\end{enumerate}
such that $\ket{v}$ is tracial, in the sense that $\braket{v|\alpha \beta|v} =
\braket{v|\beta \alpha|v}$ for all $\alpha$ and $\beta$ in the $*$-algebra
generated by the operators $M^i_a$, $i \in I$, $a \in O_i$, and $p(a,b|i,j) =
\braket{v|M^i_a M^j_b|v}$ for all $i,j \in I$, $a \in O_i$, $b \in O_j$. A
collection $(H,\{M^i_a\},\ket{v})$ as in (A)-(C) is called a
\textbf{synchronous commuting operator strategy}. If, in addition, $H$
is finite-dimensional, then $(H,\{M^i_a\},\ket{v})$ is also called a
\textbf{synchronous quantum strategy}. The synchronous quantum and
commuting operator values $\omega_q^s(\mcG)$ and $\omega_{qc}^s(\mcG)$ of a
game $\mcG$ are defined equivalently to $\omega_q(\mcG)$ and
$\omega_{qc}(\mcG)$, but with $C_q$ and $C_{qc}$ replaced by $C_q^s$ and
$C_{qc}^s$. A synchronous strategy $(H,\{M^i_a\},\ket{v})$ for a game 
$\mcG = (I,\{O_i\},\pi,V)$ is \textbf{oracularizable} if $M^i_a M^j_b = M^j_b
M^i_a$ for all $i,j \in I$, $a \in O_i$, $b \in O_j$ with $\pi(i,j) > 0$.

A theorem of Vidick \cite{Vidick_2022} (see also \cite{Pad22}) states that every quantum
correlation which is close to being synchronous, in the sense that $p(a,b|i,i)
\approx 0$ for all $i \in I$ and $a \neq b \in O_i$, is close to a synchronous
quantum correlation. This theorem has been extended to
commuting operator correlations by \cite{lin2023synchronous}. As a result, the synchronous
quantum and commuting values of a game are polynomially related to the
non-synchronous quantum and commuting values.  We use a version of this result
due to Marrakchi and de la Salle \cite{marrakchi2023synchronous}. Following
\cite{marrakchi2023synchronous}, say that a probability distribution on $I \times I$
is \textbf{$C$-diagonally dominant} if $\pi(i,i) \geq C \sum_{j \in I}
\pi(i,j)$ and $\pi(i,i) \geq C \sum_{j \in I} \pi(j,i)$ for all $i \in I$.
Then:
\begin{theorem}[\cite{marrakchi2023synchronous}]\label{thm:synchrounding}
    Suppose $\mcG$ is a synchronous game with a $C$-diagonally dominant
    question distribution. If $\omega_q(\mcG)$ (resp. $\omega_{qc}(\mcG)$) is
    $\geq 1-\eps$, then $\omega_q^s(\mcG)$ (resp. $\omega_{qc}^s(\mcG)$) is
    $\geq 1 - O((\eps/C)^{1/4})$.
\end{theorem}

A \textbf{two-prover one-round $\MIP$ protocol} is a family of nonlocal games
$\mcG_x = (I_x,\{O_{xi}\}_{i \in I_x}, \pi_x, V_x)$ for $x \in \{0,1\}^*$,
along with a probabilistic Turing machine $S$ and another Turing machine $V$,
such that 
\begin{itemize}
    \item for all $x \in \{0,1\}^*$ and $i \in I_x$, there are integers 
        $n_x$ and $m_{xi}$ such that $I_x = \{0,1\}^{n_x}$ and $O_{xi}
            = \{0,1\}^{m_{xi}}$, 

    \item on input $x$, the Turing machine $S$ outputs $(i,j) \in I \times I$
        with probability $\pi_x(i,j)$, and  

    \item on input $(x,a,b,i,j)$, the Turing machine $V$ outputs $V_x(a,b|i,j)$. 
\end{itemize}
Let $c, s : \{0,1\}^* \to \Q$ be computable functions with $c(x) > s(x)$
for all $x \in \{0,1\}^*$. A language $\mcL \subset \{0,1\}^*$ belongs $\MIP^*(2,1,c,s)$ if
there is a MIP protocol $(\{\mcG_x\}, S, V)$ such that $n_x$ and $m_{xi}$ are
polynomial in $|x|$, $S$ and $V$ run in polynomial time in $|x|$, if $x \in
\mcL$ then $\omega_q(\mcG_x) \geq c$, and if $x \not\in \mcL$ then
$\omega_q(\mcG_x) \leq s$. The function $c$ is called the \textbf{completeness
probability}, and $s$ is called the \textbf{soundness probability}. The functions
$n_x$ and $m_{xi}$ are called the \textbf{question length} and \textbf{answer length}
respectively. The class $\MIP^{co}(2,1,c,s)$ is defined equivalently to
$\MIP^*(2,1,c,s)$, but with $\omega_q$ replaced by $\omega_{qc}$. The protocols
in these cases are called $\MIP^*$ and $\MIP^{co}$ protocols. A language
belongs to $\AM^*(2)$ (resp.  $\AM^{qc}(2)$) if it has a $\MIP^*$-protocol
(resp.  $\MIP^{qc}$-protocol) in which $\pi_x$ is the uniform distribution on
$I_x \times I_x$. Such a protocol is called an $\AM^*(2)$ protocol. We can 
also define classes $\SynMIP^*$ and $\SynMIP^{co}$ by replacing the quantum
and commuting operator values by $\omega_q^s$ and $\omega_{qc}^s$.

Any language in $\MIP^*(2,1,c,s)$ is contained in $\RE$, and this remains true
even if we add more provers and rounds of communication. The $\MIP^*=\RE$ theorem
of Ji, Natarajan, Vidick, Wright, and Yuen states that $\MIP^*(2,1,1,1/2) = \RE$
\cite{ji2022mipre}. In this paper, we use the following strong version of $\MIP^* =
\RE$ due to Natarajan and Zhang \cite{natarajan2023quantum}.
\begin{theorem}[$\MIP^*=\RE$]\label{thm:mipre}
    There is a two-prover one round $\AM^*(2)$ protocol $(\{\mcG_x\}, S, V)$
    for the halting problem with completeness $c=1$ and soundness $s=1/2$, 
    such that $\mcG_x$ is a synchronous game with constant length questions,
    and $\polylog(|x|)$ length answers. Furthermore, if $\mcG_x$ has a perfect
    strategy, then it has a perfect oracularizable synchronous quantum strategy.
\end{theorem}
\begin{proof}
    \cite{natarajan2023quantum} shows that there is $\MIP^*$ protocol for the halting problem
    meeting this description. As they observe, any $\MIP^*$ protocol with a constant
    number of questions can be turned into an $\AM^*(2)$ protocol with
    completeness $c=1$ and soundness $s<1$, and then parallel repetition
    (see \Cref{sec:prep}) can be used to lower the soundness back to $1/2$.
\end{proof}

One corollary of \Cref{thm:mipre} is that it is possible to transform any
$\MIP^*$ protocol into an equivalent $\AM^*(2)$ protocol $(\{\mcG_x\}, S, V)$
as in the theorem.  Indeed, suppose $\mcP$ is a polynomial-time probabilistic
interactive Turing machine which on input $x$ acts as the verifier in a
$\MIP^*$ protocol with $k$ rounds, $p$ provers, completeness $c$, and soundness
$s$, where $k$, $p$, $c$, and $s$ are computable functions of $|x|$.  Let
$\mcT$ be the Turing machine which on input $x$, searches through $k$-round
$p$-prover quantum strategies, uses $\mcP$ to calculate the success
probability, and halts if it finds a strategy with success probability $> s$.
Let $\mcT(x)$ be the Turing machine which on empty input writes $x$ to the
input tape and then runs $\mcT$.  Finally, let $(\{\mcG_M\}, S, V)$ be the
one-round protocol for the language $\HALT = \{ M : M \text{ is a Turing
machine that halts on empty input}\}$.  The Turing machines $S$ and $V$ run in
polynomial time in the size $|M|$ of the input Turing machine $M$, and
$\mcT(x)$ has size linear in $|x|$, so the one-round protocol which runs game
$\mcG_{\mcT(x)}$ on input $x$ is a polynomial-time $\AM^*(2)$ protocol which
recognizes the same language as $\mcP$. Strikingly, this works for any
computable $k$, $p$, and $s$, not just polynomial functions of $|x|$,
since the only requirement is that $\mcT(x)$ have polynomial description size.

\begin{remark}
The underlying statement of \Cref{thm:main} (see \Cref{thm:main1}) is that
there is a two-prover perfect-zero knowledge $\MIP^*$ protocol for the halting
problem. Hence the same argument as above shows that there is an effective
procedure for transforming any $\MIP^*$ protocol into a two-prover perfect zero
knowledge $\MIP^*$ protocol.
\end{remark}

\section{BCS games}\label{sec:BCS}

We now introduce boolean constraint system games. If $V$ is a set of variables,
a \textbf{constraint on $V$} is a subset $C$ of $\Z_2^V$. We think of $\Z_2$ as
$\{\pm 1\}$ rather than $\{0,1\}$, since this is more convenient when working
with observables and measurements. In particular, we use $-1$ and $1$ to
represent true and false respectively, rather than $1$ and $0$. An
\textbf{assignment to $V$} is an element $\phi \in \Z_2^V$, and we refer to the
elements of $C$ as \textbf{satisfying assignments for $C$}. For convenience, we
assume every constraint is non-empty, i.e. has a satisfying assignment. 
A \textbf{boolean constraint system} (BCS) $B$ is a pair
$\left(X,\{(V_i,C_i)\}_{i=1}^m\right)$, where $X$ is an ordered set of
variables, $V_i$ is a nonempty subset of $X$ for all $1\leq i\leq m$, and $C_i$
is a constraint on the variables $V_i$. When working with nonlocal games, the
sets $V_i$ are sometimes called the \textbf{contexts} of the system. The order on
$X$ induces an order on the contexts $V_i$, and this will be used for some
specific models of the weighted BCS algebra in \Cref{sec:stability}. This is the
only thing we use the order on $X$ for, so it can be ignored otherwise. 
A \textbf{satisfying assignment for $B$} is an assignment $\phi$ to $X$ such
that $\phi|_{V_i} \in C_i$ for all $1 \leq i \leq m$.  Although we won't use it
until later, we define the \textbf{connectivity} of a BCS $B$ to be the maximum
over $i$ of $|\{(x,j) \in V_i \times [m] : x \in V_j\}|$, where $[m] :=
\{1,\ldots,m\}$.  In other words, the connectivity is the maximum over $i$ of
the number of times the variables in constraint $i$ appear in the constraints
of $B$. Also, if $V = \bigcup_{i=1}^k V_i$ and $C_i$ is a constraint on $V_i$,
then the \textbf{conjunction} $\wedge_{i=1}^k C_i$ is the constraint $C$ on
variables $V$ such that $\phi \in C$ if and only if $\phi|_{V_i} \in C_i$ for
all $1 \leq i \leq k$.

Let $B = \left(X,\{(V_i,C_i)\}_{i=1}^m\right)$ be a BCS, and let $\pi$ be a
probability distribution on $[m] \times [m]$.  The \textbf{BCS game} $\mcG(B,\pi)$ is the nonlocal game $([m],
C_{i\in m}, \pi, V)$, where $V(\phi_i,\phi_j|i,j) = 1$ if $\phi_i|_{V_i \cap
V_j} = \phi_j|_{V_i \cap V_j}$, and is $0$ otherwise. In other words, in
$\mcG(B,\pi)$, the players are given integers $i,j \in [m]$ according to the
distribution $\pi$, and must reply with satisfying assignments $\phi_i \in C_i$
and $\phi_j \in C_j$ respectively. They win if their assignments agree on the
variables in $V_i \cap V_j$. With this definition, $\mcG(B,\pi)$ has questions
of length $\lceil \log m \rceil$, and answer sets of length $|V_i|$. 

A \textbf{$\BCS$-$\MIP$ protocol} is a family of BCS games
$\mcG(B_x,\pi_x)$, where $B_x = (X_x,\{(V_i^x,C_i^x)\}_{i=1}^{m_x})$, along
with a probabilistic Turing machine $S$ and another Turing machine $C$, such
that 
\begin{enumerate} 
    \item on input $x$, $S$ outputs $(i,j) \in [m_x] \times [m_x]$ with probability
$\pi_x(i,j)$, and
    \item on input $(x, \phi, i)$, $C$ outputs true if $\phi \in C_i^x$ and false
        otherwise.
\end{enumerate}
Technically, this definition should also include some way of computing the sets
$X_x$ and $V_i^x$. For instance, we might say that the integers $|N_x|$ and
$|V_i^x|$ are all computable, and there are computable order-preserving
injections $[|V_i^x|] \to [|X_x|]$. However, for simplicity we ignore this
aspect of the definition going forward, and just assume that in any
$\BCS$-$\MIP^*$ protocol, we have some efficient way of working with the sets
$X_x$ and $V_i^x$, the intersections $V_i^x \cap V_j^x$, and assignments
$\phi\in\Z_2^{V_i^x}$. 
A language $\mcL$ belongs to the complexity class $\BCS$-$\MIP^*(s)$ if there is a
$\BCS$-$\MIP$ protocol as above such that $\lceil \log m_x \rceil$ and $|V_i^x|$
are polynomial in $|x|$, $S$ and $C$ run in polynomial time, if $x \in \mcL$ then
$\omega_q^s(\mcG_x) = 1$, and if $x \not\in \mcL$ then $\omega_q^s(\mcG_x) \leq s$.
The parameter $s$ is called the soundness. Any $\BCS$-$\MIP^*$ protocol for 
$\mcL$ can be transformed into a $\SynMIP^*$ protocol by playing the game $\mcG_x$
with the answer sets $C_i$ replaced by $\Z_2^{V^x_i}$, and on input $(x,\phi,\psi,
i,j)$, asking the verifier $V$ to first check that $\phi \in C_i$ and $\psi \in C_j$
using $C$, and then checking that $\phi|_{V_i \cap V_j} = \psi|_{V_i \cap V_j}$.
Hence $\BCS$-$\MIP^*(s)$ is contained in $\SynMIP^*(2,1,1,s)$. Notice that in
this modified version of the BCS game, the players are allowed to answer with
non-satisfying assignments, but they always lose if they do so. Thus any
strategy for the modified game can be converted into a strategy for the original
game with the same winning probability, and perfect strategies for both types of
games (ignoring questions that aren't in the support of $\pi$) are identical, 
so the $\SynMIP^*$ protocol has the same completeness and soundness as the
$\BCS$-$\MIP^*$ protocol.  The class $\BCS$-$\MIP^{co}(s)$ can be defined
similarly by replacing $\omega_q$ with $\omega_{qc}$, and is contained in
$\SynMIP^{co}(2,1,1,s)$.  We can also define subclasses of $\BCS$-$\MIP^{*}$
and $\BCS$-$\MIP^{co}$. For instance, we let 3SAT-$\MIP^*$ be the class of
languages with a $\BCS$-$\MIP^*$ protocol $(\{\mcG(B_x,\pi_x)\},S,C)$, in which
every constraint of $B_x$ is a 3SAT clause, i.e. a disjunction $x \OR y \OR z$,
where $x,y,z$ are either variables from $B_x$, or negations of said variables,
or constants. 

If the players receive the same question $i \in [m]$, then they must reply with
the same assignment $\phi$ to win. Consequently, if $\pi(i,i) > 0$ for all $i$
then $\mcG(B,\pi)$ is a synchronous game. This version of BCS games is
sometimes called the constraint-constraint version of the game. There is are
other variants of BCS games, sometimes called constraint-variable BCS games, in
which one player receives a constraint and another receives a variable (see
\cite{cleve2013characterization}). In this paper, we work with constraint-constraint games
exclusively, but the two types of BCS games are closely related, and can often
be used interchangeably. As per the previous section, a synchronous strategy
for $\mcG(B,\pi)$ consists of projective measurements $\{M^i_\phi\}_{\phi \in
\Z_2^{V_i}}$, $i \in [m]$, on a Hilbert space $\mcH$, along with a state $\ket{v} \in
\mcH$ which is tracial on the algebra generated by $M^i_{\phi}$.


Conversely, it is well-known that every synchronous game $\mcG = (I,
\{\mcO_i\}, \pi, V)$ can be turned into a BCS game. One way to do this (see,
e.g. \cite{PS23,Pad22}) is to make a constraint system with variables $x_{ia}$ for $i
\in I$ and $a \in \mcO_i$, 
and constraints $\OR_{a \in \mcO_i} x_{ia} = \true$ for
all $i \in [m]$ and $x_{ia} \AND x_{jb} = \false$ whenever $V(a,b|i,j) = 0$.
The variable $x_{ia}$ represents whether the player answers $a$ on
input $i$, and the constraints express the idea that the players must choose an answer for
every question, and that they should reply with winning answers (the
synchronous condition on $V$ implies that $x_{ia} \AND x_{ib} = \false$ 
is a constraint for all $i$ and $a \neq b$, which means that the players should
choose a single answer for question $i$). The BCS game $\mcG'$ associated to
this constraint system has a perfect quantum (resp. quantum approximable,
commuting operator) strategy if
and only if $\mcG$ has a perfect quantum (resp. quantum approximable, commuting
operator) strategy. Unfortunately, this construction results in a game with
answer sets $\{\pm 1\}^{O_i}$, which means that the bit-length of the
answers increases exponentially from $\mcG$. If $\omega_q(\mcG) = 1 - \eps$,
then $\omega_q(\mcG) = 1 - O(\eps / |O_i|)$, meaning that if this construction
is used in a $\MIP^*$-protocol, soundness can drop of exponentially.

To fix this, we look at the oracularization $\mcG^{orac}$ of $\mcG$. There are
several versions of $\mcG^{orac}$ in the literature, all closely related. We
use the version from \cite{natarajan2019neexp}, in which the verifier picks a question pair
$(i_1,i_2) \in I$ according to $\pi$. The verifier then picks $a,b,c \in \{1,2\}$
uniformly at random. When $a=1$, they send player $b$ both questions $(i_1,i_2)$,
and the other player question $(i_c)$. Player $b$ must respond with $a_j \in O_j$
such that $V(a_1,a_2|i_1,i_2)=1$, and the other player responds with $b \in O_{i_c}$.
The players win if $a_c = b$. If $a=2$, both players are sent $(i_1,i_2)$ and must respond
with $(a_1,a_2)$ and $(b_1,b_2)$ in $O_{i_1} \times O_{i_2}$. They win if $(a_1,a_2)
= (b_1,b_2)$.  If $\mcG$ has questions of length $q$ and answers of length $a$, 
then $\mcG^{orac}$ has questions of length $2q$ and answers of length $2a$, 
so this construction only increases the question and answer length polynomially.
The following lemma shows that this construction is sound, in the sense
that $\omega_q(\mcG^{orac})$ cannot be much larger than $\omega_q(\mcG)$.
\begin{lemma}[\cite{natarajan2019neexp,ji2022mipre}]\label{lem:orac_sound}
    Let $\mcG$ be a synchronous game. If $\mcG$ has an perfect oracularizable
    synchronous strategy, then $\mcG^{orac}$ has a perfect synchronous strategy.
    Conversely, if $\omega_q(\mcG^{orac}) = 1 - \eps$, then $\omega_q(\mcG)
    \geq 1 - \poly(\eps)$. 
\end{lemma}
\begin{proof}
    This is asserted in Definition 17.1 of \cite{natarajan2019neexp}. Although a proof isn't
    supplied, the proof follows the same lines as Theorem 9.3 of \cite{ji2022mipre}. 
\end{proof}

Given a synchronous game $\mcG = (I, \{O_i\}, \pi, V)$ where $I \subseteq
\{0,1\}^n$ and $O_i\subseteq \{0,1\}^{m_i}$, construct a constraint system
$B$ as follows. Take $X$ to be the set of variables $x_{ij}$, where $i \in I$
and $1 \leq j \leq m_i$. Let $V_i = \{x_{ij}, 1 \leq j \leq m_i\}$, and
identify $\Z_2^{V_i}$ with bit strings $\{0,1\}^{m_i}$, where the assignment
to $x_{ij}$ corresponds to the $j$th bit, and let $C_i \subseteq \Z_2^{V_i}$
be the subset corresponding to $O_i$. Let $P = \{(i,j) \in I \times I : 
\pi(i,j) > 0\}$. For $(i,j) \in P$, let $V_{ij} = V_i \cup V_j$, and let $C_{ij}
\subset \Z_2^{V_{ij}} = \Z_2^{V_i} \times \Z_2^{V_j}$ be the set of pairs
of strings $(a,b)$ such that $a \in O_i$, $b \in O_j$, and $V(a,b|i,j) = 1$.
Then $B$ is the constraint system with variables $X$ and constraints
$\{(V_i,C_i)\}_{i \in I}$ and $\{(V_{ij},C_{ij})\}_{(i,j) \in P}$. Let
$I' = I \cup P$ and $\pi^{orac}$ be the probability distribution on 
$I' \times I'$ such that
\begin{equation*}
    \pi^{orac}(i',j') = \begin{cases} 
             \tfrac{1}{8}\pi(i,j)       & i' = (i,j), j' = i \\
             \tfrac{1}{8}\pi(i,j)       & i' = (i,j), j' = j \\
             \tfrac{1}{8}\pi(i,j)       & i' = i, j' = (i,j) \\
             \tfrac{1}{8}\pi(i,j)       & i' = j, j' = (i,j) \\
            \tfrac{1}{2} \pi(i,j) & i' = j' = (i,j)  \\
            0 & \text{ otherwise}
    \end{cases}
\end{equation*}
Then $\mcG(B,\pi^{orac}) = \mcG^{orac}$, so the oracularization of a synchronous game
is a BCS game. As a result, \Cref{thm:mipre} has the following corollary:
\begin{cor}
    There is a $\BCS$-$\MIP^*$ protocol $(\{\mcG(B_x,\pi_x)\}, S, V)$ for the
    halting problem with constant soundness $s<1$, in which $B_x$ has a
    constant number of contexts and contexts of size $\polylog(|x|)$,
    and $\pi_x$ is the uniform distribution on pairs of contexts.
\end{cor}
\begin{proof}
    Let $(\{\mcG_x\},S,V)$ be the protocol from \Cref{thm:mipre}. Then
    $\mcG_x^{orac}$ is a BCS game in which the underlying BCS has a constant
    number of contexts, and the contexts have size $\polylog(|x|)$. 
    The probability distribution $\pi^{orac}$ and the constraints of $\mcG^{orac}$
    can be computed in polynomial time from $S$ and $V$, so by
    \Cref{lem:orac_sound} there is a $\BCS$-$\MIP^*$ protocol for the halting
    problem with constant soundness $s' < 1$. The probability distribution
    $\pi_x$ in the oracularization construction is not uniform.  However, it is
    not hard to see that changing the distribution $\pi_x$ in the
    oracularization game does not change completeness, and since there are only
    a constant number of contexts, replacing $\pi_x$ with the uniform
    distribution yields only a constant dropoff in soundness.
\end{proof}

\section{BCS algebras and approximate representations}\label{sec:weightedalg}

It is often worth thinking about synchronous strategies more abstractly. Recall
that $\C \Z_2^{*V}$ is the $*$-algebra generated by variables $x \in V$,
satisfying the relations $x^2 = x^* x = x x^* = 1$ for all $x \in V$, and $\C
\Z_2^{V}$ is the quotient of $\C \Z_2^{*V}$ by the relations $xy = yx$ for all
$x,y \in V$. Given an assignment $\phi$ to an ordered set of variables $V$, we let 
\begin{equation*}
    \Phi_{V,\phi} := \prod_{x \in V} \tfrac{1}{2}(1 + \phi(x) x) 
\end{equation*}
considered as a polynomial in $\C \Z_2^{*V}$, where the product is taken
with respect to the order on $V$. Given a constraint $C$ on $V$, we let 
\begin{equation*}
    \mcA(V,C) = \C \Z_2^{V} / \ang{\Phi_{V,\phi}=0 \text{ for } \phi \not\in C}.
\end{equation*}
Since $\C \Z_2^{V}$ is commutative, the image of $\Phi_{V,\phi}$ in $\C \Z_2^{V}$
is independent of the order of $V$; however, we will work with $\C \Z_2^{*V}$
in \Cref{sec:stability}.  The algebra $\mcA(V,C)$ is isomorphic to the algebra 
\begin{equation*}
    \C^*\ang{ m_{\phi}, \phi \in C : m_{\phi}^* = m_{\phi} = m_{\phi}^2 \text{ for all } \phi \in C
            \text{ and } \sum_{\phi \in C} m_{\phi} = 1},
\end{equation*}
where the isomorphism identifies $m_{\phi}$ with $\Phi_{V,\phi}$. In
particular, $\C \Z_2^{V} = \mcA(V,\Z_2^{V})$ is generated by $\Phi_{V,\phi}$ for
$\phi \in \Z_2^{V}$. Consequently if $\sigma : \mcA(V,C) \to \mcB(\mcH)$ is a
$*$-representation, then $\{\sigma(\Phi_{V,\phi})\}_{\phi \in C}$ is a
projective measurement on $\mcH$, and conversely if $\{M_{\phi}\}_{\phi \in C}$
is a projective measurement on $\mcH$, then there is a $*$-representation
$\sigma : \mcA(V,C) \to \mcB(\mcH)$ with $\sigma(\Phi_{V,\phi}) = M_{\phi}$. 

If $B = (X,\{(V_i,C_i)\}_{i=1}^m)$ is a BCS, then we let $\mcA(B)$ denote the
free product $\mcA(B) := \ast_{i \in [m]} \mcA(V_i,C_i)$. We let $\sigma_i :
\mcA(V_i,C_i) \to \mcA(B)$ denote the natural inclusion of the $i$th factor, so
$\mcA(B)$ is generated by the involutions $\sigma_i(x)$ for $i \in [m]$ and $x
\in V_i$. Equivalently, $\mcA(B)$ is generated by the projections
$\sigma_i(\Phi_{V_i,\phi})$ for $i \in [m]$ and $\phi \in C_i$. To avoid
clogging up formulas with symbols, we'll often write $\Phi_{V_i,\phi}$ instead
of $\sigma_i(\Phi_{V_i,\phi})$ when it's clear what subalgebra $\mcA(V_i,C_i)$
the element belongs to.  As with $\mcA(V,C)$, representations $\alpha$ of $\mcA(B)$
are in bijective correspondence with families of projective measurements
$\{M^i_{\phi}\}_{\phi \in C_i}$, $i \in [m]$ via the relation $M^i_{\phi} =
\alpha(\Phi_{V_i,\phi})$.  If $(\{M^i_{\phi}\}, \ket{v}, \mcH)$ is a
synchronous commuting operator strategy for $\mcG(B,\pi)$, and $\alpha :
\mcA(B) \to \mcB(\mcH)$ is the representation with $\alpha(\Phi_{V_i,\phi}) =
M^i_{\phi}$, then $a \mapsto \braket{v|\alpha(a)|v}$ is a tracial state on
$\mcA(B)$.
Conversely, if $\tau$ is a tracial state on $\mcA(B)$, then the GNS
representation theorem implies that there is a synchronous commuting operator
strategy $\mcS = (\{M^i_\phi\}, \ket{v}, \mcH)$ such that $\tau(a) =
\braket{v|\alpha(a)|v}$ where $\alpha$ is the representation corresponding to
$\{M^i_{\phi}\}$. Note that the trace is faithful on the image of the GNS representation. As a result, synchronous commuting operator strategies for
$\mcG(B,\pi)$ and tracial states on $\mcA(B)$ can be used interchangeably, and
in particular $p \in C_{qc}$ if and only if there is a tracial state $\tau$
with $p(\phi,\psi|i,j) = \tau(\Phi_{V_i,\phi} \Phi_{V_j,\psi})$ for all
$i$,$j$, $\phi$, and $\psi$.
Finite-dimensional tracial states on $\mcA(B)$ can be used interchangeably with
synchronous quantum strategies for $\mcG(B,\pi)$, and $p \in C_q$ if and only
if there is a finite-dimensional tracial state $\tau$ with $p(\phi,\psi|i,j) =
\tau(\Phi_{V_i,\phi} \Phi_{V_j,\psi})$ for all $i$,$j$, $\phi$, and $\psi$.
Similarly, $p \in C_{qa}$ if and only if there is a
Connes-embbedable tracial state $\tau$ such that $p(\phi,\psi|i,j) =
\tau(\Phi_{V_i,\phi} \Phi_{V_j,\psi})$ for
all $i$,$j$, $\phi$, and $\psi$ \cite{Kim_2018}.

A correlation $p$ is perfect for a BCS game $\mcG(B,\pi)$ if
$p(\phi,\psi|i,j)=0$ whenever $\pi(i,j) > 0$ and $(\phi,\psi)$ is a losing
answer to questions $(i,j)$.  As a result, a tracial state $\tau$ on $\mcA(B)$
is \textbf{perfect} (aka. corresponds to a perfect correlation) if and only if
$\tau(\Phi_{V_i,\phi} \Phi_{V_j,\psi}) = 0$ whenever $\phi|_{V_i \cap V_j} \neq
\psi|_{V_i \cap V_j}$. Consequently a tracial state on $\mcA(B)$ is perfect 
for $\mcG(B,\pi)$ if and only if it is the pullback of a tracial state on the
\textbf{synchronous algebra} of $\mcG(B,\pi)$, which is the quotient
\begin{align*}
    \SynAlg(B,\pi) = \mcA(B) / \ang{& \Phi_{V_i,\phi} \Phi_{V_j,\psi} = 0 \text{ for all }
            i,j \in [m] \text{ with } \pi(i,j) > 0 \\ & \text{ and } \phi \in C_i,
            \psi \in C_j \text{ with } \phi|_{V_i \cap V_j} \neq \psi|_{V_i \cap V_j}}.
\end{align*}
For BCS games, this result about perfect strategies is due to Kim, Paulsen, and
Schafhauser \cite{Kim_2018}.  The general notion of a synchronous algebra is due to
\cite{helton2017algebras}. In \cite{Goldberg_2021,PS23}, it is shown that the synchronous algebra
of a BCS game is isomorphic to the so-called BCS algebra of the game. In
working with $\MIP^*$ protocols, we also need to keep track of $\eps$-perfect
strategies.  In \cite{Pad22}, it is shown that $\eps$-perfect strategies
for a BCS game correspond to $\eps$-representations of the BCS algebra,
where an $\eps$-representation is a representation of $\mcA(B)$ such that
all the defining relations of $\SynAlg(B,\pi)$ are bounded by $\eps$ in the
normalized Frobenius norm. In this prior work, the focus was on the behaviour
of $\eps$-perfect strategies for a fixed game, so the number of questions and
answers was constant. For $\MIP^*$ protocols, the game size is not constant,
and we need to work with approximate representations where the average, rather
than the maximum, of the norms of the defining relations is bounded. For this, we
introduce the following algebraic structure:

\begin{definition}
    A \textbf{(finitely-supported) weight function} on a set $X$ is a function
    $\mu : X \to [0,+\infty)$ such that $\supp(\mu) := \mu^{-1}((0,+\infty))$
    is finite. A \textbf{weighted} $*$\textbf{-algebra} is a pair $(\mcA,\mu)$
    where $\mcA$ is a $*$-algebra and $\mu$ is a weight function on $\mcA$. 

    If $\tau$ is a tracial state on $\mcA$, then the \textbf{defect of $\tau$} is
    \begin{equation*}
        \df(\tau; \mu) := \sum_{a \in \mcA} \mu(a) \|a\|^2_{\tau},
    \end{equation*}
    where $\|a\|_{\tau} := \sqrt{\tau(a^* a)}$ is the $\tau$-norm.
    When the weight function is clear, we just write $\df(\tau)$. 
\end{definition}
Since $\mu$ is finitely supported, the sum in the definition of the
defect is finite, and hence is well-defined. Note that traces $\tau$ on a
weighted algebra $(\mcA,\mu)$ with $\df(\tau) = 0$ correspond to traces on the
algebra $\mcA / \ang{\supp(\mu)}$. In general, $\df(\tau)$ is a measure of how
far $\tau$ is from being a trace on $\mcA$.  Thus we can think of a weighted
algebra $(\mcA,\mu)$ as a presentation or model for the algebra $\mcA /
\ang{\supp(\mu)}$ that allows us to talk about approximate traces on this
algebra. 

\begin{definition}
    Let $B = (X,\{(V_i,C_i)\}_{i=1}^m)$ be a BCS, and let $\pi$ be a
    probability distribution on $[m] \times [m]$. The \textbf{(weighted) BCS
    algebra} $\mcA(B,\pi)$ is the $*$-algebra $\mcA(B)$, with weight function
    $\mu_{\pi}$ defined by 
    \begin{equation*}
        \mu_{\pi}( \Phi_{V_i,\phi} \Phi_{V_j,\psi} ) = \pi(i,j)
    \end{equation*}
    for all $i,j \in [m]$ and $\phi \in C_i$, $\psi \in C_j$ with $\phi|_{V_i
    \cap V_j} \neq \psi|_{V_i \cap V_j}$, and $\mu_{\pi}(r)=0$ for all other $r
    \in \mcA(B)$. 
\end{definition}
Note that $\mcA(B) / \ang{\supp(\mu_{\pi})}$ is the synchronous algebra
$\SynAlg(B,\pi)$ defined above, so $\mcA(B,\pi)$ is a model of this synchronous
algebra, and perfect strategies for $\mcG(B,\pi)$ correspond to tracial states
$\tau$ on $\mcA(B,\pi)$ with $\df(\tau) = 0$. The following lemma is an immediate
consequence of the definitions:
\begin{lemma}\label{lem:trace-strat}
    Let $B = (X,\{(V_i,C_i)\}_{i=1}^m)$ be a BCS, and let $\pi$ be a
    probability distribution on $[m] \times [m]$. A tracial state $\tau$
    on $\mcA(B)$ is an $\eps$-perfect strategy for $\mcG(B,\pi)$ if and only
    if $\df(\tau) \leq \eps$. 
\end{lemma}
\begin{proof}
    Let $p$ be the correlation corresponding to $\tau$, so $p(\phi,\psi|i,j)
    = \tau(\Phi_{V_i,\phi} \Phi_{V_j,\psi})$. Then 
    \begin{equation*}    
        \df(\tau) = \sum \pi(i,j) \tau(\Phi_{V_i,\phi} \Phi_{V_j,\psi}),
    \end{equation*}
    where the sum is across $i,j \in [m]$ and $\phi \in C_i$, $\psi \in C_j$ with
    $\phi|_{V_i \cap V_j} \neq \psi|_{V_i \cap V_j}$. So $\df(\tau) = 1 - \omega(\mcG(B,\pi);p)$. 
\end{proof}

\section{Homomorphisms between BCS algebras}\label{sec:weightedhom}

In addition to looking at BCS games, we also want to consider transformations
between constraint systems and the corresponding games. To keep track of how
near-perfect strategies change, we introduce a notion of homomorphism for
weighted algebras.
\begin{definition}
    Let $(\mcA,\mu)$ and $(\mcB,\nu)$ be weighted $*$-algebras, and let $C > 0$. A
    $C$\textbf{-homomorphism} $\alpha:(\mcA,\mu)\to(\mcB,\nu)$ is a
    $*$-homomorphism $\alpha:\mcA\to\mcB$ such that 
    \begin{equation*}
		\alpha(\sum_{a\in\mcA}\mu(a)a^*a) \lesssim C\sum_{b\in\mcB}\nu(b)b^*b.
	\end{equation*}
\end{definition}

The point of this definition is the following:
\begin{lemma}\label{lem:Chom}
    Suppose $\alpha : (\mcA,\mu) \to (\mcB,\nu)$ is a $C$-homomorphism. If $\tau$ is a 
    trace on $(\mcB,\nu)$, then $\df(\tau \circ \alpha) \leq C \df(\tau)$. 
\end{lemma}
\begin{proof}
    Let $A = \alpha(\sum_{a\in\mcA}\mu(a)a^*a)$ and $B = \sum_{b\in\mcB}\nu(b)b^*b$. 
    Note that
    \begin{equation*}
        \df(\tau \circ \alpha) = \sum_{a \in \mcA} \mu(a) \|a\|_{\tau \circ \alpha}
            = \sum_{a \in \mcA} \mu(a) \tau(\alpha(a^* a)) = \tau(A),
    \end{equation*}
    By the definition of $\lesssim$, there are $c_1,\ldots,c_k$ and $f_1,\ldots,f_{\ell},g_{1},\ldots,g_{\ell} \in \mcB$ such that
    \begin{equation*}
        CB - A = \sum_{i=1}^k c_i^* c_i + \sum_{j=1}^{\ell} [f_j,g_j].
    \end{equation*}
    Since $\tau$ is a tracial state, $\tau(c_i^* c_i) \geq 0$ and
    $\tau([f_j,g_j])=0$ for all $i$ and $j$. Hence $C\tau(B) \geq \tau(A)$ as required.
\end{proof}

One of the first things we can apply this idea to is changing between different
presentations of the BCS algebra. For instance:
\begin{prop}\label{prop:inter}
    Suppose $B = (X, \{(V_i,C_i)\}_{i=1}^m)$ is a BCS, and $\pi$ is a
    probability distribution on $[m] \times [m]$. Let $\mu_{inter}$ be the
    weight function on $\mcA(B)$ defined by 
    \begin{equation*}
        \mu_{inter}(\sigma_i(x) - \sigma_j(x)) = \pi(i,j)
    \end{equation*}
    for all $i \neq j \in [m]$ and $x \in V_i \cap V_j$, and $\mu_{inter}(r) = 0$ for 
    other $r \in \mcA(B)$. Then the identity map $\mcA(B) \to \mcA(B)$ gives a
    $O(1)$-homomorphism $(\mcA(B),\mu_{\pi}) \to (\mcA(B),\mu_{inter})$, and a
    $O(L)$-homomorphism $(\mcA(B),\mu_{inter}) \to (\mcA(B),\mu_{\pi})$, where
    $L = \max_{i,j} |V_i \cap V_j|$.
\end{prop}
Recall that $\sigma_i : \mcA(V_i,C_i) \to \mcA(B)$ is the natural inclusion of the $i$th factor.  
\begin{proof}
    Fix $1 \leq i,j \leq m$. Since $\Phi_{V_i,\phi}$ is a projection in $\mcA(V_i,C_i)$,
    $(\Phi_{V_i,\phi} \Phi_{V_j,\psi})^* (\Phi_{V_i,\phi} \Phi_{V_j,\psi})$ is cyclically
    equivalent to $\Phi_{V_i,\phi} \Phi_{V_j,\psi}$ for all $\phi \in C_i$, $\psi \in C_j$.
    For $x \in V_i \cap V_j$, let $R_x$ be the pairs $(\phi,\psi) \in C_i \times C_j$ such
    that $\phi(x) \neq \psi(x)$. Then 
    \begin{equation*}
        \sum_{\phi|_{V_i \cap V_j} \neq \psi|_{V_i \cap V_j}} \Phi_{V_i,\phi} \Phi_{V_j,\psi}
        \lesssim \sum_{x \in V_i \cap V_j} \sum_{(\phi,\psi) \in R_x} \Phi_{V_i,\phi} \Phi_{V_j,\psi},
    \end{equation*}
    and since $\phi|_{V_i \cap V_j}$ and $\psi|_{V_i \cap V_j}$ can disagree in at most $|V_i \cap V_j|$
    places, 
    \begin{equation*}
        \sum_{x \in V_i \cap V_j} \sum_{(\phi,\psi) \in R_x} \Phi_{V_i,\phi} \Phi_{V_j,\psi} \lesssim
        |V_i \cap V_j|
        \sum_{\phi|_{V_i \cap V_j} \neq \psi_{V_i \cap V_j}} \Phi_{V_i,\phi} \Phi_{V_j,\psi}.
    \end{equation*}
    Fix $x \in V_i \cap V_j$, and let $V_i' = V_i \setminus \{x\}$, $V_j' = V_j \setminus \{x\}$. 
    \begin{align*}
        \sum_{(\phi,\psi) \in R_x} \Phi_{V_i,\phi} \Phi_{V_j,\psi} & = 
        \sum_{\phi \in \Z_2^{V_i'}, \psi \in \Z_2^{V_j'}} \Phi_{V_i',\phi} 
            \tfrac{1}{4}\left[(1+\sigma_i(x)) (1-\sigma_j(x))
            + (1-\sigma_i(x)) (1+\sigma_j(x))\right] \Phi_{V_j',\psi} \\
            & = (1+\sigma_i(x)) (1-\sigma_j(x))
            + (1-\sigma_i(x)) (1+\sigma_j(x)),
    \end{align*}
    where the last equality holds because $\sum_{\phi \in \Z_2^{V_i'}} \Phi_{V_i',\phi}$ and
    $\sum_{\psi \in \Z_2^{V_j'}} \Phi_{V_i',\psi}$ are both equal to $1$.

    Finally $(\sigma_i(x) - \sigma_j(x))^* (\sigma_i(x) -\sigma_j(x))$ is
    cyclically equivalent to 
    \begin{equation*}
        2 - 2 \sigma_i(x) \sigma_j(x) = 
            (1+\sigma_i(x)) (1-\sigma_j(x))
            + (1-\sigma_i(x)) (1+\sigma_j(x)),
    \end{equation*}
    so the result follows.
\end{proof}
\begin{defn}\label{def:inter}
    If $B = (X, \{(V_i,C_i)\}_{i=1}^m)$ is a BCS and $\pi$ is a
    probability distribution on $[m] \times [m]$, define $\mcA_{inter}(B,\pi)$
    to be the weighted algebra $(\mcA(B), \mu_{inter})$, where $\mu_{inter}$
    is defined from $\pi$ as in \Cref{prop:inter}.
\end{defn}
It is not hard to see that $\mcA(B) / \ang{\supp(\mu_{inter})} \iso \mcA(B) /
\ang{\supp(\mu_{\pi})}$, so both $\mcA(B,\pi)$ and $\mcA_{inter}(B,\pi)$ are
weighted algebra models of $\SynAlg(B,\pi)$. 

We can also easily handle transformations of constraint systems which apply a
homomorphism to each context. Note that a homomorphism $\sigma : \mcA(V,C) \to
\mcA(W,D)$ between finite abelian $C^*$-algebras is equivalent to a function $f
: D \to C$.  Indeed, given a function $f : D \to C$, we can define a
homomorphism $\sigma$ by $\sigma(\Phi_{V,\phi}) = \sum_{W,\psi \in
f^{-1}(\phi)} \Phi_{W,\psi}$, and it is not hard to see that all homomorphisms
have this form. We extend this notion to BCS algebras in the following way.
\begin{definition}\label{def:classhom}
    Let $B = (X, \{(V_i,C_i)\}_{i=1}^m)$ and $B' = (X',\{(W_i,D_i)\}_{i=1}^m)$ be constraint
    systems. A homomorphism $\sigma : \mcA(B) \to \mcA(B')$ is a \textbf{classical homomorphism}
    if 
	\begin{enumerate}
		\item $\sigma(\mcA(V_i,C_i))\subseteq \mcA(W_i,D_i)$ for all $1\leq i\leq m$, and
		\item if $\sigma(\Phi_{V_i,\phi_i}) = \sum_k\Phi_{W_i,\psi_{ik}}$,  $\sigma(\Phi_{V_j,\phi_j}) = \sum_k\Phi_{W_j,\psi_{jl}}$, and $\phi_i|_{V_i\cap V_j}\neq \phi_j|_{V_i\cap V_j}$ then $\psi_{ik}|_{W_i\cap W_j}\neq\psi_{jl}|_{W_i\cap W_j}$ for all $k,l$.
	\end{enumerate}
\end{definition}
To explain this definition, note that condition (1) implies that $\sigma$ restricts to 
a homomorphism $\mcA(V_i,C_i) \to \mcA(W_i,D_i)$, and hence gives a collection of functions
$f_i : D_i \to C_i$ for all $1 \leq i \leq m$. Condition (2) states that if $f_i(\phi)|_{V_i \cap V_j}
\neq f_j(\psi)|_{V_i \cap V_j}$ for some $\phi \in D_i$, $\psi \in D_j$, then $\phi|_{W_i \cap W_j}
\neq \psi|_{W_i \cap W_j}$. Conversely, any collection of functions $f_i : D_i \to C_i$ 
satisfying this condition can be turned into a classical homomorphism $\sigma : \mcA(B) \to \mcA(B')$. 

\begin{lemma}\label{lem:classhom}
    Let $B = \left(X,\{(V_i,C_i)\}_{i=1}^m\right)$ and $B' =
    \left(Y,\{(W_i,D_i)\}_{i=1}^m\right)$ be constraint systems, and let $\pi$
    be a probability distribution on $[m] \times [m]$. If $\sigma : \mcA(B) \to
    \mcA(B')$ is a classical homomorphism, then $\sigma$ is a $1$-homomorphism
    $\mcA(B,\pi) \to \mcA(B',\pi)$. 
\end{lemma}
\begin{proof}
    Suppose $\sigma$ arises from a family of functions $f_i : D_i \to C_i$ as above. 
    For any $1 \leq i,j \leq m$, let $R_{ij} = \{(\phi,\psi) \in C_i \times C_j
    : \phi|_{V_i \cap V_j} \neq \psi|_{V_i \cap V_j}\}$, and let $T_{ij} =
    \{(\phi,\psi) \in D_i \times D_j : \phi|_{W_i \cap W_j} \neq \psi|_{W_i
    \cap W_j}$.  Then
	\begin{align*}
		\sigma\left(\sum_{i,j} \sum_{(\phi,\psi) \in R_{ij}} \pi(i,j) \Phi_{V_i,\phi}\Phi_{V_j,\psi}\right) & = 
            \sum_{i,j} \sum_{\phi' \in f_i^{-1}(\phi), \psi' \in f_i^{-1}(\psi)} \pi(i,j) \Phi_{W_i,\phi'} \Phi_{W_j,\psi'} \\
        & \leq \sum_{i,j} \sum_{(\phi,\psi) \in T_{ij}} \pi(i,j) \Phi_{W_i,\phi} \Phi_{W_j,\psi}.
    \end{align*}
\end{proof}

One situation where we get a classical homomorphism is the following:
\begin{corollary}\label{cor:BCSto3SATsyst}
    Let $B = \left(X,\{(V_i,C_i)\}_{i=1}^m\right)$ be a BCS, and let $B' =
    \left(X',\{(W_i,D_i)\}_{i=1}^m\right)$ be a BCS with $X \subset X'$, $V_i
    \subseteq W_i$ for all $1 \leq i \leq m$, and $W_i \cap W_j = V_i \cap V_j$ for
    all $1 \leq i,j \leq m$. Suppose that for all $i \in [m]$, $\phi \in V_i$ if
    and only if there exists $\psi \in W_i$ with $\psi|_{V_i} = \phi$.  Then
    for any probability distribution $\pi$ on $[m] \times [m]$, the
    homomorphism
    \begin{equation*}
		\sigma:\mcA(B) \to \mcA(B') : \sigma_i(x) \mapsto \sigma_i(x) \text{ for } i \in [m], x \in V_i
	\end{equation*}
    defined by the inclusions $V_i \subseteq W_i$ is a $1$-homomorphism
    $\mcA(B,\pi) \to \mcA(B',\pi)$, and there is another $1$-homomorphism
    $\sigma' : \mcA(B',\pi) \to \mcA(B,\pi)$. Furthermore, $B'$ has the same
    connectivity as $B$. 
\end{corollary}
\begin{proof}
    The homomorphism $\sigma$ is the classical homomorphism defined by the 
    functions $D_i \to C_i : \psi \mapsto \psi|_{V_i}$. 

    For the homomorphism $\sigma'$, define $f_i : V_i \to W_i$ by choosing an
    element $f_i(\phi) \in W_i$ such that $f_i(\phi)|_{V_i} = \phi$ for all
    $\phi \in V_i$. Since $W_i \cap W_j = V_i \cap V_j$, if $f_i(\phi)|_{W_i
    \cap W_j} \neq f_j(\psi)|_{W_i \cap W_j}$, then $\phi|_{V_i \cap V_j} \neq
    \psi|_{V_i \cap V_j}$, so this collection of functions defines a 
    classical homomorphism $\mcA(B') \to \mcA(B)$. 
\end{proof}
In other words, \Cref{cor:BCSto3SATsyst} implies that any tracial state $\tau$
on $\mcA(B')$ (resp. $\mcA(B)$) with $\df(\tau) \leq \eps$ pulls back to a
tracial state on $\mcA(B)$ (resp. $\mcA(B')$) with defect also bounded by
$\eps$.

\begin{rmk}\label{rmk:to3SAT}
    Let $(\{\mcG(B_x,\pi_x)\}, S, C)$ be a $\BCS$-$\MIP^*$ protocol for a
    language $\mcL$ with soundness $s$, where $B_x =
    (X_x,\{(V^x_i,C_i^x)\}_{i=1}^{m_x})$. Since $|V^x_i|$ is polynomial in
    $|x|$, and $C$ runs in polynomial time, the Cook-Levin theorem implies that
    we can find sets $W^x_i$ and constraints $D_i^x$ on $W^x_i$ as in
    \Cref{cor:BCSto3SATsyst} in which $|W^x_i|$ is polynomial in $|x|$, and
    $D_i^x$ is a 3SAT instance with number of clauses polynomial in $|x|$. By
    \Cref{lem:Chom}, we get a $\BCS$-$\MIP^*$ protocol
    $(\{\mcG(B_x',\pi_x)\}, S, \wtd{C})$ for $\mcL$ with the same soundness,
    such that $B_x' = (X_x', \{(W_i^x,D_i^x)\})$ is a constraint system
    where all the clauses $D_i^x$ are 3SAT instances, and the connectivity
    of $B_x'$ is the same as $B_x$. 
\end{rmk}

\section{BCS algebras, subdivision and stability}\label{sec:stability}

Suppose we have a BCS where each constraint is made up of subconstraints on
subsets of the variables (for instance, a 3SAT instance made up of 3SAT
clauses). In this section, we look at what happens when we split up the
contexts and constraints so that each subconstraint is in its own contex.  In
the weighted BCS algebra, splitting up a context changes the commutative
subalgebra corresponding to the context to a non-commutative subalgebra.  To
deal with this, we use a tool from the approximate representation theory of
groups, namely the stability of $\Z_2^k$.
\begin{lemma}[\cite{chapman2023efficiently}] \label{lem:z2stab}
    Let $(\mcM,\tau)$ be a tracial von Neumann algebra, and suppose $f : [k]
    \to \mcM$ is a function such that $f(i)^2 = 1$ for all $i \in [k]$ and
    $\|[f(i),f(j)]\|_{\tau}^2 \leq \eps$ for all $i,j \in [k]$, where $k \geq 1$ and
    $\eps \geq 0$. Then there is a homomorphism $\psi : \Z_2^k \to
    \mcU(\mcM)$ such that $\|\psi(x_i) - f(i)\|_{\tau}^2 \leq \poly(k) \eps$
    for all $i \in [k]$, where the $x_i$ generate $\Z_2^k$.
\end{lemma}
Here a tracial von Neumann algebra is a von Neumann algebra $\mcM$ equipped with a faithful normal tracial state $\tau$, and $\mcU(\mcM)$ is the unitary group of $\mcM$. If $\tau$ is a tracial state on a $*$-algebra $\mcA$, and $(\rho : \mcA \to \mcB(\mcH), \ket{v})$ is the GNS representation, then the closure $\mcM = \overline{\rho(\mcA)}$ of $\rho(\mcA)$ in the weak operator topology is a von Neumann algebra, and $\tau_0(a) = \bra{v}a\ket{v}$ is a faithful normal tracial state on $\mcM$.
A function $f$ satisfying the conditions of \Cref{lem:z2stab} is called an
\textbf{$\eps$-homomorphism from $\Z_2^{k}$ to $\mcU(\mcM)$}.
The following lemma is useful for the proofs in this section:
\begin{lemma}\label{lem:hermitiansquare}
    Suppose $\mcA$ is a $*$-algebra, and let $h(a) := a^* a$ denote the 
    hermitian square of $a \in \mcA$. Then $h(\sum_{i=1}^n a_i) \leq k \sum_i h(a_i)$,
    where $k = 2^{\lceil \log_2 n \rceil}$.
\end{lemma}
\begin{proof}
    Since $h(a+b) + h(a-b) = 2h(a) + 2h(b)$, we see that $h(a+b) \leq 2h(a) + 2h(b)$. 
    Thus $h(\sum_{i=1}^n a_i) \leq 2 h(\sum_{i=1}^{\lfloor n/2 \rfloor} a_i)
        + 2 h(\sum_{i=\lfloor n/2 \rfloor+1}^n a_i)$, and repeated applications gives the
    desired inequality.
\end{proof}

We now formally define a subdivision of a BCS.
\begin{definition}
    Let $B = \left(X,\{(V_i,C_i)\}_{i=1}^m\right)$ be a BCS. Suppose that for
    all $1\leq i\leq m$ there exists a constant $m_i\geq1$ and a set of constraints
    $\{D_{ij}\}_{j=1}^{m_i}$ on variables $\{V_{ij}\}_{j=1}^{m_i}$ respectively, such
    that 
    \begin{enumerate}[(1)]
        \item $V_{ij} \subseteq V_i$ for all $i \in [m]$ and $j \in [m_i]$,
        \item for every $x,y \in V_i$ and $i \in [m]$, there is a $j \in [m_i]$ such
    that $x,y \in V_{ij}$, and 
        \item $C_i = \wedge_{j=1}^{m_i}D_{ij}$ for all $i \in [m]$, where $\wedge$ is
            conjunction.
    \end{enumerate}
    The BCS $B' = \left(X,\{V_{ij},D_{ij}\}_{i,j}\right)$ is called a
    \textbf{subdivision} of $B$. When working with subdivisions, we refer to 
    $D_{ij}$ as the \textbf{clauses} of constraint $C_i$, and $m_i$ as the
    \textbf{number of clauses} in constraint $i$.  A subdivision is
    \textbf{uniform} if $m_i = m_j$ for all $i,j$. 
\end{definition}
Given a subdivision of $B$ as in the definition, let $M = \sum_{i=1}^m m_i$,
and pick a bijection between $[M]$ and the set of pairs $(i,j)$ with $1 \leq i
\leq m$ and $1 \leq j \leq m_i$. If $\pi$ is a probability distribution on $[m]
\times [m]$, let $\pi_{sub}$ be the probability distribution on $[M] \times [M]$
with $\pi_{sub}(ij,k\ell) = \pi(i,k) / m_i m_k$. Note that if $\pi$ is uniform
and the subdivision is uniform, then $\pi_{sub}$ is uniform. Any subdivision
can be turned into a uniform subdivision by repeating pairs $(V_{ij},D_{ij})$
to increase $m_i$. Note that subdivision can increase connectivity.

One of the first things we notice about subdivision is that strategies for
$\mcG(B,\pi)$ can be lifted to strategies for the subdivided game. 
\begin{prop}\label{prop:subdiviso}
    Let $B = \left(X,\{(V_i,C_i)\}_{i=1}^m\right)$ be a BCS, and let $B' =
    \left(X,\{V_{ij},D_{ij}\}_{i,j}\right)$ be a subdivision. Let $\pi$ be a
    probability distribution on $[m] \times [m]$, and let $\pi_{sub}$ be the
    probability distribution defined from $\pi$ as above. The homomorphism
    $\alpha: \mcA(B') \to \mcA(B)$ defined by $\sigma_{ij}(x) \mapsto \sigma_i(x)$
    is a $1$-homomorphism $\mcA_{inter}(B',\pi_{sub}) \to \mcA_{inter}(B,\pi)$, and
    also induces an isomorphism $\SynAlg(B',\pi_{sub}) \iso \SynAlg(B,\pi)$.
\end{prop}
Here $\sigma_{ij}(x)$ denotes the copy of $x$ in $\mcA(W_{ij},D_{ij}) \subseteq \mcA(B')$. 
\begin{proof}
	Let $h(a) = a^*a$ denote the hermitian square of $a$ as in Lemma
	\ref{lem:hermitiansquare}. By definition, $\alpha(\sigma_{ij}(x)-\sigma_{kl}(x))= \sigma_i(x)-\sigma_k(x)$. Hence
	\begin{align*}
		\alpha\Bigl(\sum_{\substack{ij\neq kl\\ x \in V_{ij}\cap V_{kl}}}\pi_{sub}(ij,kl)h(\sigma_{ij}(x)-\sigma_{kl}(x))\Bigr) &= \sum_{\substack{ij\neq kl\\ x\in V_{ij}\cap V_{kl}}}\frac{\pi(i,k)}{m_i m_k}h(\sigma_i(x)-\sigma_k(x))
		\\ &\leq\sum_{\substack{i\neq k\\ x\in V_{i}\cap V_{k}}}\pi(i,k)h(\sigma_i(x)-\sigma_k(x)),
	\end{align*}
    since each variable $x \in V_i$ appears in at most $m_i$ subclauses $V_{ij}$. 
    Hence $\alpha : \mcA_{inter}(B',\pi_{sub}) \to \mcA_{inter}(B,\pi)$ is a $1$-homomorphism.
	
    To show that the synchronous algebras are isomorphic, observe that 
    since every pair of elements $x,y \in V_i$ belongs to some $V_{ij}$, there
    is an isomorphism 
    \begin{equation*}
        \SynAlg(B',\pi_{sub}) \iso *_{i=1}^m \Z_2^{V_i} / \langle R \rangle, 
    \end{equation*}
    where $R$ is the set of relations
    $\sigma_i(\Phi_{V_{ij},\phi}) \sigma_i(\Phi_{V_{k\ell},\psi}) = 0$ for all
    $\phi$ and $\psi$ which do not agree on $V_{ij} \cap V_{k\ell}$, and 
    $\sigma_i(\Phi_{V_{ij},\phi}) = 0$ for all $\phi\not\in D_{ij}$. From these
    latter relations, it is possible to recover the relations $\Phi_{V_i,\phi}
    = 0$ for $\phi \not\in C_i$, and then to recover all the relations of
    $\SynAlg(B,\pi)$. 
\end{proof}

\Cref{prop:subdiviso} implies that $\mcG(B,\pi)$ has a perfect quantum (resp.
commuting operator) strategy if and only if $\mcG(B',\pi_{sub})$ has a perfect
quantum (resp. commuting operator) strategy. The main result of this section is that near perfect strategies for $\mcG(B',\pi_{sub})$ can be pulled back to near perfect strategies for $\mcG(B,\pi)$. For the theorem, we say that $\pi$ is
\textbf{maximized on the diagonal} if $\pi(i,i) \geq \pi(i,j)$ and
$\pi(i,i) \geq \pi(j,i)$ for all $i,j \in [m]$.
\begin{theorem}\label{Thm:subdivSound}
	Let $B = \left(X,\{(V_i,C_i)\}_{i=1}^m\right)$ be a BCS, and let $B' =
	\left(X,\{V_{ij},D_{ij}\}_{i,j}\right)$ be a subdivision of $B$ with $m_i$ clauses in constraint $C_i$.
	Let $\pi$ be a
	probability distribution on $[m] \times [m]$ that is maximized on the diagonal, and let $\pi_{sub}$ be the
	probability distribution defined from $\pi$ as above. If there is a
	trace $\tau$ on $\mcA(B',\pi_{sub})$, then there is a trace $\wtd{\tau}$ on
	$\mcA(B,\pi)$ with $\df(\wtd{\tau})\leq \poly(m,2^C,M,K)\df(\tau)$, where $C=\max_{i,j}|V_{ij}|$, $K = \max_i|V_i|$, and $M = \max_i m_i$.
\end{theorem}
For the proof of the theorem we consider several other versions of the weighted
BCS algebra, where $\mcA(V_i,C_i)$ is replaced by $\C \Z_2^{*V_i}$, and the
defining relations of $\mcA(V_i,C_i)$ are moved into the weight function.
\begin{defn}\label{def:freeBCS}
    Let $B = \left(X,\{(V_i,C_i)\}_{i=1}^m\right)$ be a BCS with a probability
    distribution $\pi$ on $[m] \times [m]$, and let $B' =
    \left(X,\{V_{ij},D_{ij}\}_{i,j}\right)$ be a subdivision, with $m_i$ clauses
    in constraint $C_i$ and probability
    distribution $\pi_{sub}$ induced by $\pi$.
    Let $\sigma_i : \C \Z_2^{*V_i} \to *_{i=1}^m \C \Z_2^{*V_i}$ denote the
    inclusion of the $i$th factor. Let $\mcA_{free}(B) := *_{i=1}^m \C \Z_2^{*V_i}$, and
    define weight functions $\mu_{inter}$, $\mu_{sat}$, $\mu_{clause}$, and $\mu_{comm}$ on $\mcA_{free}(B)$ by
    \begin{align*}
        & \mu_{inter}(\sigma_i(x) - \sigma_j(x)) = \pi(i,j) \text{ for all } i \neq j \in [m] \text{ and } x \in V_i \cap V_j, 
            \\
        & \mu_{sat}(\Phi_{V_i,\phi}) = \pi(i,i) \text{ for all } i \in [m] \text{ and } \phi \in \Z_2^{V_{i}} \setminus C_{i},
        	\\
        & \mu_{clause}(\Phi_{V_{ij},\phi}) = \pi(i,i)/m_i^2 \text{ for all } (i,j) \in [m]\times[m_i] \text{ and } \phi \in \Z_2^{V_{ij}} \setminus D_{ij}, 
            \text{ and }\\ 
        & \mu_{comm}([\sigma_i(x),\sigma_i(y)]) = \pi(i,i) \text{ for all } i \in [m] \text{ and } x,y \in V_i,
    \end{align*}
    and $\mu_{inter}(r) = 0$, $\mu_{sat}(r)= 0$, $\mu_{clause}(r)=0$, and $\mu_{comm}(r)=0$ for any
    elements $r$ other than those listed. 
    Let $\mcA_{free}(B,B',\pi)$ be the weighted algebra $(\mcA_{free}(B), \mu_{all})$,
    where $\mu_{all} := \mu_{inter} + \mu_{clause} + \mu_{comm}$.
\end{defn}
Note that $\mu_{inter}$ is the same as the weight function of the algebra
$\mcA_{inter}(B,\pi)$ defined in \Cref{def:inter}, except that it's defined on
$\mcA_{free}(B)$ rather than $\mcA(B)$.  The weight function $\mu_{sat}$ comes
from the defining relations for $\mcA(B)$, while $\mu_{clause}$ comes from the
defining relations for $\mcA(B')$, so $\mcA_{free}(B,B',\pi)$ is a mix of relations
from $\mcA_{inter}(B,\pi)$ and $\mcA_{inter}(B',\pi)$. As mentioned previously,
the context $V_i$ has an order inherited from $X$, and this is used for the
order of the product when talking about $\Phi_{V_i,\phi}$ and
$\Phi_{V_{ij},\phi}$ in $\mcA_{free}(B)$. In particular, the order on $V_{ij}$
is compatible with the order on $V_{i}$.  

The weight functions $\mu_{inter}$, $\mu_{sat}$ and $\mu_{clause}$ can also be
defined on $\ast_{i=1}^m \C \Z_2^{V_i}$ using the same formula as in
\Cref{def:freeBCS}, and we use the same notation for both versions. The
following lemma shows that we can relax $\mcA_{inter}(B,\pi)$ to $(\ast_{i=1}^m
\C \Z_{2}^{V_i}, \mu_{inter} + \mu_{clause})$, as long as $\pi$ is maximized on the diagonal.
\begin{lemma}\label{lem:correctsat}
    Let $B = \left(X,\{(V_i,C_i)\}_{i=1}^m\right)$ be a BCS, and let $\pi$ be a
    probability distribution on $[m]\times[m]$ which is maximized on the diagonal.
    Suppose $\mu_{inter}$, $\mu_{sat}$ and $\mu_{clause}$ are the weight functions defined above with respect to $\pi$. Then there is an
    $O(t)$-homomorphism $\mcA_{inter}(B,\pi) \to (\ast_{i=1}^m \C \Z_2^{V_i}, \mu_{inter}+\mu_{sat})$,
    where $t$ is the connectivity of $B$. 
    Furthermore, if $B' = \left(X,\{V_{ij},D_{ij}\}_{i,j}\right)$ is a subdivision of $B$, then
    there is an $M^2$-homomorphism $(\ast_{i=1}^m \C \Z_2^{V_i}, \mu_{inter}+\mu_{sat}) \to (\ast_{i=1}^m \C \Z_2^{V_i}, \mu_{inter}+\mu_{clause})$, where $M = \max_i m_i$ is the maximum number of clauses
    $m_i$ in constraint $i$.
\end{lemma}
\begin{proof}
    Since $C_i$ is non-empty by convention, we can choose $\psi_i\in C_i$ for every $1\leq
    i\leq m$. Define the homomorphism $\alpha:\mcA_{inter}(B,\pi) \to (\ast_{i=1}^m
    \C \Z_2^{V_i}, \mu_{inter}+\mu_{sat})$ by 
	\begin{equation*}
		\alpha(\sigma_i(x)) = \sum_{\varphi\in C_i}\Phi_{V_i,\varphi}\sigma_i(x)+\sum_{\mathclap{\varphi\in \Z^{V_i}_2\setminus C_i}}\Phi_{V_i,\varphi}\psi_i(x).
	\end{equation*}
    Let $\Phi_i = \sum_{\varphi\in C_i}\Phi_{V_i,\varphi}$, and let $h(a) =
    a^*a$ denote the hermitian square of $a$ as in Lemma \ref{lem:hermitiansquare}. Then
	\begin{equation*}
		\begin{split}
			\alpha\left[h(\sigma_i(x)-\sigma_j(x))\right]&=h(\Phi_i\sigma_i(x)+(1-\Phi_i)\psi_i(x)-\Phi_j\sigma_j(x)-(1-\Phi_j)\psi_j(x))
			\\
			&\leq4h[\Phi_i\sigma_i(x)+(1-\Phi_i)\psi_i(x)-\sigma_i(x)]
			\\
			&\quad +4h[\Phi_j\sigma_j(x)+(1-\Phi_j)\psi_j(x)-\sigma_j(x)]+4h[\sigma_i(x)-\sigma_j(x)].
		\end{split}
	\end{equation*}
	Observe that $\sigma_i(x) = \sum_{\varphi\in \Z_2^{V_i}}\Phi_{V_i,\varphi}\varphi(x)$, so
	\begin{equation*}
        h(\Phi_i\sigma_i(x)+(1-\Phi_i)\psi_i(x)-\sigma_i(x)) =
            \sum_{\mathclap{\varphi\in \Z_2^{V_i}\setminus C_i}}\Phi_{V_i,\varphi}(\psi_i(x) 
            - \varphi(x))^2\leq 4\sum_{\mathclap{\varphi\in \Z_2^{V_i}\setminus C_i}}\Phi_{V_i,\varphi}. 
	\end{equation*}
	Thus 
    {
        \small
	\begin{align*}
        \alpha\Bigl(\sum_{\substack{1\leq i \neq j\leq m\\x\in V_i\cap V_j}} \pi(i,j)h(\sigma_i(x)-\sigma_j(x))\Bigr)&
            \leq \sum_{\substack{1\leq i \neq j\leq m\\x\in V_i\cap V_j}} \pi(i,j)\Bigl(16\sum_{\mathclap{\varphi\in \Z_2^{V_i}\setminus C_i}}\Phi_{V_i,\varphi} + 16\sum_{\mathclap{\varphi\in \Z_2^{V_j}\setminus C_j}}\Phi_{V_j,\varphi}+4h(\sigma_i(x)-\sigma_j(x))\Bigr)
			\\
        & \leq \sum_{a \in \ast_{i=1}^m \C \Z_2^{V_i}} 4 \mu_{inter}(a) a^* a + \sum_{a \in \ast_{i=1}^m \C \Z_2^{V_i}} 32 t \mu_{sat}(a) a^* a \\
        &\leq O(t)\sum_{a\in \ast_{i=1}^m \C \Z_2^{V_i}}(\mu_{inter}(a)+\mu_{sat}(a))a^*a,
	\end{align*}
    }
	since $\pi$ is maximized on the diagonal.

    Next, suppose $B'$ is a subdivision of $B$. If $\phi \in \Z_2^{V_i}
    \setminus C_i$, then we can choose $j_{\phi} \in [m_i]$ such that
    $\phi|_{V_{i{j_\phi}}} \not\in D_{i{j_\phi}}$. Since $\displaystyle\sum_{\phi :
    \phi|_{V_{ij}} = \phi'} \Phi_{V_i,\phi} = \Phi_{V_{ij},\phi'}$,
    \begin{equation*}
        \sum_{\phi \not\in C_i} \Phi_{V_i,\phi} 
        = \sum_{1 \leq j \leq m_i} \sum_{\phi : j_{\phi} = j} \Phi_{V_i,\phi}
        \leq \sum_{1 \leq j \leq m_i} \sum_{\phi : \phi|_{V_{ij}} \not\in D_{ij}} \Phi_{V_i,\phi}
        = \sum_{1 \leq j \leq m_i} \sum_{\phi' \not\in D_{ij}} \Phi_{V_{ij},\phi'}.
    \end{equation*}
    Hence
    \begin{equation*}
        \sum_{r} \mu_{sat}(r) r^* r \leq M^2 \sum_{r} \mu_{clause}(r) r^* r,
    \end{equation*}
    where the $M^2$ comes from the fact that we divide by $m_i^2$ in the definition of $\mu_{clause}$.
    Thus the identity map $(\ast_{i=1}^m \C \Z_2^{V_i},\mu_{inter}+\mu_{sat})
    \to (\ast_{i=1}^m \C \Z_2^{V_i}, \mu_{inter}+\mu_{clause})$ is an $M^2$-homomorphism.
\end{proof}
The following proposition shows how to construct tracial states on $\mcA_{inter}(B,\pi)$ from tracial states on $\mcA_{free}(B,B',\pi)$.
\begin{prop}\label{prop:correction}
    Let $B = \left(X,\{(V_i,C_i)\}_{i=1}^m\right)$ be a BCS, and let $\pi$ be a
    probability distribution on $[m] \times [m]$ which is maximized on the
    diagonal. Let $B' = \left(X,\{V_{ij},D_{ij}\}_{i,j}\right)$ be a subdivision of $B$ with
    $m_i$ clauses in constraint $C_i$. 
    If $\tau$ is a trace on $\mcA_{free}(B,B',\pi)$, then there is a
    trace $\wtd{\tau}$ on $\mcA_{inter}(B,\pi)$ such that $\df(\wtd{\tau}) \leq
    \poly(m,2^C,M,K)\df(\tau)$, where $C = \max_{ij}|V_{ij}|$, $K = \max_{i}|V_i|$, and $M = \max_{i} m_i$. Furthermore, if $\tau$ is finite-dimensional then so is $\wtd{\tau}$. 
\end{prop}
\begin{proof}
    Since $\pi$ is maximized on the diagonal, if $\pi(i,i) = 0$ then $\pi(i,j)
    = \pi(j,i) = 0$ for all $j \in [m]$, and the variables in $V_i$ do not appear
    in $\supp(\mu_{inter})$. Thus we may assume without loss of generality that
    $\pi(i,i) > 0$ for all $i \in [m]$. Let $\tau$ be a trace on
    $\mcA_{free}(B,B',\pi)$. By the GNS construction there is
    a $*$-representation $\rho$ of $\mcA_{free}(B,B',\pi)$ acting on a Hilbert
    space $\mcH_0$ with a unit cyclic vector $\psi$ such that $\tau(a) =
    \langle\psi|\rho(a)|\psi\rangle$ for all $a \in \mcA_{free}(B)$. Let
    $\mcM_0 = \overline{\rho(\mcA_{free}(B))}$ be the weak operator closure of the image of $\rho$, and let $\tau_0$
    be the faithful normal tracial state on $\mcM_0$ corresponding to
    $\ket{\psi}$ (so $\tau_0 \circ \rho = \tau)$. 

	 For all $i \in [m]$ the restriction of $\rho$ to $\mbZ_2^{*V_i}$ is a $\df(\tau;\mu_{comm})/\pi(i,i)$-homomorphism from $\mbZ_2^{V_i}$ into $(\mcM_0,\tau_0)$, so by \Cref{lem:z2stab} there is a representation $\rho_i: \mbZ_2^{V_i}\rightarrow \mcU(\mcM_0)$ such that 
	 \begin{equation}\label{eq:stability}
        \|\rho_i(x_j)-\rho(x_j)\|_{\tau_0}^2\leq \dfrac{\poly(K)}{\pi(i,i)}\df(\tau;\mu_{comm})
	\end{equation}
	for all generators $x_j \in \mbZ_2^{V_i}$. Suppose $x \in V_i \cap V_j$, and let $\wtd{\rho} : \ast_{i=1}^m \C \Z_2^{V_i} \to \mcM_0$ be
    the homomorphism defined by $\wtd{\rho}(x) = \rho_i(x)$ for $x\in \mbZ_2^{V_i}$. Then
    \begin{align*}
        \|\wtd{\rho}(\sigma_i(x)-\sigma_j(x))\|_{\tau_0}^2
        & \leq 4 \|\wtd{\rho}(\sigma_i(x)) - \rho(\sigma_i(x))\|_{\tau_0}^2
            + 4 \|\wtd{\rho}(\sigma_j(x)) - \rho(\sigma_j(x))\|_{\tau_0}^2 \\
            & \quad\quad\quad + 4 \|\rho(\sigma_i(x)-\sigma_j(x))\|_{\tau_0}^2 \\
        & \leq \dfrac{\poly(K)}{\pi(i,i)}\df(\tau;\mu_{comm}) + 4 \|\sigma_i(x) - \sigma_j(x)\|_{\tau}^2.
    \end{align*}
	Since $\pi$ is maximalized on the diagonal, and $|\{(i,j,x) : i \neq j \in [n], x \in V_i\cap V_j\}|\leq mt$ where $t$ is the connectivity of $B$, we conclude that
    \begin{align*}
        \df(\tau_0 \circ \wtd{\rho}; \mu_{inter}) & \leq \sum_{i\neq j} \sum_{x \in V_i \cap V_j} \pi(i,j)\left( \dfrac{\poly(K)}{\pi(i,i)}\df(\tau;\mu_{comm})
        + 4 \|\sigma_i(x) - \sigma_j(x)\|_{\tau}^2\right) \\ 
        & \leq  O(mt\poly(K) \df(\tau;\mu_{comm}) + \df(\tau;\mu_{inter})).
    \end{align*}

    For any $S\subseteq V_i$, let $x_S := \prod_{x \in S} x \in \Z_2^{*V_i}$, where the
    order of the product is inherited from the order on $X$. By \Cref{eq:stability},
    \begin{equation*}
    	\|\wtd{\rho}(x_S)-\rho(x_S)\|_{\tau_0}^2 \leq \dfrac{\poly(K)}{\pi(i,i)}\df(\tau;\mu_{comm}),
    \end{equation*}
    where the degree of $K$ has increased by one. Since $\Phi_{V_{ij},\phi} =
    \tfrac{1}{2^{|V_{ij}|}} \sum_{S \subseteq V_{ij}} \phi(x_S)x_S$, we get that
  \begin{equation*}
  	\|\wtd{\rho}(\Phi_{V_{ij},\phi}) - \rho(\Phi_{V_{ij},\phi})\|_{\tau_0}^2
  	\leq \frac{1}{2^{|V_{ij}|}} \sum_{S \subseteq V_{ij}} \|\wtd{\rho}(x_S) - \rho(x_S)\|_{\tau_0}^2
  	\leq \dfrac{\poly(K)}{\pi(i,i)}\df(\tau;\mu_{comm}).
  \end{equation*}
    If $1 \leq i \leq m$, $1 \leq j \leq m_i$, and $\phi \not\in D_{ij}$, then 
    \begin{equation*}
    	\|\wtd{\rho}(\Phi_{V_{ij},\phi})\|_{\tau_0}^2 \leq
    	2 \|\wtd{\rho}(\Phi_{V_{ij},\phi}) - \rho(\Phi_{V_{ij},\phi})\|_{\tau_0}^2 + 
    	2 \|\rho(\Phi_{V_{ij},\phi})\|_{\tau_0}^2, 
    \end{equation*}
    and hence 
    \begin{align*}
        \df(\tau_0 \circ \wtd{\rho}; \mu_{clause}) & = \sum_{i,j} \frac{\pi(i,i)}{m_i^2}
            \sum_{\phi \not\in D_{ij}} \|\wtd{\rho}(\Phi_{V_{ij},\phi})\|_{\tau_0}^2 \\
        & \leq \sum_{i,j} \sum_{\phi \not\in D_{ij}} \frac{\pi(i,i)}{m_i^2} \left( \dfrac{\poly(K)}{\pi(i,i)}\df(\tau;\mu_{comm}) + 2 \|\Phi_{V_{ij},\phi}\|^2_{\tau}\right) \\
        & \leq \sum_{i,j} 2^C \frac{\poly(K)}{m_i^2} \df(\tau;\mu_{comm})+ 2 \df(\tau; \mu_{clause})   \\
        & \leq m^2 2^C \poly(K)\df(\tau;\mu_{comm}) + 2 \df(\tau;\mu_{clause}).
    \end{align*}
    We conclude that $\wtd{\tau} = \tau_0 \circ
    \wtd{\rho}$ is a tracial state on $\ast_{i=1}^m \C \Z_2^{V_i}$ with $\df(\wtd{\tau};\mu_{inter}+\mu_{clause})$
    bounded by 
    \begin{align*}
         O(\df(\tau;\mu_{inter}) + \df(\tau;\mu_{clause}) + (m^2 2^C+ mt)\poly(K) \df(\tau;\mu_{comm})).
    \end{align*}
    Since $t \leq O(mK)$, we conclude that 
    \begin{equation*}
        \df(\wtd{\tau};\mu_{inter}+\mu_{clause}) \leq \poly(m,2^C,K) \df(\tau; \mu_{inter}+\mu_{clause}+\mu_{comm}).
    \end{equation*}
     By \Cref{lem:correctsat}, there is a
    $O(tM^2)$-homomorphism $\mcA_{inter}(B,\pi) \to
    (\ast^m_{i=1}\C\Z^{V_i}_2,\mu_{inter}+\mu_{clause})$, and pulling
    $\wtd{\tau}$ back by this homomorphism gives the proposition.
\end{proof}
Finally, we can pull back tracial states from the subdivision algebra
$\mcA_{inter}(B',\pi_{sub})$ to traces on $\mcA_{free}(B,B',\pi)$.
\begin{prop}\label{prop:subdivhom}
    Let $B = \left(X,\{(V_i,C_i)\}_{i=1}^m\right)$ be a BCS, and let $B' =
    \left(X,\{V_{ij},D_{ij}\}_{i,j}\right)$ be a subdivision of $B$.
    Let $\pi$ be a
    probability distribution on $[m] \times [m]$, and let $\pi_{sub}$ be the
    probability distribution defined from $\pi$ as above. Then there is a $\poly(M,2^{C})$-homomorphism $\mcA_{free}(B,B',\pi) \to \mcA_{inter}(B',\pi_{sub})$,
    where $C = \max_{ij}|V_{ij}|$ and $M = \max_{i}m_i$.
\end{prop}
\begin{proof}
    For each $1 \leq i \leq m$ and $x \in V_i$, choose an index $1 \leq r_{ix} \leq m_i$
    such that $x \in V_{ir_{ix}}$. Also, for each $x,y \in V_i$, choose an index $i_{xy}$
    such that $x,y \in V_{i_{xy}}$. Define $\alpha : *_{i=1}^m \Z_2^{*V_i} \to \mcA(B')$ by
    $\alpha(\sigma_i(x)) = \sigma_{ir_{ix}}(x)$. It follows immediately from
    the definitions that $\alpha$ is a $O(M^2)$-homomorphism $(\mcA_{free}(B), \mu_{inter})
    \to \mcA_{inter}(B',\pi_{sub})$. 
    Moving on to $\mu_{comm}$, observe that if $h(a) = a^*a$ as in \Cref{lem:hermitiansquare} then
    \begin{align*}
    	\alpha(h([\sigma_i(x),\sigma_i(y)])) &= h\left(\sigma_{ir_{ix}}(x)\sigma_{ir_{iy}}(y)-\sigma_{ir_{iy}}(y)\sigma_{ir_{ix}}(x)\right)
    	\\
    	&\leq  
    	4h((\sigma_{ir_{ix}}(x)-\sigma_{i_{xy}}(x))\sigma_{ir_{iy}}(y))+4h(\sigma_{i_{xy}}(x)(\sigma_{i_{xy}}(y)-\sigma_{ir_{iy}}(y)))+
    	\\
    	&\quad +4h((\sigma_{ir_{iy}}(y)-\sigma_{i_{xy}}(y))\sigma_{ir_{ix}}(x))+4h(\sigma_{i_{xy}}(y)(\sigma_{ir_{ix}}(x)-\sigma_{i_{xy}}(x)))
    	\\
    	&\lesssim 8h(\sigma_{ir_{ix}}(x)-\sigma_{i_{xy}}(x))+8h(\sigma_{ir_{iy}}(y)-\sigma_{i_{xy}}(y)),
    \end{align*}
    where we use the fact that $[\sigma_{i_{xy}}(x),\sigma_{i_{xy}}(y)]=0$, and that $U^*a^*aU$ is cyclically equivalent to $a^*a$ if $UU^* = 1$. For any given $x \in V_i$
    and $1 \leq j \leq m_i$, the number of elements $y \in V_i$ with $i_{xy} = j$ is bounded by
    $|V_{ij}|$. Hence
    \begin{equation*}
        \sum_{i} \sum_{x,y \in V_i} \pi(i,i) \alpha(h([\sigma_i(x),\sigma_i(y)]))
            \lesssim O(C M^2) \sum_{i,j,j'} \sum_{x \in V_{ij} \cap V_{ij'}} \frac{\pi(i,i)}{m_i^2} h(\sigma_{ij}(x) - \sigma_{ij'}(x)),
    \end{equation*}
    where $\sigma_{ij} : \C \Z_2^{V_{ij}} \to \mcA(B')$ is the inclusion of the $ij$th factor.
    We conclude that there is an $O(CM^2)$-homomorphism
    $(\mcA_{free}(B),\mu_{comm}) \to \mcA_{inter}(B',\pi_{sub})$.

    Finally, for $\mu_{clause}$, if $i\in [m]$, $j\in [m_i]$, and $\phi \not\in D_{ij}$ then $\sigma_{ij}(\Phi_{V_{ij},\phi}) = 0$, so 
    \begin{align*}
        \alpha(\Phi_{V_{ij},\phi}) & = \alpha(\Phi_{V_{ij},\phi}) -  \sigma_{ij}(\Phi_{V_{ij},\phi})
         \\
         & = \frac{1}{2^{|V_{ij}|}} \sum_{S \subseteq V_{ij}}
            \prod_{x \in S} \phi(x) \sigma_{ir_{ix}}(x) - \frac{1}{2^{|V_{ij}|}} \sum_{S \subseteq V_{ij}}
            \prod_{x \in S} \phi(x) \sigma_{ij}(x)
            \\
        & = \frac{1}{2^{|V_{ij}|}} \sum_{S \subseteq V_{ij}} \sum_{x \in S} u_{x,S} \phi(x)
            (\sigma_{ir_{ix}}(x)  - \sigma_{ij}(x)) v_{x,S},
    \end{align*}
    where $u_{x,S}$ is the product of $\phi(y) \sigma_{ij}(y)$ for $y \in S$ appearing
    before $x$ in the order on $V_i$, and $v_{x,S}$ is the product of $\phi(y) \sigma_{ir_{iy}}(y)$
    for $y \in S$ appearing after $x$ in the order on $V_i$. Since there are less than
    $|V_{ij}| \cdot 2^{|V_{ij}|}$ terms in this sum, and $\phi(x) u_{x,S}$ and $v_{x,S}$ are unitary,
    \begin{align*}
        h(\alpha(\Phi_{V_{ij},\phi})) & \lesssim \frac{2 |V_{ij}|}{2^{|V_{ij}|}} \sum_{S \subseteq V_{ij}}
        \sum_{x \in S} h(\sigma_{ir_{ix}}(x) - \sigma_{ij}(x)) \\
            & = \frac{|V_{ij}|}{2^{|V_{ij}|-1}} \sum_{x \in V_{ij}} \sum_{x \in S \subseteq V_{ij}} 
                h(\sigma_{ir_{ix}}(x) - \sigma_{ij}(x)) 
                \\
            & = |V_{ij}| \sum_{x \in V_{ij}} h(\sigma_{ir_{ix}}(x) - \sigma_{ij}(x)).
    \end{align*}
    Hence
    \begin{align*}
        \sum_{i\in[m],j\in[m_i]} \frac{\pi(i,i)}{m_i^2} \sum_{\phi \not\in D_{ij}} \alpha(h(\Phi_{V_{ij},\phi}))
            & \lesssim \sum_{i,j} \frac{\pi(i,i)}{m_i^2} \sum_{\phi \not\in D_{ij}} C \sum_{x \in V_{ij}}
                h(\sigma_{ir_{ix}(x)} - \sigma_{ij}(x)) \\
            & \leq C \cdot 2^C \sum_{i,j} \frac{\pi(i,i)}{m_i^2} \sum_{x \in V_{ij}} h(\sigma_{ir_{ix}} - \sigma_{ij}(x)).
    \end{align*}
    Since every term in the latter sum occurs in the sum $\sum_{r} \mu'(r) r^*
    r$ for the weight function $\mu'$ of $\mcA_{inter}(B',\pi_{sub})$, $\alpha$
    is a $C \cdot 2^C$-homomorphism $(\mcA_{free}(B),\mu_{clause}) \to
    \mcA_{inter}(B',\pi_{sub})$. We conclude that $\alpha$ is an $O(M^2 + CM^2 + C 2^C)$-homomorphism
    $\mcA_{free}(B,\pi) \to \mcA_{inter}(B',\pi_{sub})$, and $O(M^2 + CM^2 + C 2^C) \leq \poly(M,2^C)$.
\end{proof}

\begin{proof}[proof of \Cref{Thm:subdivSound}]
 Applying \Cref{prop:subdivhom} and \Cref{prop:correction} yields the result.
\end{proof}

\section{Parallel repetition}\label{sec:prep}

Let $\mcG = (I, \{O_i\}_{i \in I}, \pi, V)$ be a nonlocal game. The
\textbf{$n$-fold parallel repetition} of $\mcG$ is the game 
\begin{equation*}
    \mcG^{\otimes n} = (I^n, \{O_{\ul{i}}\}_{\ul{i} \in I^n}, \pi^{\otimes n}, V^{\otimes n}),
\end{equation*}
where
\begin{enumerate}
    \item $I^n$ is the $n$-fold product of $I$, 
    \item if $\ul{i} \in I^n$, then $O_{\ul{i}} := O_{i_1} \times O_{i_2} \times \cdots \times O_{i_n}$,
    \item if $\ul{i},\ul{j} \in I^n$, then $\pi^{\otimes n}(\ul{i},\ul{j}) = \prod_{k=1}^n \pi(i_k,j_k)$, and
    \item if $\ul{i},\ul{j} \in I^n$, $\ul{a} \in O_{\ul{i}}$, $\ul{b} \in O_{\ul{j}}$, then
        $V^{\otimes n}(\ul{a},\ul{b}|\ul{i},\ul{j}) = \prod_{k=1}^n V(a_k,b_k|i_k,j_k)$.
\end{enumerate} 
In other words, the players each receive a vector of questions $\ul{i} =
(i_1,\ldots,i_n)$ and $\ul{j} = (j_1,\ldots,j_n)$ from $\mcG$, and must reply
with a vector of answers $(a_1,\ldots,a_n)$ and $(b_1,\ldots,b_n)$ to each
question. Each pair of questions $(i_k,j_k)$, $1 \leq k \leq n$ is sampled
independently from $\pi$, and the players win if and only if $(a_k,b_k)$ is a
winning answer to questions $(i_k,j_k)$ for all $1 \leq k \leq n$. If $\mcG$
has questions of length $q$ and answers of length $a$, then $\mcG^{\otimes n}$
has questions of length $nq$ and answers of length $na$.

If $p$ is a correlation for $\mcG$, let $p^{\otimes n}$ be the correlation for
$\mcG^{\otimes n}$ defined by
\begin{equation*}
    p^{\otimes n}(\ul{a},\ul{b}|\ul{i},\ul{j}) = \prod_{k=1}^n p(a_k,b_k|i_k,j_k).
\end{equation*}
It is easy to see that $p^{\otimes n}$ is a quantum (resp. commuting operator)
correlation if and only if $p$ is a quantum (resp. commuting operator) correlation, and
that $\omega(\mcG^{\otimes n};p^{\otimes n}) = \omega(\mcG,p)^n$. Hence if $\omega_q(\mcG)=1$
(resp. $\omega_{qc}(\mcG)=1$) then $\omega_q(\mcG^{\otimes n})=1$ (resp.
$\omega_{qc}(\mcG^{\otimes n}) = 1$) as well. If $\omega_q(\mcG) < 1$,
then $\omega_q(\mcG^{\otimes n}) \geq \omega_q(\mcG)^n$ (and the same
for the commuting operator value), but this inequality is not always tight.
However, Yuen's parallel repetition theorem states that the game value goes
down at least polynomially in $n$:
\begin{thm}[\cite{yuen2016parallel}]\label{thm:prep}
    For any nonlocal game $\mcG$, if $\delta = 1 - \omega_q(\mcG) > 0$, then
    $\omega_q(\mcG^{\otimes n}) \leq b/ \poly(\delta,n)$, where $b$ is the length
    of the answers of $\mcG$.
\end{thm}

Suppose $B = (X,\{(V_i,C_i)\}_{i=1}^m)$ is a BCS and that $\pi$ is a probability
distribution on $[m] \times [m]$. For any $n \geq 1$, let $X^{(n)} := X \times
[n]$, and $V^{(k)}_i = V_i \times \{k\} \subseteq X^{(n)}$. We can think of
$X^{(n)}$ as the disjoint union of $n$ copies of $X$, and $V^{(k)}_i$ as the
copy of $V_i$ from the $k^{th}$ copy of $X$. Since $V^{(k)}_i$ is a copy of
$V_i$, we can identify $\Z_2^{V_i^{(k)}}$ with $\Z_2^{V_i}$ in the natural way. If $\ul{i} \in [m]^n$, let $V_{\ul{i}} = \cup_{j = 1}^k V^{(k)}_{i_j}$ and $C_{\ul{i}} = C_{i_1} \times
\cdots \times C_{i_k} \subseteq \Z_2^{V_{\ul{i}}} = \Z_2^{V^{(1)}_{i_1}} \times
\cdots \times \Z_2^{V_{i_k}^{(k)}}$. Let $B^{(n)} := (X^{(n)},
\{(V_{\ul{i}},C_{\ul{i}})_{\ul{i} \in [m]^n}\})$. Given a distribution $\pi$ on
$[m] \times [m]$, consider the game $\mcG(B^{(n)},\pi^{\otimes n})$, where
$\pi^{\otimes n}$ is the product distribution as above. In this game, the players are
given questions $\ul{i}$ and $\ul{j}$ from $[m]^n$ respectively, and must reply
with elements $\ul{\phi} \in C_{\ul{i}}$ and $\ul{\psi} \in C_{\ul{j}}$
respectively. They win if and only if $\ul{\phi}$ and $\ul{\psi}$ agree on
$V_{\ul{i}} \cap V_{\ul{j}} = \bigcup_{k=1}^n V^{(k)}_{i_k} \cap
V^{(k)}_{j_k}$.  But this happens if and only if $\phi_k$ and $\psi_k$ agree on
$V_{i_k} \cap V_{j_k}$. Thus $\mcG(B^{(n)}, \pi^{\otimes n})$ is the
parallel repetition $\mcG(B,\pi)^{\otimes n}$. We record this in the following
lemma:
\begin{lemma}\label{lem:BCSParallel}
    If $\mcG$ is a BCS game, then so is the parallel repetition $\mcG^{\otimes n}$.
\end{lemma}

To illustrate the purpose of parallel repetition, suppose that $(\{\mcG_x\}, S, V)$ is a $\MIP^*(2,1,1,s)$-protocol for
a language $\mcL$, where $\mcG_x = (I_x, \{O_{xi}\}, \pi_x, V_x)$ and has
answer length $a_x$. If $n_x$ is a polynomial in $|x|$, then $\pi_x^{\otimes
n_x}$ can be sampled in polynomial time by running $S$ independently $n$ times,
and $V_x^{\otimes n_x}$ can also be computed in polynomial time by running $V$
repeatedly. If $S^{\otimes n_x}$ and $V^{\otimes n_x}$ are these Turing
machines for sampling $\pi_x^{\otimes n_x}$ and computing $V_x^{\otimes n_x}$
respectively, then $(\{\mcG_x^{\otimes n_x}\}, S^{\otimes n_x}, V^{\otimes
n_x})$ is a $\MIP^*(2,1,1,s')$-protocol for $\mcL$, where $s' = a_x /
\poly(1-s) \cdot \poly(n_x)$. Since $a_x$ is polynomial in $|x|$, if $1-s =
1/\poly(|x|)$, then we can choose $n_x$ such that $s'$ is any constant $< 1$. By \Cref{lem:BCSParallel} the same can be done for $\BCS$-$\MIP^*$.

\section{Perfect zero knowledge}\label{sec:pzk}
An $\MIP$ protocol is perfect zero knowledge if the verifier gains no new
information from interacting with the provers.  If the players' behaviour in a
game $\mcG = (I,\{O_i\}_{i \in I},\pi,V)$ is given by the correlation $p$, then
what the verifier (or any outside observer) sees is the distribution
$\{\pi(i,j)p(a,b|i,j)\}$ over tuples $(a,b|i,j)$.  Consequently a
$\MIP^*$-protocol $(\{\mcG_x\}, S, V)$ is said to be perfect zero-knowledge
against an honest verifier if the players can use correlations $p_x$ for
$\mcG_x$ such that the distribution $\{\pi(i,j)p(a,b|i,j) \}$ can be sampled in
polynomial time in $|x|$. However, a dishonest verifier seeking to get more
information from the players might sample the questions from a different
distribution $\pi'$ from $\pi$. To be perfect zero-knowledge against a
dishonest verifier, it must be possible to efficiently sample
$\{\pi'(i,j)p_x(a,b|i,j)\}$ for any efficiently sampleable distribution $\pi'$,
and this is equivalent to being able to efficiently sample from
$\{p_x(a,b|i,j)\}_{(a,b) \in O_i \times O_j}$ for any $i,j$. This leads to the
definition (following \cite[Definition 6.3]{coudron2019complexity}):

\begin{definition}\label{def:PZK}
    Let $\mcP = (\{\mcG_x\}, S, V)$ be a two-prover one-round $\MIP^*$ protocol for a
    language $\mcL$ with completeness $c$ and soundness $s$, where $\mcG_x =
    (I_x,\{O_{xi}\}, \pi_x, V_x)$. The protocol $\mcP$ is \textbf{perfect zero
    knowledge} if for every string $x$, there is
    a correlation $p_x$ for $\mcG_x$ such
    that
	\begin{enumerate}
		\item for all $i,j \in I_x$, the distribution $\{p_x(a,b|i,j)\}$ can be sampled in polynomial time in $|x|$, and
		\item if $x\in \mcL$ then $p_x \in C_{qa}$ and $\omega(\mcG_x,p_x)=1$.
    \end{enumerate}
	The class $\PZK$-$\MIP^*(2,1,c,s)$ is the class of languages with a perfect zero knowledge
    two-prover one round $\MIP^*$ protocol with completeness $c$ and soundness $s$.
\end{definition}
By replacing $C_{qa}$ with $C_{qc}$, we get another class
$\PZK$-$\MIP^{co}$. If we replace $\MIP^*$ protocols with $\BCS$-$\MIP^*$ (resp. $\BCS$-$\MIP^{co}$) protocols and $C_{qa}$ with $C^s_{qa}$ (resp. $C^s_{qc}$) we get the class $\PZK$-$\BCS$-$\MIP^*$ (resp. $\PZK$-$\BCS$-$\MIP^{co}$).

For the one-round protocols that we are considering, parallel repetition
preserves the property of being perfect zero knowledge. 
\begin{proposition}
    Let $(\{\mcG_x\},S,V)$ be a $\PZK$-$\MIP^*(2,1,1,s)$ protocol, and let $n_x$ be a 
    polynomial function of $|x|$. 
    Then the parallel repeated protocol $(\{\mcG_x^{\otimes n_x}\},S^{\otimes n_x}, V^{\otimes n_x})$
    is also perfect zero knowledge. 
\end{proposition}
\begin{proof}
    Let $p_x$ be a correlation for the game $\mcG$ that satisfies the two
    requirements of \Cref{def:PZK}. Then
    $\{p_x^{\otimes n_x}(\ul{a},\ul{b}|\ul{i},\ul{j})\}_{\ul{a},\ul{b}}$ can be
    sampled in polynomial time in $|x|$ for all $\ul{i}$, $\ul{j}$ by independently sampling from $\{p_x(a,b,i_{\ell},j_{\ell})\}_{a,b}$ for each pair $(i_{\ell},j_{\ell})$ from $\ul{i} = (i_{1},\dots,i_{n_{x}})$ and $\ul{j} = (j_{1},\dots,j_{n_{x}})$. If $x \in \mcL$, then
    $\omega(\mcG_x^{\otimes n_x}; p_x^{\otimes n}) = 1$, and it is not hard to see
    that $p_x^{\otimes n_x} \in C_{qa}$.
\end{proof}
We will now prove our main result that any proof system in
$\BCS$-$\MIP^*$ or $\BCS$-$\MIP^{co}$ can be turned into a perfect zero
knowledge $\BCS$-$\MIP^*$ or $\BCS$-$\MIP^{co}$ protocol. For this purpose, we
use the perfect zero knowledge proof system for 3SAT due to Dwork, Feige,
Kilian, Naor, and Safra \cite{Dwork1992LowC2}, slightly modified for the proof
of quantum soundness.  For the construction, we assume that we start with a
$\BCS$-$\MIP^*$ protocol (and in the proof of \Cref{thm:main}, this will be a
3SAT-$\MIP^*$ protocol).  Following \cite{Dwork1992LowC2}, the new proof system
is constructed in three steps. First, we apply a transformation called
oblivation, then turn the resulting system into a permutation branching program
via Barrington's theorem \cite{Barrington86}, and finally rewrite the
permutation branching programs using the randomizing tableaux of Kilian
\cite{Kilian1990UsesOR}.  We start by describing obliviation.
\begin{definition}\label{def:obliviation}
    Given a BCS $B = (X,\{(V_i,C_i)\}^m_{i=1})$ and $n \geq 1$, let $Z =
    X\times [n]$, and $U_i = V_i\times [n]$ for any $1\leq
    i\leq m$. To make the elements of $Z$ look more like variables, we denote
    $(x,i)$ by $x(i)$. Let $E_i \subseteq \Z_2^{U_i}$ be the set of
    assignments $\phi$ to $U_i$ such that the assignment $\psi$ to $V_i$
    defined by $\psi(x) = \psi(x(1))\cdots \psi(x(n))$ is in $C_i$. The
    \textbf{obliviation of $B$ of degree $n$} is the constraint system
    $\Obl_n(B) = (Z,\{(U_i,E_i)\}_{i=1}^m)$. 
\end{definition}
The point of obliviation is the following:
\begin{lemma}\label{lem:obliviate}
    Suppose $B = (X,\{(V_i,C_i)\}^m_{i=1})$ is a BCS, and let $B' = \Obl_n(B)$
    for some $n \geq 1$, using the notation from \Cref{def:obliviation}. Then:
    \begin{enumerate}[(a)]
        \item There is a classical homomorphism $\alpha : \mcA(B) \to \mcA(B')$
            such that $\alpha(\sigma_i(x)) = \sigma_i(x(1)\cdots x(n))$ for all
            $i \in [m]$ and $x \in V_i$, where $\sigma_i$ is the inclusion of
            the $i$th factor for $\mcA(B)$ and $\mcA(B')$. 
    
        \item Let $\Gamma$ be the set of sequences $x_1,\ldots,x_k$ in $Z$ of
            length $1 \leq k \leq n-1$, such that there is some $i \in [k]$
            with $x_i \neq x_j$ for all $j \in [k] \setminus \{i\}$.  If $\pi$
            is a probability distribution on $[m] \times [m]$, and $\tau$ is a
            tracial state on $\mcA(B)$, then there
            is a tracial state $\wtd{\tau}$ on $\mcA(B')$ such that $\tau =
            \wtd{\tau} \circ \alpha$, $\df(\wtd{\tau};\mu_\pi) = \df(\tau;\mu_\pi)$, and
            $\wtd{\tau}( \sigma_{i_1}(x_1) \cdots \sigma_{i_k}(x_k)) = 0$ for
            all sequences $x_1,\ldots,x_k$ in $\Gamma$ and indices $i_1,\ldots,i_k \in [m]$
            such that $x_j \in U_{i_j}$ for all $1 \leq j \leq k$. If $\tau$
            is finite-dimensional (resp. Connes-embeddable), then $\wtd{\tau}$
            is also finite-dimensional (resp. Connes-embeddable). 

        \item For any $1 \leq i \leq m$, the set $\{ \prod_{x \in S} x 
                : S \subseteq U_i, |S| < n/2\}$ of monomials in $U_i$ of
                degree less than $n/2$ is linearly independent in $\mcA(U_i,E_i)$. 
        \end{enumerate}
\end{lemma}
In particular, if $\tau$ is perfect then $\wtd{\tau}$ is perfect. 
\begin{proof}
    Define $f_i : \Z_2^{U_i} \to \Z_2^{V_i}$ for each $i \in [m]$ by
    $f_i(\phi)(x) = \phi(x(1)) \cdots \phi(x(n))$ for $\phi \in \Z_2^{U_i}$
    and $x \in V_i$. By definition, $\phi \in E_i$ if and only
    if $f_i(\phi) \in C_i$, so $f_i(E_i) = C_i$. If $f_i(\phi)(x) \neq 
    f_j(\psi)(x)$ for some $\phi \in \Z_2^{U_i}$, $\psi \in \Z_2^{U_j}$, 
    and $x \in V_i \cap V_j$, then we must have $\phi(x(i)) \neq \psi(x(i))$
    for some $i$. Since
    \begin{equation*}
        \sigma_i(x(1)\cdots x(n)) = \sum_{\phi \in \Z_2^{U_i}} f_i(\phi)(x)  \Phi_{U_i,\phi}
    \end{equation*}
    for all $x \in V_i$, $i \in [m]$, the functions $f_i$ correspond to a
    classical homomorphism $\alpha : \mcA(B) \to \mcA(B')$ with
    $\alpha(\sigma_i(x)) = \sigma_i(x(1)\cdots x(n))$ for all $i \in [m]$ and
    $x \in V_i$. This proves part (a).

    Conversely, given $y \in \Z_2^{X \times [n-1]}$ and $\phi \in \Z_2^{V_i}$,
    define $\phi_y \in \Z_2^{U_i}$ by $\phi_y(x(1)) = \phi(x) y(x,1)$,
    $\phi_y(x(j)) = y(x,j-1) y(x,j)$ for $2 \leq j \leq n-1$, and $\phi_y(x(n))
    = y(x,n-1)$. Since $f_i(\phi_y) = \phi$, the function $\phi \mapsto \phi_y$ sends $C_i$
    to $E_i$. Also if $\phi \in \Z_2^{V_i}$ and $\psi \in \Z_2^{V_j}$, then
    $\phi_y |_{U_i \cap U_j} \neq \psi_y |_{U_i \cap U_j}$ if and only if
    $\phi|_{V_i \cap V_j} \neq \psi |_{V_i \cap V_j}$, so the functions
    $\phi \mapsto \phi_y$ determine a classical homomorphism $\beta_y : \mcA(B')
    \to \mcA(B)$ with $\beta_y(\sigma_i(x(1))) = \sigma_i(x) y(x,1)$, $\beta_y(\sigma_i(x(j)))
    = y(x,j-1) y(x,j)$ for $2 \leq j \leq n-1$, and $\beta_y(\sigma_i(x(n))) = y(x,n-1)$
    for all $i \in [m]$ and $x \in V_i$. 

    Given a tracial state $\tau$ on $\mcA(B)$, define a tracial state $\wtd{\tau}$ on $\mcA(B')$ by 
        $\wtd{\tau} = 2^{-|X|(n-1)} \sum_{y} \tau \circ \beta_y$, where the sum is over all
    $y \in \Z_2^{X \times [n-1]}$. Notice that if $\tau$ is finite-dimensional
    (resp.  Connes-embeddable), then $\wtd{\tau}$ is also finite-dimensional (resp.
    Connes-embeddable).
    Since $\beta_y \circ \alpha$ is the identity on $\mcA(B)$, $\wtd{\tau} \circ \alpha = \tau$.
    Since $\beta_y$ and $\alpha$ are $1$-homomorphisms, 
    \begin{equation*}
        \df(\tau \circ \beta_y; \mu_{\pi}) \leq \df(\tau;\mu_{\pi}) = \df(\tau \circ \beta_y \circ \alpha;\mu_{\pi})
        \leq \df(\tau \circ \beta_y; \mu_{\pi})
    \end{equation*}
    for any $y$, so $\df(\tau \circ \beta_y; \mu_{\pi}) = \df(\tau;\mu_{\pi})$
    and hence $\df(\wtd{\tau}; \mu_{\pi}) = \df(\tau;\mu_{\pi})$.  
    
    Finally, if $x_1,\ldots,x_k$ is a sequence in $Z$, and
    $i_1,\ldots,i_k$ is a sequence in $[m]$ such that $x_j \in U_{i_j}$, then 
    there is an element $a \in \mcA(B)$ and set $S \subseteq X \times [n-1]$
    such that
    \begin{equation*}
        \beta_y(\sigma_{i_1}(x_1) \cdots \sigma_{i_k}(x_k)) = m_y \tau(a)
    \end{equation*}
    for all $y \in \Z_2^{X \times [n-1]}$, where $m_y := \prod_{(x,j) \in S} y(x,j)$.
    If $x_1,\ldots,x_k$ is in $\Gamma$, then $S$ is non-empty, and $\sum_{y} m_y = 0$. Hence
    \begin{equation*}
        \wtd{\tau} (\sigma_{i_1}(x_1) \cdots \sigma_{i_k}(x_k)) = 2^{-|X|(n-1)} \sum_{y} m_y \tau(a) = 0.
    \end{equation*}
    This proves part (b).

    For part (c), pick a tracial state $\tau$ on the finite-dimensional
    $C^*$-algebra $\mcA(V_i,C_i)$ (since $C_i$ is non-empty, this algebra is
    non-trivial). As in the proof of part (b), we can define a tracial state
    $\wtd{\tau} = 2^{|X|(n-1)} \sum_{y} \tau \circ \beta_y$ on $\mcA(U_i,E_i)$
    with the property that $\wtd{\tau}(x_1 \cdots x_k) = 0$ if $1 \leq k \leq
    n-1$ and $x_1,\ldots,x_k \in U_i$ are
    distinct. If $S, T \subseteq U_i$, then
    \begin{equation*}
        \prod_{x \in S} x \cdot \prod_{x \in T} x = \prod_{x \in S \Delta T} x,
    \end{equation*}
    where $S \Delta T := (S \cup T) \setminus (S \cap T)$. If $|S|, |T| < n/2$, then
    $|S \Delta T| < n$, and $S \Delta T = \emptyset$ if and only if $S = T$. Hence
    by part (b), 
    \begin{equation*}
        \wtd{\tau}(\prod_{x \in S} x \cdot \prod_{x \in T} x) 
            = \begin{cases} 1 & S = T \\
                            0 & S \neq T 
                \end{cases}.
    \end{equation*}
    It follows that the monomials $\{\prod_{x \in S} x : S \subseteq U_i, |S| < n/2\}$
    are linearly independent. 
\end{proof}

A \textbf{permutation branching
program} of width $5$ and depth $d$ on a set of variables $X$  is a tuple $P =
(X, \{(x_i,\pi_{1}^{(i)},\pi_{-1}^{(i)})\}_{i=1}^{d}, \sigma)$ where 
$x_i \in X$ and $\pi_1^{(i)}, \pi_{-1}^{(i)}$ are elements of the
permutation group $S_5$ for all $1 \leq i \leq d$, and $\sigma \in S_5$ is a
5-cycle. A permutation branching program $P$ defines a map $P : \Z_2^{X} \to
S_5$ via $P(\phi) = \prod_{i=1}^d \pi^{(i)}_{\phi(x_i)}$. A program $P$
\textbf{recognizes a constraint $C \subseteq \Z_2^{X}$} if $P(\phi) = \sigma$
for all $\phi \in C$, and $P(\phi) = e$ for all $\phi \not\in C$, where $e$ is
the identity in $S_5$.
\begin{theorem}[Barrington \cite{Barrington86}]
    Suppose a constraint $C \subseteq \Z_2^X$ is recognized by a depth $d$ fan-in 2
    boolean circuit. Then $C$ is recognized by a permutation branching program
    of depth $4^d$ on the variables $X$. 
\end{theorem}
For the rest of the section, we assume that we have a canonical way of turning
constraints described by fan-in 2 boolean circuits into permutation branching
programs using Barrington's theorem. 

The final ingredient is randomizing tableaux, which are described using
constraints of the form $x_1 \cdots x_n = \gamma$, where the variables
$x_1,\ldots,x_n$ take values in $S_5$, $\gamma$ is a constant in $S_5$, and the
product is the group multiplication. Since $|S_5|=120 < 2^7$, we can encode
permutations as bit strings of length $7$ by choosing an enumeration
$S_5=\{e=\gamma_0, \ldots,\gamma_{119}\}$, and identifying $\gamma_j$ by its
index $j$ in binary. This means that any permutation-valued variable can be
represented by $7$ boolean variables, and similarly a permutation-valued
constraint $x_1 \cdots x_n = \gamma$ can be rewritten as the constraint on $7n$
boolean variables which requires the boolean variables corresponding to $x_i$
to encode a permutation value, and the product of all the permutations to be
equal to $\gamma$. Since we want our final output to be a boolean constraint system,
we use permutation-valued variables and permutation-valued constraints as
short-hand for boolean constraint systems constructed in this way.  We can now
define randomizing tableaux, still following \cite{Dwork1992LowC2} with small
modifications.
\begin{definition}\label{def:Tab}
    Let $B = (X,\{(V_i,C_i)\}_{i=1}^m)$ be a BCS, where each $C_i$ is described
    by a fan-in 2 boolean circuit. Let $P_i = (V_i,\{(x_{ij},\pi_{1}^{(ij)},
    \pi_{-1}^{(ij)})\}_{j=1}^{d_i},\sigma_i)$ be the permutation branching program
    recognizing $C_i$. For each $i \in [m]$, let
	\begin{equation*}
		W_i = V_{i}\sqcup\{T_{i}(p,q) : (p,q) \in [4] \times [d_i]\} \sqcup \{r_{i}(j,k): (j,k) \in [3] \times [d_i-1]\},
    \end{equation*}
    where $T_i(p,q)$ and $r_i(j,k)$ are new permutation-valued variables (and thus
    represent 7 boolean variables each), and let
	\begin{equation*}
        Y = X\sqcup\{T_{i}(p,q),r_{i}(j,k):(i,p,q,j,k)\in{[m]\times[4]\times[d_i]\times[3]\times[d_i-1]}\}
    \end{equation*}
    be the union of all the original and new variables. The variables $T_i(p,q)$
    are called tableau elements, and the variables $r_i(j,k)$ are called randomizers.

    Let $D_i$ be the constraint on variables $W_i$ which is the conjunction of the
    following clauses:
	\begin{enumerate}
		\item $T_{i}(1,q) = \pi^{(iq)}_{x_q}$ for all $q \in [d_i]$,
		\item $T_i(p+1,q) = r_i(p,q-1)^{-1}T_i(p,q)r_i(p,q)$ for $q \in [d_i]$ and $p \in [3]$, where
            we use the notation $r_i(p,0) = r_i(p,d_i) = e$, 
		\item $\prod_{1\leq q\leq d_i}T_i(4,q) = \sigma_i$, and
        \item a trivial constraint (meaning that all assignment are allowed) on any pair $x,y$
            of original or permutation-valued variables which do not appear in one of the above
            constraints.
	\end{enumerate}
    The \textbf{tableau} of $B$ is $\Tab(B) = (Y,\{(W_i,D_i)\}_{i=1}^m)$, interpreted as a 
    boolean constraint system. We further let $\{W_{ij},D_{ij})\}_{j=1}^{m_i}$
    be a list of the clauses in (1)-(4) making up $D_i$. The \textbf{subdivided
    tableau} of $B$ is $\Tab_{sub}(B) = (Y,\{(W_{ij},D_{ij})\}_{i\in [m], j\in
    [m_i]})$.
\end{definition}
Compared to \cite{Dwork1992LowC2}, we've added the trivial constraints (4), as
well as an extra row of the tableau.  As mentioned above, the product in the
constraints on the permutation-valued variables in parts (1)-(4) of the
definition is the group product in $S_5$. The constraints in part (1) involve
both original variables $x_q$ and permutation-valued variables $T_i(1,q)$, and
say that the value of $T_i(1,q)$ is either $\pi^{(iq)}_{1}$ or
$\pi^{(iq)}_{-1}$ depending on the value of $x_q$.  In part (4), $x$ and $y$
can be either an original or a permutation-valued variable. If one of them is a
permutation-valued variable, then all the corresponding boolean variables
encoding the permutation-valued variable are included in the constraint (so the
constraint on $x$ and $y$ may involve up to $14$ boolean variables). Since the
constraints in part (4) are trivial, they do not contribute to $D_i$, but they
are included in the list of clauses $(W_{ij},D_{ij})$ of the subdivided
tableau. The point of the constraints in part (4) is that, with them,
$\Tab_{sub}(B)$ is a subdivision of $\Tab(B)$. The extra row of the tableau is
needed to compensate for the inclusion of these constraints in $\Tab_{sub}(B)$
(see \Cref{rmk:prooffails}).
As in \cite{Dwork1992LowC2}, the constraints $D_i$ encode the constraints $C_i$
as follows:
\begin{lemma}[\cite{Dwork1992LowC2}]\label{lem:classicalRT}
    Suppose $B = (X,\{(V_i,C_i)\}_{i=1}^m)$ is a BCS, and let $\Tab(B) =
    (Y,\{(W_i,D_i)\}_{i=1}^m)$.  If $\psi \in D_i$, then $\psi|_{V_i} \in C_i$.
    Conversely, if $r \in S_5^{R_i}$, where $R_i = \{r_i(j,k) : (j,k) \in [3]
    \times [d_i]\}$ is the set of randomizers in $W_i$, and $\phi \in C_i$,
    then there is a unique element $\phi_r \in D_i$ such that $\phi_r|_{V_i}
    = \phi$ and $\phi_r|_{R_{i}} = r$.
\end{lemma}
In this lemma, the statement that $\phi_r|_{R_i} = r$ means that for every
randomizer $r_i(j,k) \in R_i$, the restriction of $\phi$ to the boolean
variables corresponding to $r_i(j,k)$ is the encoding of the permutation
$r(r_i(j,k))$.
\begin{proof}
    If $\psi\in D_i$, then by constraint (2), $\prod_q T_i(p+1,q) =
    \prod_q T_i(p,q)$. Since $\prod_q T_i(4,q) = \sigma_i$ by constraint (3),
    $\prod_q \pi_{x_q}^{(iq)} = \sigma_i$. Since the permutation branching
    program $P_i$ recognizes $C_i$, we conclude that $\psi|_{V_i}\in C_i$. 
	
    Conversely, given an assignment $r \in S_5^{R_i}$ to the variables $R_i$
    and $\phi \in  C_i$, we can set $T_i(1,q) = \pi_{\phi(x_q)}^{(iq)}$ and
    $T_i(p+1,q) = r_i(p,q-1)^{-1}T_i(p,q)r_i(p,q)$ to get an assignment where
    $\prod_q T_i(4,q) = \sigma_i$.
\end{proof}
Although the permutation-valued variables in $\Tab(B)$ are
shorthand for boolean variables, it is helpful to be able to work with the
permutation-valued variables directly in $\mcA(\Tab(B))$. Suppose for a moment
that $x_1,\ldots,x_7$ are variables in a set $V$, and $C$ is a constraint on
$V$ which includes the requirement that $x_1,\ldots,x_7$ encode a
permutation-valued variable $x$. Let $S = \{x_1,\ldots,x_7\}$. If $\phi \in
\Z_2^{S}$, then $\Phi_{S,\phi} = 0$ in $\mcA(V,C)$ unless $\phi$ is the binary
representation of an index $0 \leq j < 120$, in which case we also write
$\Phi_{S,\phi}$ as $\Phi_{S,j}$. Hence the subalgebra of $\mcA(V,C)$ is
generated by the single unitary $\sum_{j=0}^{119} e^{2 \pi i j / 120}
\Phi_{S,j}$, which we denote by the same symbol as the permutation-valued
variable $x$.  In particular, if $B = (X,\{(V_i,C_i)\}_{i=1}^m)$ and $\Tab(B) =
(Y,\{(W_i,D_i)\}_{i=1}^m)$ as in \Cref{def:Tab}, then we can refer to
$T_i(p,q)$ and $r_i(j,k)$ as unitary elements of $\mcA(W_i,D_i)$ of order
$120$, and they generate the same subalgebra as the boolean variables encoding
them.  Since these variables do not occur in any other context $W_j$ for $j
\neq i$, we also use $T_i(p,q)$ and $r_i(j,k)$ to refer to $\sigma_i(T_i(p,q))$
and $\sigma_i(r_i(j,k))$ in $\mcA(\Tab(B))$. We use the same convention 
for $\mcA(W_{i\ell}, D_{i\ell})$, although since the variables $T_i(p,q)$ and
$r_i(j,k)$ occur in more than one constraint of $\Tab_{sub}(B)$, we are
stuck with the notation $\sigma_{i\ell}(T_i(p,q))$ and
$\sigma_{i\ell}(r_i(j,k))$ when refering to these variables in
$\mcA(\Tab_{sub}(B))$. With these conventions, we can state the following
noncommutative version of \Cref{lem:classicalRT}.
\begin{lemma}\label{lem:RT}
    Suppose that $B = (X,\{(V_i,C_i)\}_{i=1}^m)$ is a BCS, and let $\Tab(B) =
    (Y,\{(W_i,D_i)\}_{i=1}^m)$. Let $R_i = \{r_i(j,k) : (j,k) \in [3] \times
    [d_i-1]\}$ be the set of randomizers in $W_i$, and let $R = \bigcup_i R_i$.
    \begin{enumerate}[(a)]
        \item The natural map 
            \begin{equation*}
                \mcA(V_i,C_i) \otimes \C \Z_{120}^{R_i} \to \mcA(W_i,D_i) : 
                    x_i \mapsto x_i, r_i(j,k) \mapsto r_i(j,k)
            \end{equation*}
            is an isomorphism. In particular, $\mcA(W_i,D_i)$ is generated as an
            algebra by $V_i \cup R_i$, and $\mcA(\Tab(B))$ is generated by $\bigcup_i
            \{\sigma_i(x) : x \in V_i\} \cup R$. 

        \item The natural inclusion $\alpha : \mcA(B) \to \mcA(\Tab(B))$ defined by
             $\alpha(\sigma_i(x)) = \sigma_i(x)$ for $i \in [m]$ and $x \in V_i$
            is a classical homomorphism.

        \item If $r \in S_5^{R}$, then
            there is a classical homomorphism $\beta_r : \mcA(\Tab(B)) \to \mcA(B)$
            such that for all $i \in [m]$, if $x \in V_i$ then $\beta_r(\sigma_i(x)) 
             = \sigma_i(x)$, and if $x \in R_i$ then $\beta_r(x) = e^{2 \pi i j / 120}$
            where $r(x) = \gamma_j$ in the enumeration of $S_5$ fixed above.
    
        \item Let $\mcM$ be the set of monomials in $\mcA(B)$ of the form $u \sigma_i(z)^a v$, where
            $z \in \mcR_i$ for some $i \in [m]$, $1 \leq a < 120$, and $u$ and $v$ are monomials
            in $\{\sigma_j(x) : j \in [m], x \in V_j \cup R_j\}$ which do not contain $z$.
            If $\pi$ is a probability distribution on $[m] \times [m]$, and $\tau$
            is a tracial state on $\mcA(B)$, then there is a tracial state
            $\wtd{\tau}$ on $\mcA(\Tab(B))$ such that $\tau = \wtd{\tau} \circ
            \alpha$, where $\alpha$ is the classical homomorphism from part
            (b), $\df(\wtd{\tau}; \mu_{\pi}) = \df(\tau; \mu_{\pi})$, and 
            $\wtd{\tau}(y) = 0$ for all $y \in \mcM$. Furthermore, if $\tau$ is
            finite-dimensional (resp. Connes-embeddable), then $\wtd{\tau}$ is
            also finite-dimensional (resp. Connes-embeddable). 
    \end{enumerate}
\end{lemma}
\begin{proof}
    For part (a), the algebra  $\C \Z_{120}^{R_i}$ has a basis consisting of the joint spectral projections
    \begin{equation*}
    	\Phi_{R_i,r} = \prod_{x\in R_i}\lambda(r(x))^{-1}\prod_{\gamma_k\neq r(x)}(x-e^{2\pi ik/120}), \quad r \in \Z_{120}^{R_i},
    \end{equation*}
    where $\lambda(\gamma_j) = \prod_{k\neq j}(e^{2\pi ij/120}-e^{2\pi ik/120})$.
    Hence $\mcA(V_i,C_i) \otimes \C \Z_{120}^{R_i}$ has a
    basis consisting of the elements $\Phi_{V_i,\phi} \otimes \Phi_{R_i,r}$
    for $\phi \in C_i$ and $r \in \Z_{120}^{R_i}$. Using the enumeration
    of $S_5$ fixed earlier, we can interpret $\Z_{120}^{R_i}$ as the
    set $S_5^{R_i}$ of permutation-valued assignments to $R_i$.
    The natural homomorphism $\mcA(V_i,C_i) \otimes \C \Z_{120}^{R_i}
    \to \mcA(W_i,D_i)$ sends $\Phi_{V_i,\phi} \otimes \Phi_{R_i,r}$
    to $\sum_{\psi} \Phi_{W_i,\psi}$, where the sum is across all 
    $\psi \in D_i$ such that $\psi|_{V_i} = \phi$ and $\psi|_{R_i} = r$.
    By \Cref{lem:classicalRT}, the restriction map $\phi \mapsto 
    \phi|_{V_i \cup R_i}$ is a bijection between $D_i$ and $C_i \times
    S_5^{R_i}$, so this homomorphism is an isomorphism.

    Parts (b) and part (c) follow immediately from \Cref{lem:classicalRT} and
    the definition of a classical homomorphism. Alternatively, part (b) also follows from
    \Cref{cor:BCSto3SATsyst}.

    The proof of part (d) is similar to the proof of \Cref{lem:obliviate}, part (b). Given
    a tracial state $\tau$ on $\mcA(B)$, let $\wtd{\tau}$ be the tracial state on $\mcA(\Tab(B))$
    defined by $\wtd{\tau} = \tfrac{1}{120^{|R|}} \sum_{r} \tau \circ \beta_r$, where the 
    sum is over $r \in S_5^R$. If $\tau$ is finite-dimensional (resp.
    Connes-embeddable), then $\wtd{\tau}$ is finite-dimensional (resp.
    Connes-embeddable). Since $\beta_r(\sigma_i(x)) = \sigma_i(x)$ for all $i
    \in [m]$ and $x \in V_i$, $\beta_r \circ \alpha$ is the identity on
    $\mcA(B)$, and $\wtd{\tau} \circ \alpha = \tau$. By parts (b) and (c), 
    $\df(\tau \circ \beta_r) \leq \df(\tau) =
    \df(\tau \circ \beta_r \circ \alpha) \leq \df(\tau \circ \beta_r)$.
    This means that $\df(\tau \circ \beta_r) = \df(\tau)$, so $\df(\tau) = \df(\wtd{\tau})$. 
    Finally, suppose $y \in \mcM$, so $y = u \sigma_i(z)^a v$ for some $z \in R_i$, $1 \leq a < 120$, and
    monomials $u$, $v$ which do not contain $z$. By part (c), there is some monomial $y'$ in $\{\sigma_j(x) : j \in [m],
    x \in V_i\}$ such that for all $r \in S_5^R$, we have $\beta_r(y) = e^{2 \pi aij / 120} c_{r'} y'$,
    where $r(z) = \gamma_j$, and $c_{r'} \in \C$ depends only on $r' = r|_{R \setminus \{z\}}$. 
    Hence 
    \begin{equation*}
        \wtd{\tau}(y) = \tfrac{1}{120^{|R|}} \sum_{r \in S_5} \tau(\beta_r(y))
            = \tfrac{1}{120^{|R|}} \sum_{j=0}^{120} e^{2 \pi a ij / 120} \sum_{r' \in S_5^{R\setminus \{z\}}} c_{r'} \tau(y')
            = 0,
    \end{equation*}
    finishing the proof of part (d).
\end{proof}

We need one more general fact about permutation-valued variables. 
\begin{lemma}\label{lem:nolinearterm}
    Let $f : S_5^m \to S_5$ be a function, and suppose $(V,C)$ is a 
    boolean constraint encoding the constraint $x = f(y_1,\ldots,y_m)$ on
    permutation-valued variables $x,y_1,\ldots,y_m$. If $1 \leq n < 120$, then
    \begin{equation*}
        x^n = \sum_{a} c_a y_1^{a_1} \cdots y_m^{a_m}
    \end{equation*}
    for some coefficients $c_a \in \C$, where the sum is over all integer vectors $a =
    (a_1,\ldots,a_m)$ with $0 \leq a_1,\ldots,a_m < 120$. 
    Furthermore, if for every $\pi_1,\ldots,\pi_{m-1} \in S_5$, the set
    $\{f(\pi_1,\ldots,\pi_{k-1},\pi,\pi_{k},\ldots,\pi_{m-1}) :
    \pi \in S_5\}$ is equal to $S_5$, then $c_{a} = 0$ if $a_k = 0$.
\end{lemma}
\begin{proof}
    Let $Y_k$ be the set of boolean variables representing $y_k$, and
    let $X$ be the set of boolean variables representing $x$. 
    The constraint $x = f(y_1,\ldots,y_m)$ states that 
    \begin{equation*}
        \Phi_{X,\ell} = \sum_{(\gamma_{j_1},\ldots,\gamma_{j_m}) \in f^{-1}(\gamma_{\ell})} \Phi_{Y_1,j_1}
            \cdots \Phi_{Y_m,j_m},
    \end{equation*}
    where $\{\gamma_0,\ldots,\gamma_{119}\}$ is our chosen enumeration of
    $S_5$. Since $\Phi_{Y_k,j_k}$ is a polynomial in $y_k$, and $x^m$ is a
    linear combination of the projections  $\Phi_{X,\ell}$ for $0 \leq \ell <
    120$, we get $x^n = g(y_1,\ldots,y_n)$, where $g = \sum_{a} c_a y_1^{a_1} \cdots
    y_m^{a_m}$ is a polynomial in $y_1,\ldots,y_m$. Since $y_k^{120}=1$, we can
    further assume that the sum is over vectors $a = (a_1,\ldots,a_k)$ with $0
    \leq a_k < 120$ for all $k$.

    Given $0 \leq j_1,\ldots,j_m < 120$, let $\phi_j : \mcA(V,C) \to \C$ be the
    homomorphism sending $\Phi_{Y_k,a} \mapsto \delta_{aj_k}$ for all $1 \leq k
    \leq k$. This homomorphism sends $y_k \mapsto \omega^{j_k}$ and $x \mapsto
    \omega^{\ell}$, where $\omega = e^{2 \pi i / 120}$, and $\gamma_{\ell} = 
    f(\gamma_{j_1},\ldots,\gamma_{j_m})$. We use the notation 
    \begin{equation*}
        A_1,\ldots,\widecheck{A}_{k},\ldots,A_m
    \end{equation*}
    to denote the list $A_1,\ldots,A_m$ with the element $A_k$ omitted. 
    If, for some $k$, we fix $0 \leq j_1,\ldots,\widecheck{j}_{k},\ldots,j_m < 120$, then
    \begin{align*}
        \sum_{0 \leq j_k < 120} \phi_j(g) = \sum_a c_a \prod_{t \neq k} \omega^{j_t a_t} 
                                            \sum_{0 \leq j_k < 120} \omega^{j_k a_k} = h(\omega^{j_1},\ldots,\widecheck{\omega}^{j_{k}},
        \ldots,\omega^{j_m}),
    \end{align*}
    where $h = g(y_1,\ldots,y_{k-1},0,y_{k+1},\ldots,y_m)$. If
    $\{f(\gamma_{j_1},\ldots,\gamma_{j_m}) : 0 \leq j_k < 120\}$ is equal to
    $S_5$, then 
    \begin{equation*}
            \sum_{0 \leq j_k < 120} \phi_j(x^n) = \sum_{0 \leq \ell < 120} \omega^{n \ell} = 0
    \end{equation*}
    for $1\leq n < 120$, and we conclude that
    \begin{equation*}
        h(\omega^{j_1},\ldots,\widecheck{\omega}^{j_{k}}, \ldots,\omega^{j_m}) = 0.
    \end{equation*}
    If this occurs for all choices of $0 \leq j_1,\ldots,\widecheck{j}_k,\ldots,j_m < 120$,
    then $h$ must be the zero polynomial, so $c_a = 0$ if $a_{k} = 0$.
\end{proof}
Although \Cref{lem:nolinearterm} is stated for general functions $f$, we are only going
to use it for the group multiplication and inverse functions, i.e. $f(y_1,y_2) = y_1 y_2$
and $f(y) = y^{-1}$. For these functions, the additional hypothesis on $f$ holds for all
indices $k$. Thus the lemma states that if $(V,C)$ encodes the constraint $x = y_1 y_2$,
then $x$ is a polynomial in $y_1$ and $y_2$ such that all monomials contain both $y_1$
and $y_2$, and similarly for the constraint $x = y^{-1}$. 

We can now prove the main algebraic lemma that we use to prove perfect zero knowledge.
\begin{lemma}\label{lem:algpzk}
	Given a BCS $B = (X,\{(V_i,C_i)\}_{i=1}^m)$, let $\Tab(B) =
	(Y,\{(W_i,D_i)\}_{i=1}^m)$, and let $\Tab_{sub}(B) = (Y,\{(W_{ij},D_{ij})\}_{i\in [m], j\in
		[m_i]})$. Let $R_i = \{r_i(j,k) : (j,k) \in [3] \times
	[d_i-1]\}$ be the set of randomizers in $W_i$. Then:
	\begin{enumerate}[(a)]
        \item Suppose $(W_{ij},D_{ij})$ is a constraint from $\Tab_{sub}(B)$ of type
            (1), (2), or (4) in \Cref{def:Tab}. If $y$ is a polynomial in $W_{ij}$,
            then $y$ is equal in $\mcA(W_i,D_i)$ to a polynomial in $S \cup R_i$,
            where $W_{ij}\cap V_i \subseteq S \subseteq V_i$ and $|S|\leq 2$.
        \item Suppose $(W_{ij},D_{ij})$ is a constraint from $\Tab_{sub}(B)$ of type
            (3). If $y$ is a polynomial in $W_{ij}$ then $y$ is equal in $\mcA(W_i, D_i)$
            to a polynomial in $V_i \cup R_i$ where every non-scalar monomial contains
            a variable from $R_i$.
        \item If $y$ is a polynomial in $W_{ij}$ and $z$ is a polynomial in
            $W_{ik}$ for some $i \in [m]$, $j,k \in [m_i]$, then $yz$ is equal in
            $\mcA(W_i,D_i)$ to a polynomial in $V_i \cup R_i$ in which every monomial
            either contains a variable from $R_i$ or has degree $\leq 4$.
	\end{enumerate}
\end{lemma}
\begin{proof}
    Fix $i \in [m]$, and consider the permutation-valued variables $T_i(p,q)$
    in $\mcA(W_i,D_i)$. The constraints of type (1) in \Cref{def:Tab} imply that
    $T_i(1,q)$ is a polynomial in $x_q$ for all $q\in [d_i]$. The constraints of
    type (2) along with \Cref{lem:nolinearterm} imply that $T_i(p+1,q)$ is a
    polynomial in $\{r_i(p,q-1)$, $r_i(p,q), T_i(p,q)\}$, and vice versa $T_i(p,q)$ is a polynomial in $\{r_i(p,q-1)$, $r_i(p,q), T_i(p+1,q)\}$. Recall that $r_i(p,0) = r_i(p,d_i) = 1$; for
    notational convenience we use the convention that they are present in every
    monomial, although note they aren't elements of $R_i$. It follows that $T_i(p,q)$ is a polynomial in $\{x_q\}\cup \{r_i(p',q-1),
    r_i(p',q): 1 \leq p'< p\}$, and also a polynomial in $\{T_i(4,q)\} \cup
    \{r_i(p',q-1), r_i(p',q) : p \leq p' \leq 3\}$.  Finally, the constraint of type (3) implies that for
    any $q\in [d_i]$, the variable $T_i(4,q)$ is a polynomial in
    $\{T_i(4,q'):q'\neq q\}$. 

    For part (a), suppose that $y$ is a polynomial in $W_{ij}$.  By the
    previous paragraph, if $(W_{ij}, D_{ij})$ is a constraint of type (1), then
    $y$ can be written as a polynomial in $x_q$, where $\{x_q\} = W_{ij} \cap V_i$.
    If $(W_{ij}, D_{ij})$ is a constraint of type (2) then $y$ can be written
    as a polynomial in $\{x_q\}\cup R_i$ for some $q\in [d_1]$ (and $W_{ij}
    \cap V_i = \emptyset$). If $(W_{ij}, D_{ij})$ is a constraint of type (4)
    then $W_{ij}$ has size two, and $y$ can be written as a polynomial in
    $\{x_q,x_{q'}\}\cup R_i$ for some $q,q'\in [d_i]$, where $W_{ij} \cap V_i
    \subseteq \{x_q,x_{q'}\}$. This finishes the proof of part (a).

    For part (b), if $(W_{ij}, D_{ij})$ has type (3), then we can write $y$ as
    a polynomial in $\{T_i(4,q) : q \in [d_i-1]\}$. Suppose $M =
    T_{i}(4,q_1)^{a_1} \cdots T_{i}(4,q_k)^{a_k}$ is a monomial in this latter
    set of variables, where $k \geq 1$, $1 \leq q_1 < \ldots < q_k < d_i$, and
    $0 \leq a_1,\ldots,a_k < 120$. By \Cref{lem:nolinearterm},
    $T_i(4,q_j)^{a_j}$ is a polynomial in $\{x_{q_j}\} \cup
    \{r_i(p',q_j-1),r_i(p',q_j) : p' \in [3]\}$ such that every monomial
    contains all the randomizers.  When we multiply these polynomials together
    to get the monomial $M$, some of these randomizers may cancel out. However the
    randomizers $r_i(p',q_k)$ for $p'\in [3]$ appear only in the polynomial for
    $T_i(4,q_k)$.  As a result, $M$ is a polynomial in $V_i \cup R_i$ such that
    every monomial contains $r_i(p',q_k)$ for all $p' \in [3]$. We conclude that
    $y$ can be written as a sum of monomials in $V_i\cup R_i$, such that each
    non-scalar monomial contains the randomizers $\{r_i(p',q) : p' \in [3]\}$
    for some $q \in [d_i-1]$. In particular, every non-scalar monomial contains
    some randomizer, finishing the proof of (b).
    
    For part (c), suppose $y$ and $z$ are polynomials in $W_{ij}$ and $W_{ik}$ respectively.
    By part (a), if $(W_{ij},D_{ij})$ and $(W_{ik},D_{ik})$ are constraints of
    type (1), (2) or (4) then $y$ and $z$ both have $V_i$-degree less than or equal
    to two, and thus $yz$ has $V_i$-degree less than or equal to four.  Suppose
    without loss of generality that $(W_{ij},D_{ij})$ is the constraint of type
    (3).  If $(W_{ik}, D_{ik})$ is the same constraint, then $yz$ is a polynomial
    in $W_{ij}$, and is covered by part (b).
    
    Suppose $(W_{ik},D_{ik})$ has type (2), so $W_{ik} = \{r_i(p,q-1),
    r_i(p,q), T_i(p,q), T_i(p+1,q)\}$ for some $p\in [3]$, $q\in [d_i]$. If $p \in
    [2]$, then $z$ is a polynomial in $\{x_q\}\cup\{r_i(p',q-1), r_i(p',q) : 1 \leq
    p' \leq p \}$.  Since $y$ can be written as a polynomial in $V_i\cup R_i$ such
    that every non-scalar monomial contains $r_i(3,q)$ for some $q\in [d_i-1]$,
    $yz$ can be written as a polynomial in $V_i\cup R_i$ such that every monomial
    either has $V_i$-degree at most one or contains $r_i(3,q)$ some $q\in [d_i-1]$.
    If $p = 3$, then $z$ can be written as a polynomial in $\{T_i(4,q), r_i(3,q-1),
    r_i(3,q)\}$ for some $q \in [d_i]$. For any $0\leq a<120$, $T_i(4,q)^ay$ can be
    written as a polynomial in $V_i\cup R_i$ such that every non-scalar monomial
    contains the randomizers $r_i(1,q')$, $r_i(2,q')$ for some $q' \in
    [d_i-1]$. So $yz$ is a polynomial in $V_i\cup R_i$ such that every monomial
    either has $V_i$-degree zero or contains $r_i(1,q')$, $r_i(2,q')$ for some
    $q'\in [d_i-1]$. 
    
    Next suppose $(W_{ik},D_{ik})$ has type (4), and let $F_i = \{T_i(4,q):
    q\in [d_i]\}$.  For $q \in[d_i]$, $T_i(1,q)$ can be written as a polynomial in
    $x_q$, $T_i(2,q)$ can be written as a polynomial in $\{x_q, r_i(1,q-1),
    r_i(1,q)\}$, and $T_i(3,q)$ can be written as a polynomial in $\{T_i(4,q),
    r_i(3,q-1),r_i(3,q)\}$.  Hence every element $W_i$ can be written as a
    polynomial in $V_i\cup R_i\cup F_i$ of $V_i$-degree at most one such that no
    monomial contains $r_i(p,q)$, $r_i(p',q)$ for some $q\in [d_i-1]$, and $p\neq
    p'$.  Thus $z$ can be written as a polynomial $V_i\cup R_i\cup F_i$ with
    $V_i$-degree at most two, and such that for all $q\in [d_i-1]$, no monomial
    contains all the randomizers $\{r_i(p,q) : p \in [3]\}$. If $M$ is any monomial
    in $F_i$ then $My$ can be written as a polynomial in $V_i\cup R_i$ such that
    every non-scalar monomial contains $\{r_i(p,q): p\in [3]\}$ for some $q\in
    [d_i-1]$.  Hence $yz$ can be written as a polynomial in $V_i\cup R_i$ where
    every monomial either has $V_i$-degree at most two or contains a variable from
    $R_i$.
        
    Finally if $(W_{ik},D_{ik})$ has type (1), then $z$ is a polynomial in
    $x_q$ for some $q$, and as in the previous paragraph, $yz$ can be written as a
    polynomial in $V_i\cup R_i$ where every monomial either has $V_i$-degree at
    most two or contains a variable from $R_i$. We conclude that part (c) holds.
\end{proof}
\begin{rmk}\label{rmk:prooffails}
    Note that the proof of \Cref{lem:algpzk}, part (c) fails if we use a
    three-row tableau in \Cref{def:Tab} rather than a four-row tableau. Indeed,
    suppose we used three-row tableaux. If $(W_{ij}, D_{ij})$ is the constraint
    of type (3), and $(W_{ik},D_{ik})$ is the constraint of type (4) with
    $W_{ik} = \{r_{i}(1,q),r_i(2,q)\}$, then it is possible for $yz$ to have
    monomials of degree $\geq 5$ that do not contain any randomizers. For instance,
    when $q = 5$, we can take $y = T_i(3,1) \cdots T_i(3,5)$. This corresponds to
    the fact that, with three-row tableaux, we can recover the group product 
    $T_i(1,1) \cdots T_i(1,q)$ from the variables $T_i(3,q')$, $q' \in [q]$
    and the randomizers $r_i(1,q)$, $r_i(2,q)$. 
\end{rmk}

Combining \Cref{lem:algpzk} with \Cref{lem:obliviate}, for any BCS $B$ we can
define a perfect correlation $p$ for the BCS game $\mcG(\Tab_{sub}(\Obl_5(B)))$
such that $p$ is a quantum correlation if and only if $\mcG(B)$ has a perfect
quantum strategy. 
\begin{prop}\label{prop:pzkcordef}
    Suppose $B = (X,\{(V_i,C_i)\}_{i=1}^m)$ is a BCS with $m$ constraints, and
    $\pi$ is a probability distribution on $[m] \times [m]$ such that $\pi(i,j)
    > 0$ for all $i,j \in [m]$. Let $\Obl_5(B) =
    (Z,\{(U_i,E_i)\}_{i=1}^m)$, $\Tab(\Obl_5(B)) =
    (Y,\{(W_i,D_i)\}_{i=1}^m)$, and $\Tab_{sub}(\Obl_5(B)) =
    (Y,\{(W_{ij},D_{ij})\}_{i\in [m], j\in [m_i]})$. 
    Let $R_i = \{r_i(j,k) : (j,k) \in [3] \times [d_i-1]\}$ be the set of
    randomizers in $W_i$.  For any $i \in [m]$ and $n
    \geq 1$, let $\Lambda_{i,n}$ be the set of non-scalar monomials over $U_i
    \cup R_i$ which either contain an element of $R_i$, or have degree at most
    $n$.  Let $\Lambda$ be the subspace of $\mcA(\Tab(\Obl_5(B)))$ defined by 
    \begin{equation*}
        \Lambda = \C 1\ \oplus\ \vspan \bigcup_{i \in [m]}
            \sigma_i(\Lambda_{i,4})\ \oplus\ \vspan \bigcup_{i \neq j \in [m]}
            \sigma_i(\Lambda_{i,2}) \sigma_j(\Lambda_{j,2}),
    \end{equation*}
    and let $f : \Lambda \to \C$ be the linear functional defined by $f(1) =
    1$, $f(\sigma_i(x)) = 0$ for all $x \in \Lambda_{i,4}$, and $f(\sigma_i(x)
    \sigma_j(y)) = \delta_{xy}$ for all $x \in \Lambda_{i,2}$, $y \in
    \Lambda_{j,2}$, where $\delta_{ab}$ is the Kronecker delta, i.e.
    $\delta_{ab} = 1$ if $a = b$, and is $0$ otherwise. 
    Let $\alpha : \mcA(\Tab_{sub}(\Obl_5(B))) \to \mcA(\Tab(\Obl_5(B)))$ be the
    homomorphism sending $\sigma_{ij}(x) \mapsto \sigma_i(x)$ for all $x \in
    \mcA(W_{ij},D_{ij})$, as in \Cref{prop:subdiviso}. 

    For every $i,k \in [m], j \in [m_i], l \in [m_k]$ and assignments $\phi$ and
    $\psi$ to $W_{ij}$ and $W_{kl}$ respectively, let
    \begin{equation*}
        p(\phi,\psi| ij,kl) = f(\alpha(\Phi_{W_{ij},\phi} \Phi_{W_{kl},\psi})).
    \end{equation*}
    Then $p$ is a perfect correlation for the BCS game
    $\mcG(\Tab_{sub}(\Obl_5(B)),\pi_{sub})$, and $p \in C_{q}$ (resp.
    $C_{qa}$, $C_{qc}$) if and only if $\mcG(B, \pi)$ has a perfect quantum
    correlation in $C_{q}$ (resp. $C_{qa}$, $C_{qc}$). 
\end{prop}
\begin{proof}
    We first observe that the linear functional $f$ is well-defined, by showing
    that it can be defined on a larger subspace. Indeed, for any set of
    variables $S$, let $\mcM(S)$ be the set of non-scalar monomials in $S$, and
    let $\mcM_n(S) \subseteq \mcM(S)$ be the subset of monomials of degree at
    most $n$. Since we're assuming that $(V_i,C_i)$ has at least one satisfying
    assignment, $\mcA(V_i,C_i)$ has a tracial state. Applying part (b) of
    \Cref{lem:obliviate} to the constraint system containing the single
    constraint $(V_i,C_i)$, we see that $\mcA(U_i,E_i)$ has a tracial state
    $\tau_i$ such that $\tau_i(x) = 0$ for all $x \in \mcM_4(U_i)$. Hence $\C 1
    \cap \vspan \mcM_4(U_i) = \{0\}$ in $\mcA(U_i,E_i)$. By
    \Cref{lem:obliviate}, part (c), the set $\mcM_2(U_i)$ is linearly
    independent in $\mcA(U_i,E_i)$.  Hence we can choose a basis $\Xi_i$ for
    $\mcA(U_i,E_i)$ which contains $\{1\} \cup \mcM_2(U_i)$, and such that
    $\vspan \mcM_4(U_i) \subseteq \vspan \Xi_i \setminus \{1\}$. By
    \Cref{lem:RT}, part (a), the set $\{ab : a \in \Xi_i, b \in \mcM(R_i)\}$ is
    a basis for $\mcA(W_i,D_i)$. Let $\Theta_i$ be the set of non-identity
    elements in this basis. Because $\mcA(\Tab(B))$ is a free product of the
    algebras $\mcA(W_i,D_i)$, the set
    \begin{equation*}
        \Theta := \{1\} \cup \bigcup_{i \in [m]} \sigma_i(\Theta_i) \cup \bigcup_{i \neq j \in [m]}  
            \sigma_i(\Theta_i) \sigma_j(\Theta_j)
    \end{equation*}
    is linearly independent in $\mcA(\Tab(B))$. Define a linear functional $f$
    on the span of $\Theta$ by setting $f(1) = 1$, $f(\sigma_i(x)) = 0$ for all
    $x \in \Theta_i$, and 
    \begin{equation*}
        f(\sigma_i(x)\sigma_j(y)) = \begin{cases} 1 & x \text{ and } y \text{ are both in } \mcM_2(U_i \cap U_j)
            \text{ and } x=y \\
                                                  0 & \text{otherwise}
                                    \end{cases}
    \end{equation*}
    for all $x \in \Theta_i$,
    $y \in \Theta_j$ with $i \neq j$. The image of the set $\Lambda_{i,4}$ in
    $\mcA(W_i,D_i)$ is contained in the span of $\Theta_i$, so the span of $\Theta$
    contains the subspace $\Lambda$. Furthermore, if $x \in \Lambda_{i,4}$, then 
    $f(\sigma_i(x)) = 0$. Suppose $x \in \Lambda_{i,2}$ and $y \in \Lambda_{j,2}$ with
    $i \neq j$. If $x$ contains an element of $R_i$, then $x$ is contained in the 
    span of $\{ab : a \in \mcM(U_i), b \in \mcM(R_i), b \neq 1\}$, and $f(\sigma_i(x)
    \sigma_j(y)) = 0 = \delta_{xy}$. The same is true if $y$ contains an element of
    $R_j$. If neither $x$ or $y$ contains an element of $R_i$ or $R_j$ respectively,
    then $x \in \mcM_2(U_i)$ and $y \in \mcM_2(U_j)$ are elements of $\Theta_i$ and
    $\Theta_j$ respectively. The only way for $x$ and $y$ to be equal is if both
    belong to $\mcM_2(U_i \cap U_j)$, so $f(\sigma_i(x) \sigma_j(y)) = \delta_{xy}$. 
    Thus the restriction of $f$ to $\Lambda$ is the linear functional defined in the
    proposition.
    
    Since $f$ is well defined, \Cref{lem:algpzk} implies that $p$ is well
    defined. Since $\sum_{\phi}\Phi_{W_{ij}\phi} = 1$, it follows that
    $\sum_{\phi,\psi}p(\phi,\psi|ij,kl) = 1$ for every $i,k \in [m], j \in [m_i]$
    and $l \in [m_k]$.  To show that $p$ is a perfect correlation for
    $\mcG(\Tab_{sub}(\Obl_5(B)))$, we need to show that $p(\phi,\psi|ij,kl)\geq 0$
    for all $\phi \in D_{ij}$, $\psi\in D_{kl}$, $i,k \in [m], j \in [m_i]$ and $l
    \in [m_k]$, and that $p(\phi,\psi|ij,kl) = 0$ if $\phi|_{W_{ij}\cap W_{kl}}
    \neq \psi|_{W_{ij}\cap W_{kl}}$.  If $i = k$, then $\alpha(\Phi_{W_{ij},\phi})$
    and $\alpha(\Phi_{W_{kl},\psi})$ are both projections in the commutative
    algebra $\mcA(W_i,D_i)$, and thus their product is also a projection. Since
    $C_i$ is non-empty by assumption, $\mcA(V_i,C_i)$ has a tracial state. If $B_i
    = (V_i,\{(V_i,C_i)\})$ is the constraint system for the single constraint
    $C_i$, then $\Obl_5(B_i) = (U_i,\{(U_i,E_i)\})$ and $\Tab(\Obl_5(B_i)) =
    (W_i,\{(W_i,D_i)\})$. By \Cref{lem:obliviate}, part (b), there is a tracial
    state $\tau_i$ on $\mcA(U_i,E_i)$ such that $\tau_i(x) = \delta_{x,1}$ for all
    $x\in \mcM_4(U_i)$. By \Cref{lem:RT}, part (d), there is a tracial state
    $\wtd{\tau}_i$ on $\mcA(W_i,D_i)$ such that $\wtd{\tau}_i(x) = \tau_i(x)$ for
    all $x\in \mcM(U_i)$, and $\wtd{\tau}_i(x) = 0$ for all monomials $x \in
    \mcM(U_i\cup R_i)$ containing an element of $R_i$. Since $\wtd{\tau}_i(1) = 1$
    and $\wtd{\tau}_i(x) = 0$ for all $x \in \Lambda_{i,4}$, the linear functionals
    $f$ and $\wtd{\tau}_i$ agree on $\C1\oplus \Lambda_{i,4}$, and
    $f(\alpha(\Phi_{W_{ij},\phi} \Phi_{W_{kl},\psi})) =
    \wtd{\tau}_i(\alpha(\Phi_{W_{ij},\phi} \Phi_{W_{kl},\psi}))\geq 0$.  If
    $\phi|_{W_{ij}\cap W_{kl}}\neq \psi|_{W_{ij}\cap W_{kl}}$ then
    $\alpha(\Phi_{W_{ij},\phi})\alpha(\Phi_{W_{kl},\psi}) = 0$ in $\mcA(W_i,D_i)$,
    and $f(\alpha(\Phi_{W_{ij},\phi} \Phi_{W_{kl},\psi})) = 0$.
        
    If $i\neq k$ and neither $(W_{ij},D_{ij})$ or $(W_{kl},D_{kl})$ are
    constraints of type (3) in \Cref{def:Tab}, then by \Cref{lem:algpzk} there
    exist $S_i\subseteq U_i$ and $S_k\subseteq U_k$ of size at most two, such that
    $W_{ij}\cap U_i\subseteq S_i$, $W_{kl}\cap U_k\subseteq S_k$,
    $\Phi_{W_{ij},\phi}$ is a polynomial in $S_i\cup R_i$, and $\Phi_{W_{kl},\psi}$
    is a polynomial in $S_k \cup R_k$.  Since $\mcM_2(S_i)$ is a linearly
    independent in $\mcA(U_i, E_i)$, part (a) of \Cref{lem:RT} implies that the
    subalgebra of $\mcA(W_i,D_i)$ generated by $S_i\cup R_i$ is isomorphic to
    $\C\Z_2^{S_i}\times\Z_{120}^{R_i}$, and similarly with the algebra generated by
    $S_k\cup R_k$ in $\mcA(W_k,D_k)$. Hence the subalgebra $\mcC$ of
    $A(\Tab(\Obl_5(B)))$ generated by $S_i\cup S_k\cup R_i\cup R_k$ is isomorphic
    to the group algebra of $(\Z_2^{S_i}\times\Z_{120}^{R_i}) \ast
    (\Z_2^{S_k}\times\Z_{120}^{R_k})$.  Let $H$ be the quotient of this free
    product group by the relations $\sigma_i(x)=\sigma_k(x)$ for all $x\in
    S_i\cap S_k$, where $\sigma_i(x)$ and $\sigma_k(x)$ are the group
    generators corresponding to $x$ in the first and second factors of the free
    product respectively, and let 
    \begin{equation*}
    	q: (\Z_2^{S_i}\times\Z_{120}^{R_i}) \ast (\Z_2^{S_k}\times\Z_{120}^{R_k}) \rightarrow H
    \end{equation*}
    be the quotient map. Observe that 
    \begin{equation*}
        H = \Z_2^{\ast S_i\cup S_k}\ast \Z_{120}^{\ast R_i\cup R_k} / \langle
        xy = yx \text{ for } x,y \text{ in } S_i\cup R_i \text{ or } S_k\cup R_k \rangle
    \end{equation*}
    is a graph product. By the normal form theorem for graph products
    \cite{green1990}, if $g \in \mcM(S_i\cup R_i)$ and $h\in \mcM(S_j\cup R_j)$,
    then $q(gh) = 1$ if and only if $g,h\in \mcM(S_i\cap S_j)$ and $g=h$. Hence if
    $\tau$ is the canonical trace on the group algebra $\C H$, then $\tau\circ
    q(\sigma_i(g)\sigma_k(h))= f(\sigma_i(g)\sigma_k(h))$. We conclude that
    $f(\alpha(\Phi_{W_{ij},\phi}\Phi_{W_{kl},\psi})) = \tau\circ
    q(\alpha(\Phi_{W_{ij},\phi}\Phi_{W_{kl},\psi}))$. Since $\Phi_{W_{ij},\phi}$
    and $\Phi_{W_{kl},\psi}$ are projections, $\tau\circ
    q(\alpha(\Phi_{W_{ij},\phi}\Phi_{W_{kl},\psi}))\geq 0$. Suppose
    $\phi(x)\neq\psi(x)$ for some $x\in W_{ij}\cap W_{kl}$. Then we must have $x\in
    U_i\cap U_k$, so $x\in S_i\cap S_k$. Since
    $\frac{1+\phi(x)x}{2}\Phi_{W_{ij},\phi} = \Phi_{W_{ij},\phi}$ and
    $\frac{1+\psi(x)x}{2}\Phi_{W_{kl},\psi} = \Phi_{W_{kl},\psi}$, we have 
    \begin{align*}
        q(\alpha(\Phi_{W_{ij},\phi}\Phi_{W_{kl},\psi})) &=
            q\left(\alpha(\Phi_{W_{ij},\phi}) \sigma_i\left(\frac{1+\phi(x)x}{2}\right) 
            \sigma_k\left(\frac{1+\psi(x)x}{2}\right) \alpha(\Phi_{W_{kl},\psi})\right)\\
            &=q(\alpha(\Phi_{W_{ij},\phi}))\left(\frac{1+\phi(x)x}{2}\right)\left(\frac{1+\psi(x)x}{2}\right)
                q(\alpha(\Phi_{W_{kl},\psi})) = 0.
    \end{align*}
    Thus $f(\alpha(\Phi_{W_{ij},\phi}\Phi_{W_{kl},\psi})) = 0$ if
    $\phi|_{W_{ij}\cap W_{kl}}\neq \psi|_{W_{ij}\cap W_{kl}}$.

    Finally, suppose $i \neq k$ and $(W_{ij},D_{ij})$ is a constraint of type
    (3). By \Cref{lem:algpzk}, part (b), we can write $\alpha(\Phi_{W_{ij},\phi}) =
    \lambda 1 + \sum_{x} c_x x$ for some coefficients $\lambda$, $c_x \in \C$,
    where the sum is over monomials $x \in \mcM(U_i \cup R_i)$ containing an
    element of $R_i$. Let $\wtd{\tau}_i$ and $\wtd{\tau}_k$ be the tracial
    states on $\mcA(W_i,D_i)$ and $\mcA(W_k,D_k)$ defined above. Since
    $\wtd{\tau}_i$ is equal to $f$ on $\Lambda_{i,4}$ and $\Phi_{W_{ij},\phi}$
    is a projection, $\lambda = f(\alpha(\Phi_{W_{ij},\phi})) =
    \wtd{\tau}_i(\alpha(\Phi_{W_{ij},\phi})) \geq 0$. Similarly,
    $f(\alpha(\Phi_{W_{kl},\psi})) = \wtd{\tau}_k(\alpha(\Phi_{W_{kl},\psi})) \geq 0$.
    If $x \in \mcM(U_i \cup R_i)$ contains an element of $R_i$,
    then $x \alpha(\Phi_{W_{kl},\psi}) \in \sigma_i(\Lambda_{i,4}) \oplus
    \sigma_i(\Lambda_{i,2}) \sigma_k(\Lambda_{k,2})$, so $f(x
    \alpha(\Phi_{W_{kl},\psi})) = 0$.  We conclude that
    $f(\alpha(\Phi_{W_{ij},\phi} \Phi_{W_{kl},\psi})) = \lambda
    f(\alpha(\Phi_{W_{kl},\psi})) \geq 0$. It is not possible to have
    $\phi|_{W_{ij} \cap W_{kl}} \neq \psi|_{W_{ij} \cap W_{kl}}$ when
    $i \neq k$ and $(W_{ij},D_{ij})$ has type (3), since $W_{ij} \cap W_{kl} =
    \emptyset$. 

    This finishes the proof that $p$ is a perfect correlation for
    $\mcG(\Tab_{sub}(\Obl_5(B)), \pi_{sub})$. If $p \in C_{qc}$ (resp.
    $C_{qa}$, $C_q$), then $\mcG(\Tab(\Obl_5(B)), \pi)$ also has a perfect
    strategy in $C_{qc}$ (resp. $C_{qa}$, $C_q$) by \Cref{prop:subdiviso}. 
    This means that there is a tracial state (resp. Connes-embeddable tracial
    state, finite-dimensional tracial state)
    $\wtd{\tau}$ on $\mcA(\Tab(\Obl_5(B)))$ with $df(\tau;\mu_\pi)=0$. 
    By \Cref{lem:RT}, part (b), and \Cref{lem:obliviate}, part (a), there is a
    $1$-homomorphism $\mcA(B) \to \mcA(\Tab(\Obl_5(B)))$, and pulling back
    $\wtd{\tau}$ by this $1$-homomorphism yields a perfect strategy for $\mcG(B,\pi)$
    in $C_{qc}$ (resp. $C_{qa}$, $C_q$). Conversely, if $\mcG(B,\pi)$ has a 
    perfect strategy in $C_{qc}$, then there is a
    tracial state $\tau$ on $\mcA(B)$ with $\df(\tau;\mu_\pi)=0$. By
    \Cref{lem:obliviate}, part (b), there is a tracial state $\tau'$ on
    $\mcA(\Obl_5(B))$ such that $\df(\tau'; \mu_{\pi}) = 0$,
    $\tau'(\sigma_i(x)) = 0$ for all $x \in \mcM_4(U_i) \setminus \{1\}$, $i
    \in [m]$, and $\tau'(\sigma_i(x) \sigma_j(y))=0$ for all $x \in
    \mcM_2(U_i)$, $y \in \mcM_2(U_j)$, $i \neq j \in [m]$ with $x \neq y$. By
    \Cref{lem:RT}, part (d), there is a tracial state $\wtd{\tau}$ on
    $\mcA(\Tab(\Obl_5(B)))$ with $\df(\wtd{\tau}; \mu_{\pi}) = 0$,
    $\wtd{\tau}(u) = \tau'(u)$ for all monomials $u$ in $\{\sigma_i(x) : 
    i \in [m], x \in U_i\}$, and $\wtd{\tau} (u \sigma_i(z)^a v) = 0$ for all $i \in
    [m]$, $z \in R_i$, $1 \leq a < 120$, and monomials $u,v$ in $\{\sigma_j(x)
    : j \in [m], x \in U_j \cup R_j\}$ which do not contain $z$.
    Observe that
    $\wtd{\tau}(1) = 1$, and $\wtd{\tau}(\sigma_i(x)) = 0$ for all $x \in
    \Lambda_{i,4}$.  Similarly, if $x \in \Lambda_{i,2}$ and $y \in
    \Lambda_{j,2}$ are not equal, then $\wtd{\tau}(\sigma_i(x) \sigma_j(y)) =
    0$. By \Cref{prop:inter}, $\df(\wtd{\tau};\mu_{inter}) = 0$, and since
    $\pi(i,j) > 0$ for all $i,j \in [m]$, $\| \sigma_i(x) -
    \sigma_j(x)\|_{\wtd{\tau}} = 0$ for all $x \in U_i \cap U_j$.  Since
    $\|\cdot\|_{\wtd{\tau}}$ is unitarily bi-invariant, we get that
    $\|\sigma_i(x) - \sigma_j(x)\|_{\wtd{\tau}} = 0$ for all $x \in \mcM(U_i
    \cap U_j)$, and hence $\wtd{\tau}(\sigma_i(x)\sigma_j(x)) = 1$ for all $x
    \in \mcM(U_i \cap U_j)$. It follows that $\wtd{\tau}|_{\Lambda} = f$, so
    $p(\phi,\psi|ij,kl) = \wtd{\tau} \circ \alpha (\Phi_{W_{ij},\phi}
    \Phi_{W_{kl},\psi})$ for all $\phi \in D_{ij}$, $\psi \in D_{kl}$,
    $i,k \in [m]$, $j \in [m_i]$, $l \in [m_k]$. We conclude that $p \in C_{qc}$. 
    If $\mcA(B)$ has a perfect strategy in $C_{qa}$ (resp. $C_q$), then
    we can take $\tau$ to be Connes-embeddable (resp. finite-dimensional),
    so $\wtd{\tau}$ will be Connes-embeddable (resp. finite-dimensional),
    and $p \in C_{qa}$ (resp. $C_{q}$).
\end{proof}

\begin{rmk}\label{rmk:operational}
    The correlation $p$ in \Cref{prop:pzkcordef} is described algebraically.
    Alternatively, it's not hard to see that the correlation $p(\phi,\psi|ij,kl)$
    can be simulated using the following procedure: If neither
    $(W_{ij},D_{ij})$ or $(W_{kl},D_{kl})$ has type (3), then pick an
    assignment to the variables $Z \cup R$ uniformly at random, and fill in the
    variables $T_i(p,q)$ so that the constraints $(W_{ij},D_{ij})$ of types
    (1), (2), and (4) are satisfied, to get an assignment $\gamma$ to $Y$. 
    Output $\phi = \gamma|_{W_{ij}}$ and $\psi = \gamma|_{W_{kl}}$. If
    one of $(W_{ij},D_{ij})$ or $(W_{kl},D_{kl})$ has type (3), then for each
    $r \in [m]$, pick an assignment to $R_r$ and a satisfying assignments to
    $(U_r,E_r)$ uniformly at random, and fill in the variables $T_j(p,q)$ to
    get a satisfying assignment $\gamma_r$ to $(W_r,D_r)$. Output $\phi =
    \gamma_i|_{W_{ij}}$ and $\psi = \gamma_k|_{W_{kl}}$. This procedure will
    output $\phi,\psi$ with probability $p(\phi,\psi|ij,kl)$. 
\end{rmk}
The description of the correlation $p$ in \Cref{rmk:operational} is simpler
than the algebraic description. On the other hand, without the algebraic
description, it's harder to see that the correlation generated in
\Cref{rmk:operational} is quantum when $\mcA(B)$ has a perfect quantum
strategy. In fact, if we use three-row tableaux rather than four-row tableaux
in \Cref{def:Tab}, then the procedure in \Cref{rmk:operational} is still
well-defined, and simulates a perfect strategy for the game. However, by
\Cref{rmk:prooffails}, the simulated correlation is not necessarily quantum
even if $\mcA(B)$ has a perfect quantum strategy --- something that is not
immediately apparent from the description of the procedure. 

We are now ready to prove our main result about constructing perfect zero
knowledge protocols:
\begin{thm}\label{thm:makePZK}
    Let $(\{\mcG(B_x,\pi_x)\}, S, C)$ be a $\BCS$-$\MIP^*$ protocol for a
    language $\mcL$ with completeness $1$ and soundness $1-f(x)$, such that
    each context of $B_x$ has constant size, and $\pi_x$ is maximized on the
    diagonal. Then there is a $\PZK$-$\BCS$-$\MIP^*$ protocol
    $(\{\mcG(B_x',\pi_x')\}, \wtd{S}, \wtd{C})$ for $\mcL$ with completeness
    $1$ and soundness $1-f(x)/\poly(m_x)$, where $m_x$ is the number of
    contexts in $B_x$. If $\pi_x$ is uniform, then
    $\pi_x'$ is also uniform, and if $\mcG(B_x,\pi_x)$ has a perfect
    finite-dimensional tracial state, then so does $\mcG(B_x',\pi_x')$.
\end{thm}
\begin{proof}
    Let $B_x' = \Tab_{sub}(\Obl_5(B_x))$, and let $\pi_x'$ be the subdivision
    of $\pi_x$ corresponding to the subdivision of $\Tab(\Obl_5(B_x))$ into
    $\Tab_{sub}(\Obl_5(B_x))$. If $\pi_x$ is uniform, then $\pi_x'$ is also
    uniform. Let $p_x$ be the correlation for $\mcG(B_x',\pi_x')$ defined in
    \Cref{prop:pzkcordef}. Because $B_x$ has contexts of constant size, $\Obl_5(B_x)$
    and $\Tab(\Obl_5(B_x))$ also have contexts of constant size. As a result,
    the number of clauses in the constraints of $\Tab(\Obl_5(B_x))$ is
    constant, as is the size of each clause (where by clause we mean the
    constraints of type (1)-(4) in \Cref{def:Tab}). Hence the Turing machines
    $S$ and $C$ can be turned into Turing machines $\wtd{S}$ and $\wtd{C}$ such
    that $(\{\mcG(B_x',\pi_x')\}, \wtd{S}, \wtd{C})$ is a $\BCS$-$\MIP^*$ protocol.
    Similarly, since all the constraints of $\Tab(\Obl_5(B))$ have constant size,
    there is a Turing machine which, given questions and answers $i,j,\phi,\psi$
    for $\mcG(B_x',\pi_x')$, can produce $p_x(\phi,\psi|i,j)$ in polynomial time
    in $i$, $j$, and $x$. Since the number of answers for any question is constant, the
    correlation $p_x$ can be simulated in polynomial time in $x$.

    If $x \in \mcL$, then $B_x$ has a perfect strategy in $C_{qa}$, so $p_x \in C_{qa}$,
    and hence $\mcG(B_x',\pi_x')$ has a perfect strategy in $C_{qa}$. Similarly, if
    $B_x$ has a perfect strategy in $C_q$, then $\mcG(B_x',\pi_x')$ has a perfect strategy
    in $C_q$ as well. Conversely, suppose
    that $\tau$ is a tracial state on $\mcA(B_x')$. Since the size of contexts and
    number of clauses in each constraint of $\Tab(\Obl_5(B_x))$ are constant, 
    the parameters $C$, $M$, and $K$ in \Cref{Thm:subdivSound} when going from
    $\Tab(\Obl_5(B_x))$ to $\Tab_{sub}(\Obl_5(B_x))$ are all constant. Since
    $\Tab(\Obl_5(B_x))$ has $m_x$ contexts, \Cref{Thm:subdivSound} implies that
    there is a tracial state $\tau_0$ on $\mcA(\Tab(\Obl_5(B_x)))$ with
    $\df(\tau_0) \leq \poly(m_x) \df(\tau)$. Since there is a classical
    homomorphism $\mcA(B_x) \to \mcA(\Tab(\Obl_5(B_x)))$ by Lemmas \ref{lem:obliviate}
    and \ref{lem:RT}, we conclude that there is a tracial state $\tau_1$ on
    $\mcA(B_x)$ with $\df(\tau_1) \leq \poly(m_x) \df(\tau)$. Hence if
    $x \not\in \mcL$, then there is no synchronous strategy $p$ for
    $\mcG(B_x',\pi_x')$ with $\omega_q(\mcG(B_x',\pi_x'), p) \geq 1 -
    f(n)/\poly(m_x)$. Hence $(\{\mcG(B_x',\pi_x')\}, \wtd{S}, \wtd{C})$ is a
    $\BCS$-$\MIP^*$ protocol for $\mcL$ with soundness $1 - f(x) / \poly(m_x)$. 
\end{proof}

\begin{theorem}\label{thm:main1}
    There is a perfect zero knowledge $\BCS$-$\MIP^*(2,1,1,1 - 1/\poly(n))$
    protocol for the halting problem in which the verifier selects questions
    according to the uniform distribution, the questions have length
    $\polylog(n)$, and the answers have constant length. Furthermore, if
    a game in the protocol has a perfect strategy, then it has a perfect
    synchronous quantum strategy.
\end{theorem}
\begin{proof}
    By \Cref{thm:mipre}, there is a $\BCS$-$\MIP^*$ protocol
    $(\{\mcG(B_x,\pi_x)\}, S, V)$ for the halting problem with constant
    soundness $s<1$, in which $B_x$ has a constant number of contexts and
    contexts of size $\polylog(|x|)$, and $\pi_x$ is the uniform distribution
    on pairs of contexts. Furthermore, if $\mcG(B_x,\pi_x)$ has a perfect
    strategy, then it has a perfect synchronous quantum strategy.
    By \Cref{rmk:to3SAT}, $(\{\mcG(B_x,\pi_x)\}, S, C)$ 
    can be turned into a $\BCS$-$\MIP^*$ protocol $(\{\mcG(B'_x,\pi_x)\},S,C)$
    where $B_x' = (X_x', \{(W_i^x,D_i^x)\})$, $D_i$ is a 3SAT instance with
    number of clauses polynomial in $|x|$, and $|W^x_i|$ is polynomial in
    $|x|$. Then by subdividing the $B_x'$ into a 3SAT we obtain a 3SAT protocol
    $(\{\mcG(B^{3SAT}_x,\pi^{3SAT}_x)\},S,C)$ with number of clauses polynomial
    in $|x|$, and $\pi^{3SAT}_x$ is uniform. There is a $1$-homomorphism
    $\mcA(B^{3SAT}_x, \mu_{\pi^{3SAT}_x}) \to \mcA(B_x, \mu_{\pi_x})$,
    so if $\mcG(B_x,\pi_x)$ has a perfect synchronous quantum strategy, so
    does $\mcG(B^{3SAT}_x,\pi^{3SAT}_x)$.
    The theorem follows from \Cref{thm:makePZK}.
\end{proof}

\begin{proof}[Proof of \Cref{thm:main}]
    By the discussion after \Cref{thm:mipre}, it is enough to show that 
    there is a two-prover one-round perfect zero knowledge $\MIP^*$ protocol
    for the halting problem with completeness $1$, soundness $1/2$, and uniform
    probability distribution. 
    Let $(\{\mcG(B_x, \pi_x)\}, S, C)$ be the $\BCS$-$\MIP^*$ protocol from
    \Cref{thm:main1}, so in particular $B_x$ has $m_x$ contexts, where $m_x =
    \poly(|x|)$, and $\pi_x$ is the uniform distribution on $[m_x] \times
    [m_x]$. Since the uniform distribution is $1/2m_x$-diagonally dominant,
    \Cref{thm:synchrounding} implies that $(\{\mcG(B_x, \pi_x)\}, S, C)$ has
    soundness $1-1/\poly(n)$ when considered as a $\MIP^*$ protocol. 
    The result follows from \Cref{thm:prep} using a polynomial amount of parallel repetition.
\end{proof}

As mentioned in the introduction, the proof of Theorem \ref{thm:main} implies 
that the halting problem is many-one reducible to membership in $C_q$. In fact, there
is a reduction such that if the Turing machine does not halt, then the corresponding 
correlation is bounded away from the closure $C_{qa}$ of $C_q$:
\begin{cor}
    There is a polynomial-time computable function $p$ from Turing machines to synchronous
    correlations such that if $M$ halts then $p(M) \in C_q^s$, and if $M$ does not halt then
    there is a linear functional $f$ on the space of correlations such that $f(p(M)) = 1$ and
    $f(p') \leq 1/2$ for all $p' \in C_{qa}$. 
\end{cor}
\begin{proof}
    Let $(\{\mcG(B_M, \pi_M)\}, S, C)$ be the $\BCS$-$\MIP^*$ protocol for the
    halting problem with completeness one and soundness $1/2$ constructed in
    the proof of \Cref{thm:main}, where the index $M$ runs through Turing machines.
    Let $p(M)$ be the correlation for $\mcG(B_M,\pi_M)$ as in \Cref{def:PZK}.
    That $p(M)$ is in $C_q$ follows from \Cref{thm:main1}, and the fact that if
    $p \in C_q$, then $p^{\otimes n} \in C_q$. The corollary then follows with
    the linear functional $f$ defined by $f(p') = \omega(\mcG(B_M,\pi_M),p')$. 
\end{proof}
Note that the number of inputs and outputs for the correlation $p(M)$ depends
on the size of the Turing machine $M$. 

Finally, we also have:
\begin{theorem}\label{thm:main2}
    $\PZK$-$\BCS$-$\MIP^{co}(2,1,1,1-1/\poly(n)) = \BCS$-$\MIP^{co}(2,1,1,1-1/\poly(n))$. 
\end{theorem}
The proof is similar to the proof of \Cref{thm:main1}.

\bibliography{pzk} 
\bibliographystyle{alpha}

\end{document}

%% file: ams_header.tex
\usepackage{amsmath,amsthm,amsfonts,amssymb,latexsym,verbatim,bbm,hyperref,mathtools,mathabx}
\usepackage[shortlabels]{enumitem}
\usepackage{subcaption}
\usepackage{afterpage}
\usepackage{cleveref}

\DeclareRobustCommand{\SkipTocEntry}[5]{}

\allowdisplaybreaks

\setlist{itemsep=.5\baselineskip,topsep=.5\baselineskip}

\numberwithin{equation}{section}
\theoremstyle{plain}
\newtheorem{theorem}{Theorem}[section]
\newtheorem{thm}[theorem]{Theorem}
\newtheorem{lemma}[theorem]{Lemma}
\newtheorem{proposition}[theorem]{Proposition}
\newtheorem{prop}[theorem]{Proposition}

\newtheorem{definition}[theorem]{Definition}
\newtheorem{defn}[theorem]{Definition}
\newtheorem{remark}[theorem]{Remark}
\newtheorem{corollary}[theorem]{Corollary}
\newtheorem{cor}[theorem]{Corollary}

\newtheorem{rmk}[theorem]{Remark}

\newcommand{\iso}{\cong}

\newcommand{\wtd}{\widetilde}
\DeclareMathOperator{\vspan}{span}

\newcommand{\C}{\mathbb{C}}
\newcommand{\Z}{\mathbb{Z}}

\newcommand{\Q}{\mathbb{Q}}

\newcommand{\eps}{\epsilon}

\newcommand{\mcA}{\mathcal{A}}
\newcommand{\mcB}{\mathcal{B}}
\newcommand{\mcC}{\mathcal{C}}

\newcommand{\mcG}{\mathcal{G}}
\newcommand{\mcH}{\mathcal{H}}

\newcommand{\mcK}{\mathcal{K}}
\newcommand{\mcL}{\mathcal{L}}
\newcommand{\mcM}{\mathcal{M}}

\newcommand{\mcO}{\mathcal{O}}
\newcommand{\mcP}{\mathcal{P}}

\newcommand{\mcR}{\mathcal{R}}
\newcommand{\mcS}{\mathcal{S}}
\newcommand{\mcT}{\mathcal{T}}
\newcommand{\mcU}{\mathcal{U}}

\newcommand{\mbC}{\mathbb{C}}

\newcommand{\mbZ}{\mathbb{Z}}
